\documentclass{cccg21}
\usepackage{graphicx,amssymb,amsmath}
\usepackage{subcaption}
\usepackage{xspace}
\usepackage{tabu}

\usepackage{booktabs}
\usepackage{wrapfig}
\usepackage{amsbsy}
\usepackage{placeins}
\usepackage{xcolor}

\newcommand{\mypar}[1]{\smallskip\noindent{\bfseries\boldmath #1.}}
\newcommand{\edge}[2]{\ensuremath{(#1,#2)}}
\newcommand{\xtriangle}[3]{\ensuremath{\triangle #1 #2 #3}}

\graphicspath{{./figures/}}

\title{Near-Delaunay Metrics}

\author{Nathan van Beusekom\thanks{Department of Mathematics and Computer Science, TU Eindhoven, the Netherlands; {\tt n.a.c.v.beusekom@tue.nl, k.a.buchin@tue.nl, h.o.koerts@student.tue.nl, w.meulemans@tue.nl, b.speckmann@tue.nl.}} \qquad Kevin Buchin$^*$ \qquad Hidde Koerts$^*$\\ Wouter Meulemans$^*$ \qquad Benjamin Rodatz\thanks{Department of Computer Science, University of Oxford; \newline {\tt benjamin.rodatz@cs.ox.ac.uk.}} \qquad  Bettina Speckmann$^*$}

\index{van Beusekom, Nathan}
\index{Buchin, Kevin}
\index{Koerts, Hidde}
\index{Meulemans, Wouter}
\index{Rodatz, Benjamin}
\index{Speckmann, Bettina}

\usepackage{enumitem}

\begin{document}
\thispagestyle{empty}
\maketitle

\begin{abstract}
We study metrics that assess how close a triangulation is to being a Delaunay triangulation, for use in contexts where a good triangulation is desired but constraints (e.g., maximum degree) prevent the use of the Delaunay triangulation itself. Our near-Delaunay metrics derive from common Delaunay properties and satisfy a basic set of design criteria, such as  being invariant under similarity transformations. We compare the metrics, showing that each can make different judgments as to which triangulation is closer to Delaunay. We also present a preliminary experiment, showing how optimizing for these metrics under different constraints gives similar, but not necessarily identical results, on random and constructed small point sets.
\end{abstract}

\section{Introduction}
\label{sec:introduction}

Delaunay triangulations are a common construct in computational geometry: practically any class on computational geometry teaches about Delaunay triangulations of point sets and their duality to the Voronoi diagram.
They are efficiently computable, and have many desirable properties, for instance, lexicographically maximizing the minimum angle of the corners of the resulting triangles~\cite{sibson1978locally}, having bounded dilation~\cite{dobkin1990delaunay,keil1992classes},
minimizing the maximum circumcircle~\cite{d1989optimal}, maximizing the minimum enclosed circle~\cite{d1989optimal,rajan1994optimality} of its triangles.
As such, they form the basis of many algorithms that require some triangulation of a point set.

However, there are various scenarios imaginable where the Delaunay triangulation is not immediately applicable due to constraints on the desired triangulation, such as a limited vertex degree or a set of edges that needs to be included. In such cases, we may want to find a triangulation that is as close as possible to the Delaunay triangulation while adhering to the constraints. This, however, requires a way to measure how \emph{near-Delaunay} a triangulation is. Another scenario in which such a measure would be useful is when a triangulation is already given -- for example, the triangulation of a terrain. In such a setting, we could use the measure to assess the quality of the triangulation.

An example of a triangulation that aims to be as close as possible to the Delaunay triangulation given a set of edges that needs to be included is the Constrained Delaunay Triangulation (CDT)~\cite{PaulChew1989,lee1986generalized}.
It provides an alternative definition of when an edge or triangle may be part of the triangulation, such that it is ``as close as possible'' to being Delaunay for the given constraints. However, it does not help in assessing how close a triangulation is to being the actual Delaunay triangulation, nor does it generalize to other forms of constraints.

In this work, we consider various ways to measure
how close a triangulation is to being Delaunay. The problem of studying the properties of triangulations that are close to the Delaunay triangulation was proposed at CCCG 2017 by O'Rourke~\cite{openproblemscccg17}. In this context, two measures were already proposed~\cite{openproblemscccg17,mullen2011hot}, which we discuss further in Section~\ref{sec:metrics}. The aim of our work, is to explore a broader range of measures together with algorithms to compute them and with an analytical and experimental comparison.
We identify several criteria for near-Delaunay metrics:
\begin{enumerate}[label={\textit{C\arabic*}},ref={\textit{C\arabic*}}]
\item \label{c:dtopt} The Delaunay triangulation should obtain the perfect score. We do not want to distinguish how nice Delaunay triangulations are of different point sets -- because this would be a factor of the point set itself. A non-Delaunay triangulation should always score less than perfect, such that any non-Delaunay triangulation is considered less Delaunay.
\item \label{c:continuous} The measure should behave continuously for slight perturbations of the point set. Triangles that ``severely'' violate properties of a Delaunay triangulation should score worse than those with ``slight'' violations.
\item \label{c:invariant} The measure should be invariant under similarity transformations (translations, rotations, scaling and reflections). Though this criteria follows from the first for the Delaunay triangulation (which is invariant under similarity transformations), it is desirable for this to also hold for non-Delaunay triangulations, such that triangulations of different point sets can still be reasonably compared.
\item \label{c:decomp} The measure should be decomposable, that is, evaluated separately on different elements of the triangulation. Though not strictly necessary, this allows various forms of aggregation (worst situation, average situation, etc.).
\end{enumerate}
We observe, that the combination of criterion \ref{c:invariant} and~\ref{c:decomp} suggests that the metric should also be ``locally invariant''. That is, even within the same triangulation, a decomposition element that is subject to the same constraints to another element up to similarity transformations should score the same. That is, the metric should not naturally award higher (or lower) scores to elements that are larger; instead small and large elements should contribute equally to the overall metric. In other words, a triangulation should not be considered near-Delaunay, simply because the deviations from the Delaunay property are only at very small elements.

We propose and compare several \emph{near-Delaunay metrics}.
These metrics (as shown in Table~\ref{tab:results}) satisfy our four criteria. They differ in the properties of the Delaunay triangulation they aim to capture, as well as how they decompose the given triangulation. For this, we use decompositions into: (1) quadrilaterals -- a quadrilateral here is an edge and its two incident triangles, ignoring all other points; (2) edges -- an edge is considered in context of all other points; (3) triangles in context of all other points.
We show that these measures behave differently, capturing different aspects of how near-Delaunay a triangulation is. Finally, we briefly compare how our measures relate to the CDT.

\mypar{Related work}
One of the well-known results in computational geometry is that any triangulation can be transformed into the Delaunay triangulation using Lawson flips \cite{lawson1977software}. A natural consideration would thus be to measure the number of flips necessary to transform a triangulation into a Delaunay triangulation. Though it satisfies \ref{c:dtopt} and \ref{c:invariant}, it satisfies neither \ref{c:continuous} -- it is a discrete measure and it is not immediately based on the violations of a Delaunay property -- nor \ref{c:decomp}. An additional complication is that computing such flip distances is generally hard~\cite{pilz2014flip}.

Another natural consideration is to measure the number of points in any circumcircle of a triangle in the triangulation. This leads to the notion of higher-order Delaunay triangulations~\cite{gudmundsson2002higher}. We do not consider them in this paper, since a small perturbation can greatly change the number of points in a circumcircle. That is, the resulting measure would not adhere to \ref{c:continuous}. But when this criterion is not needed, higher-order Delaunay triangulations are well suited to obtain triangulations close to the Delaunay triangulation that at the same time are optimized for another criterion. Van Kreveld et al.~\cite{van2010optimization} discuss optimizing over first-order Delaunay triangulations for various criteria like minimizing the maximum degree.

Computing a triangulation that is as close as possible to the Delaunay triangulation (given certain constraints) can be seen as optimization problem.
There are many papers studying optimal triangulations under various criteria. Bern et al.~\cite{bern1993edge} show how to efficiently compute triangulations under criteria like maximizing the minimum height or minimizing the maximum eccentricity. Unfortunately for other criteria, computing optimal triangulations is more difficult. For instance computing the minimum-weight triangulation is NP-hard~\cite{mulzer2008minimum}. While the complexity of finding the minimum-dilation triangulation is open~\cite{dilation-survey}, many related problems on minimizing dilation are NP-hard~\cite{giannopoulos2010computing}. Similarly, the complexity of finding the minimum-degree triangulation seems open, while it is NP-hard if as a constraint certain edges have to be included~\cite{jansen1993one,kant1997triangulating}.

The fact that for many optimization criteria efficient algorithms are not known, also limits the size of the point sets that we can include in our experiments, in which we for instance want to compute near-Delaunay triangulations with additional constraints like a degree bound or a bound on the weight of the triangulation. We here resort to enumerating triangulations, however, the total number of triangulation for a given point set is exponential~\cite{aichholzer2004lower,sharir2011counting}.

\begin{table*}[t]
\caption{Overview of metrics. Listed with each is the form of decomposition, the Delaunay property it is based on, and the running time for computing the measure for a given triangulation.}
\label{tab:results}
\small
\begin{tabu} to \textwidth {X[1c,m]|X[1,c,m]|X[1,c,m]|X[1,c,m]|X[1,c,m]|X[1,c,m]|X[1,c,m]}
        \toprule
        Opposing Angles & Dual Edge Ratio & Dual Area overlap & Lens & Shrunk Circle & Triangular Lens & Shrunk Circumcircle\\
        \midrule
        \includegraphics[page=1]{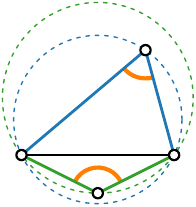} &
        \includegraphics[page=2]{metric_illustrations_small} &
        \includegraphics[page=3]{metric_illustrations_small} &
        \includegraphics[page=4]{metric_illustrations_small} &
        \includegraphics[page=5]{metric_illustrations_small} &
        \includegraphics[page=6]{metric_illustrations_small} &
        \includegraphics[page=7]{metric_illustrations_small} \\
        \midrule
        Quadrilateral & Quadrilateral & Quadrilateral & Edge & Edge & Triangle & Triangle\\
        \midrule
        Max-min angle & Voronoi Dual (edges) & Voronoi Dual (faces) & Empty Circle & Empty Circle & Empty Circumcircle & Empty Circumcircle\\
        \midrule
        $O(n)$ & $O(n)$ & $O(n)$ & $O(n^2)$ & $O(n^2)$ & $O(n^2)$ & $O(n^2)$\\
        \bottomrule
\end{tabu}
\end{table*}

\section{Near-Delaunay metrics}\label{sec:metrics}

We define seven near-Delaunay metrics as shown in Table~\ref{tab:results}, which all satisfy our four criteria. Throughout, we assume that we are to measure a triangulation $T$ on a point set $P$. We further assume, for sake of simplicity, that $P$ is in general position: no three points are on a line and no four points are cocircular. We organize them below into three categories, depending on their form of decomposition (quadrilateral, edge and triangle). We always describe the measure just for a single decomposition element, silently assuming some form of aggregation such as taking their sum or extremal value.

\subsection{Quadrilateral-based metrics}

We start by introducing intuitive metrics that evaluate quadrilaterals. A quadrilateral consists of a non-convex-hull edge and the two triangles that are incident to it, defined by two vertices. Particularly, note that any point in $P$ that is not one of the four defining vertices does not influence the metric on this particular quadrilateral -- this contrasts the edge-based and triangle-based metrics. Thus, to compute the metrics it is sufficient to be given only the quadrilateral.
Consequently, the metrics we introduce for quadrilaterals can be computed in constant time per quadrilateral and in linear time overall. Note that for each quadrilateral-based metric, lower scores mean closer to Delaunay (contrasting our other two forms of decomposition).

Throughout, we denote by $\edge{u}{v}$ the defining edge of the quadrilateral, and with $p$ and $q$ the \emph{opposing} vertices of the two incident triangles $\triangle{u}{v}{p}$ and $\triangle{u}{v}{q}$. We use $c_p$ and $c_q$ to denote the center of the circumcircles of these triangles. A quadrilateral is \emph{locally Delaunay} if the edge $\edge{u}{v}$ is part of the Delaunay triangulation of $\{u,v,p,q\}$ -- in other words, if the circumcircle of the one triangle does not contain the other vertex.

\mypar{Opposing Angles}
At CCCG 2017, O'Rourke suggested that a triangulation $T$ is a near-Delaunay triangulation if the opposite angles $\alpha$ and $\beta$ of a quadrilateral sum to at most $\pi + \varepsilon$ for $\varepsilon \geq 0$ \cite{openproblemscccg17} (see Fig.~\ref{fig:opposing_angles}). If $\varepsilon = 0$, then $T$ is Delaunay. Hence, it is natural to consider the smallest $\varepsilon$ for a triangulation as a metric for how close a triangulation is to a Delaunay triangulation.
We can readily interpret $\varepsilon$ on a per-quadrilateral basis for a metric in our context; O'Rourke's suggestion is simply to take their maximum as the overall metric.

\begin{figure}[b]
    \centering
    \includegraphics[page=1]{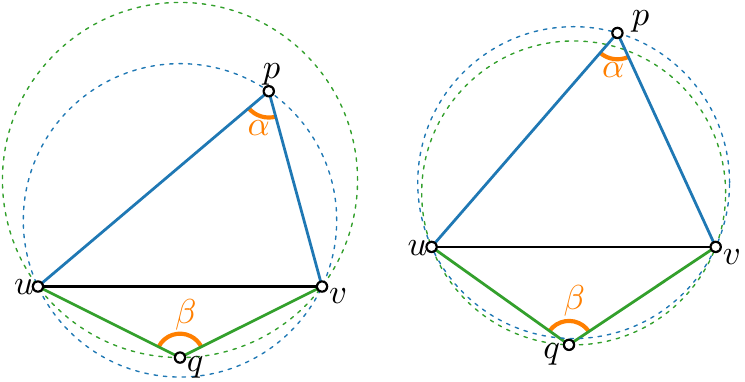}
    \caption{Opposing Angles metric. (left) Sum of opposing angles is less than $\pi$ for quadrilateral that is not locally Delaunay. (right) It is at least $\pi$ for quadrilaterals that are locally Delaunay.}
    \label{fig:opposing_angles}
\end{figure}

The intuition behind the metric is, when the sum of two opposing angles is larger than $\pi$, the empty circumcircle property is violated. When the sum is close to~$\pi$, a slight movement in the points can restore the empty circumcircle property. However as the sum grows larger, the points generally need to move further to restore the property, unless $p$ or $q$ is very close to $u$ or $v$, in which case a small movement is sufficient.
Thus, we can use the sum of two opposing angles to evaluate how far from Delaunay a quadrilateral is.

\mypar{Dual Edge Ratio}
Additionally, Mitchell~\cite{openproblemscccg17} suggested ``measuring the signed distance between circumcenters of triangles sharing  an  edge;  for  Delaunay  triangulations  this is  simply  the  dual  edge  length  and  non-negative, but  for  non-Delaunay triangulations  the  circumcenters  can  be  in  the wrong  order  and hence  have  a  negative  distance between them.  So one could look at the ratio of the dual edge signed-length to the primal edge length (for 2D triangulations) as a continuous measure of how close it is to non-Delaunay.''~\cite{openproblemscccg17} -- which are also known as Hodge-optimized triangulations \cite{mullen2011hot}.

We adapt this metric to evaluate quadrilaterals: we measure only the negative distance part of the suggestion, to satisfy criterion~\ref{c:dtopt} and not distinguish between Delaunay triangles. The ``wrong order'' referred to above matched to the quadrilateral being not locally Delaunay.
We thus define the Dual Edge Ratio as
\[
\begin{cases}
      0, & \parbox[t]{.35\textwidth}{if the quadrilateral is locally Delaunay.} \\
      \frac{d(c_p,c_q)}{d(u,v)}, & \text{otherwise.}
\end{cases} \]

Intuitively, if the quadrilateral is not locally Delaunay, ratio of the distance between $c_p$ and $c_q$ and the length of $\edge{u}{v}$ roughly corresponds to how skinny the triangles are and thus how far from Delaunay they are as well (see Fig.~\ref{fig:dual_edge_ratio}).
Clearly, the Delaunay triangulation scores zero on all its quadrilaterals. Any quadrilateral that is not locally Delaunay, and hence does not locally describe the dual of the Voronoi diagram, will score greater than zero. Any non-Delaunay triangulation must have at least one such quadrilateral.

\begin{figure}[t]
    \centering
    \includegraphics[page=2]{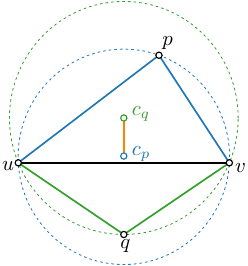}
    \caption{Dual Edge Ratio metric. The dual edge (orange) for a quadrilateral that is not locally Delaunay. We measure its length relative to the length of $\edge{u}{v}$.}
    \label{fig:dual_edge_ratio}
\end{figure}

\newpage
\mypar{Dual Area Overlap}
Similar to the previous metric, we may use the duality to Voronoi diagrams but in a different manner. Rather than looking at the distance, we may also consider an area-based metric. Intuitively, we consider the incorrect area of overlap between the ``local Voronoi cells'' that we may construct from the triangles of the quadrilateral.

The local Voronoi cell of $p$ (and $q$ analogously) is defined by the bisectors of $u$ and $v$ with $p$, which cross in $c_p$. If the quadrilateral is locally Delaunay, then the cells of $p$ and $q$ are disjoint. However, if the quadrilateral is not locally Delaunay, then they must overlap. The area of this overlap divided by the squared length of $\edge{u}{v}$ is our Dual Area Overlap metric (see Fig.~\ref{fig:dual_area_overlap}). We normalize using the squared edge length here, to ensure scale invariance, criterion~\ref{c:invariant}.

\begin{figure}[b]
    \centering
    \includegraphics[page=3]{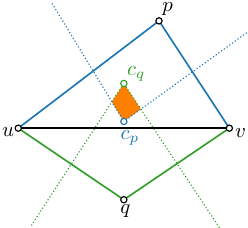}
    \caption{Dual Area Overlap metric. The local Voronoi cells of $p$ and $q$ (dotted) for a quadrilateral that is not locally Delaunay. We measure the overlap (orange).}
    \label{fig:dual_area_overlap}
\end{figure}

The intuition behind this metric is similar to the intuition from the Dual Edge Ratio metric. The area of the overlap of the two Voronoi regions implies how far the quadrilateral is from having a proper non-overlapping Voronoi dual.
A larger area means that points will have to move further to reach a non-overlapping dual, while a smaller area implies that a small movement in the points can already achieve this.

\subsection{Edge-based metrics}

We now turn to edge-based metrics, that evaluate each edge $\edge{u}{v}$ of $T$ in context of all other points in $P$. That is, it penalizes edges, even if the defined quadrilateral is locally Delaunay.
We present two new metrics below, both of which are based on the same principle: as an edge of the Delaunay triangulation must be the chord of a circle that does not strictly contain any other vertices of $P$, we consider how much we much deform a circle to find such an empty deformed circle instead. The difference between our two metrics is how they perform this deformation. Note that in both cases, higher scores mean closer to Delaunay, contrasting quadrilateral-based metrics.

\mypar{Lens}
When an empty circle exists, it can be seen as two circulars arcs, one on each side of the edge. The arc in the one halfplane with respect to the line spanned by $\edge{u}{v}$ excludes from its interior all vertices of $P$ in that same halfplane. With our Lens metric, we reverse this idea to deform our circle into a lens; the ``sharpness'' of this lens is then our metric.

Specifically, consider all points $P' \subseteq P$ that lie on one side of the line spanned by edge $\edge{u}{v}$. The circular arc from $u$ to $v$ through some point of $P'$ that is minimal in terms of segment area (or equivalently, arc length or central angle) is the largest arc possible on this side of the edge that does not contain any point of $P'$ in its segment area. Let $a$ and $a'$ denote the two circular arcs obtained this way. We consider the ``interior'' angle $\alpha$ between the tangent directions of $a$ and $a'$ at $u$ as our metric (see Fig.~\ref{fig:lens}). If $\alpha \geq \pi$, the edge is a Delaunay edge and we cap the metric to $\pi$ to satisfy criterion~\ref{c:dtopt}; for any non-Delaunay edge, $\alpha < \pi$.

\begin{figure}[h]
    \centering
    \includegraphics[page=4]{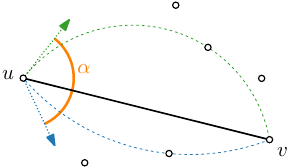}
    \caption{Lens metric. We find the largest empty arc on both sides of $\edge{u}{v}$, and measure the angle $\alpha$ between their tangent directions at $u$.}
    \label{fig:lens}
\end{figure}

We can easily compute the metric for a triangulation in quadratic time: for each edge, find the point that gives the smallest arc on both sides; then compute the angle at which the arcs touch.

\mypar{Shrunk Circle}
With our second metric, we consider a different type of deformation: scaling. That is, we aim to find a smaller empty circle that relates to the edge we wish to measure. For a non-Delaunay edge, such a circle cannot have $\edge{u}{v}$ as a chord, but it may still overlap the edge. As a Delaunay edge would be overlapped fully by an empty circle, we define our Shrunk Circle metric as the maximal fraction of the edge that can be overlapped by an empty circle $C$ (see Fig.~\ref{fig:shrunk_circle})). Note that, whereas the Lens metric considers both halfplanes independently, this is not the case here.

\begin{figure}[t]
    \centering
    \includegraphics[page=5]{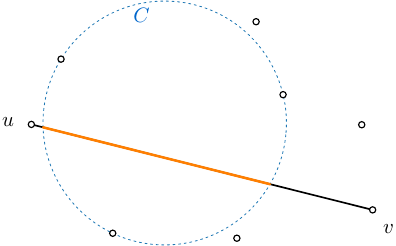}
    \caption{Shrunk Circle metric. We find the empty circle $C$ that covers edge $\edge{u}{v}$ most (orange).}
    \label{fig:shrunk_circle}
\end{figure}

To compute this measure for an edge $\edge{u}{v}$ of triangulation $T$ on point set $P$, we observe the following: if an empty circle does not touch at least two points of $P$, we can readily grow it to a circle $C'$ that does touch two points and strictly encompasses the previous circle and thus the overlap with $\edge{u}{v}$ does not decrease. In other words, we need to consider only maximal circles with centers on the Voronoi diagram $P$. The lemma below argues that the overlap along one edge of the Voronoi diagram is convex and thus we need to test only its endpoints. In fact, we need to test only its vertices since unbounded edges occur only for pairs of points on the convex hull and thus the overlap of such circles with $\edge{u}{v}$ are determined only by the circle's part inside the convex hull -- which is maximized at bounded side of the unbounded edge. We can thus first compute the Delaunay triangulation for $P$, and then for every edge $\edge{u}{v}$ of $T$, test the circumcircles of the Delaunay triangulation explicitly. This readily gives an quadratic-time algorithm for the overall metric.

\begin{lemma}\label{lem:convex}
Let $e$ be an edge of the Voronoi diagram of $P$. The overlap of maximal circles along $e$ with edge $\edge{u}{v}$ is a convex function.
\end{lemma}
\begin{proof}
Let $p$ and $q$ denote the points defining $e$. As the problem is invariant under translation, rotation and scaling, assume without loss of generality that $p = (0,1)$ and $q=(0,-1)$; this implies that $e$ is along the horizontal axis.
Maximal circles $C'$ are thus fully defined by their center $(m,0)$.
Let $\ell \colon y = a x + b$ denote the line spanned by $\edge{u}{v}$.
We first consider the overlap of $C'$ with $\ell$ as a function of $m$.

The two intersection points of $\delta C'$ with $\ell$ are obtained by solving $|(x,ax+b)-(m,0)|_2 = |(c,0)-p|_2$, which simplifies to $(a^2+1) x^2 + 2(ab - m)x + (b^2 - 1) = 0$. The difference in x-coordinates between the two solutions to this quadratic equation are given by $\sqrt{D}/(a^2+1)$, where $D = 4(ab - m)^2 - 4 (a^2+1) (b^2 - 1)$ is the discriminant; the amount over overlap with $\ell$ is thus $a \sqrt{D}/(a^2+1)$.
As $D$ is a convex quadratic function in $m$, so is $a \sqrt{D}/(a^2+1)$. To note, $D$ is negative if the circle does not overlap $\ell$, in which case the overlap is trivially zero: technically, the overlap is thus a function $a \sqrt{\max\{0,D\}}/(a^2+1)$.
This proves that the overlap with $\ell$ is a convex function.

Since we are only interested in the values of $m$ along $e$, that is, for which $C'$ is empty, the overlap with $\ell$ is either fully within $\edge{u}{v}$ or fully outside. Since the intersections behave continuously, the overlap with $\edge{u}{v}$ thus behaves convex as well.
\end{proof}

\subsection{Triangle-based metrics}

For our triangle-based metrics, we measure a triangle $\xtriangle{u}{v}{w}$ of $T$ in context of all other points in $P$. Analogous to our edge-based metrics, we deform the circumcircle $C$ of $\xtriangle{u}{v}{w}$ (which is empty for a Delaunay triangle) to find a suitable empty deformed circle.  In contrast to the edge-based metrics, we now have a single fixed circle $C$ which guides (and constrains) our metric. Specifically, we restrict our deformations to be contained in the circumcircle. As with edge-based metrics, higher scores mean closer to Delaunay.

\begin{figure}[b]
    \centering
    \includegraphics[page=6]{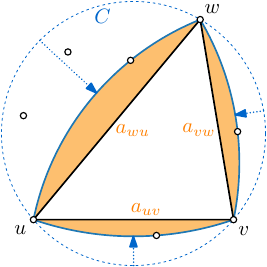}
    \caption{Triangular Lens metric. We compute the largest arc inside $C$ that does not contain any points in its segment area for each edge of the triangle. Note that the three arcs are determined independently.}
    \label{fig:triangular_lens}
\end{figure}

\mypar{Triangular Lens}
Similarly to the Lens metric, when a empty circumcircle exists, it can be considered as three circular arcs on the outside of the triangle. We again replace each arc by the largest arc that is contained in the arc of the circumcircle and that contains no other points of $P$.

For our metric, we measure the fraction of the area in $C$ but outside $\xtriangle{u}{v}{w}$ that is covered by the constructed lens. Let $a_{uv}$, $a_{vw}$ and $a_{wu}$ denote the segment areas of the three arcs constructed (see Fig.~\ref{fig:triangular_lens}). Interpreting $C$ and $\xtriangle{u}{v}{w}$ as their enclosed areas, the score is then
\[ \frac{a_{uv}+a_{vw}+a_{wu}}{C - \xtriangle{u}{v}{w}}. \]

We subtract the triangle to be able to assess and to meaningfully aggregate both skinny and fat triangles; this also means that each triangle's score lies in $(0,1]$, where a score of 0 is only achievable in the limit (a point converging on each edge).

We easily compute the metric for a triangulation in quadratic time: for each triangle, find for each edge the smallest arc by testing all other vertices; then compute the segment areas and compute the resulting fraction.

\mypar{Shrunk Circumcircle}
As with our edge-based measures, the Triangular Lens metric deforms independently in the three segment areas defined by the edges. With our Shrunk Circumcircle metric, we consider a variant that considers all points simultaneously instead.

Specifically, we aim to find an empty circle $C'$ that is contained in the circumcircle $C$ and intersects all three sides of $\xtriangle{u}{v}{w}$.
The score we associate with this circle $C'$ is $\frac{C' - I}{C - I}$, where $I$ is the area of the inscribed circle and we identify $C'$ and $C$ with the areas of these circles as well (see Fig.~\ref{fig:shrunk_circumcircle} (left)).

Note that we subtract the areas (numbers), we do not use set subtraction (difference of the shapes) and their resulting area; in our figures, we use examples where the inscribed circle is contained in the largest empty circle such that these two variants are the same, to easily visualize the score, but this is not necessarily the case.

\begin{figure}[b]
    \centering
    \includegraphics[page=7]{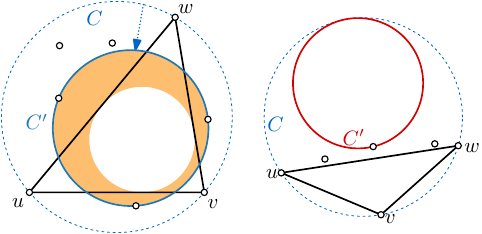}
    \caption{Shrunk Circumcircle metric. (left) We look for the empty circle $C'$ in $C$ of maximal radius. The score (orange) is its area minus the inscribed circle. (right) $C'$ must intersect all three sides, for the measure to be meaningful when points are close to the triangle.}
    \label{fig:shrunk_circumcircle}
\end{figure}

The Shrunk Circumcircle metric is the circle $C'$ with the highest associated score.
For any triangle, $I$ and $C$ are constant, and thus it is simply the largest circle satisfying the two constraints.
We subtract $I$ as a lower bound on the largest empty circle (since the triangle itself must be empty), and divide by $C - I$ to normalize the the score to the range $(0,1]$ independent of triangle shape and make the metric scale invariant.

We constrain $C'$ to lie within $C$, such that the area we count is always an actual circle: if $C'$ was to be allowed to grow outside $C$, it would either count area outside of the area that the Delaunay triangulation ``considers'', or the region we use the area of is not a circle in itself.

We constrain $C'$ to intersect all three sides to make the score meaningful, even for very skinny triangles: otherwise, the largest empty circle may simply be fully in one segment, and even yield a relatively high score, though points are very close to the triangle edges (see Fig.~\ref{fig:shrunk_circumcircle} (right)).

To compute the metric for a triangle $\xtriangle{u}{v}{w}$ in context of $P$, we first argue about the properties of largest $C'$.
Circle $C'$ must either touch two points strictly inside $C$, or touch one such point and the boundary of $C$ -- otherwise, we can grow the circle into one that strictly encloses $C'$ while still adhering to the constraints. As such, its center lies on the edges of the \emph{local Voronoi diagram} of circle $C$ and the points of $P$ inside $C$ (see Fig.~\ref{fig:curvedVoronoiCompute}). Curved segments for circle centers equidistant to the boundary of $C$ and one of the points inside $C$, and straight segments defined by two points inside $C$. These curved segments are elliptical: the distance from a point on this curve to the center of $C$ and to the contained point of $P$ sums up to the radius of $C$. By construction, segments incident to the outerface of the diagram are curved segments; interior segments are straight.

\begin{figure}[b]
    \centering
    \includegraphics{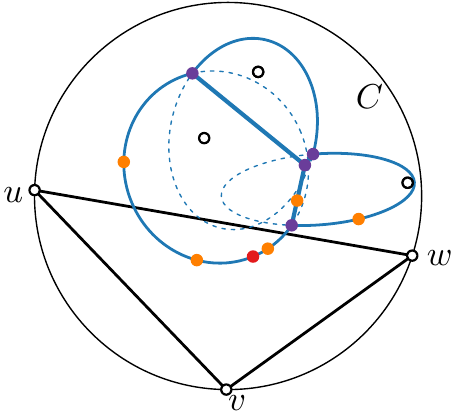}
    \caption{The local Voronoi diagram (blue) consists of elliptical and straight segments. The critical placements: endpoints of straight segments (purple), the furthest point on an ellipse (red), and circle touching a triangle edge (orange).}
    \label{fig:curvedVoronoiCompute}
\end{figure}

In case the entire ellipse is part of the local Voronoi diagram (and hence in fact the only curve), we may consider it an elliptical arc, starting and ending at the furthest point from its defining point.
Then, both along a straight segment as well as along a curved segment of the local Voronoi diagram, the radius of $C'$ behaves unimodally: the function has a single (local and global) minimum at the closest point of the segment to the defining point(s). Along a segment, the maximal circle is hence either one of a small set of \emph{critical placements} (if they exist): the endpoints of the segment, the furthest point of the ellipse, or one of the $O(1)$ placements where the defined circle touches a triangle edge. We further observe that the straight segments are a subset of the (full) Voronoi diagram of $P$.

With the insights above, we can define an algorithm to compute the Shrunk Circumcircle metric for all triangles in $T$ in quadratic time. First, we compute the (full) Voronoi diagram of $P$. Then, for each triangle, we compute the metric as follows, in linear time. If there are no points inside the circumcircle, we have a Delaunay triangle and the metric is 1. If the circumcircle contains one point, then the local Voronoi diagram is a single ellipse and we test the critical placements. If the circumcircle contains more than one point, we find all straight segments local Voronoi diagram by traversing the full Voronoi diagram, testing the critical placements. For any straight segment found as such, we first shorten it to ensure that it does not define circles extending outside $C$. Then we test its critical placements. If the segment was shortened, we know that its defining points also define a segment of the local Voronoi diagram. For these segments, we also test the critical placements. By ``testing'' in the above, we mean testing whether the defined circle intersects all three sides (it is contained in $C$ by construction), and if so, see if its radius is larger than any circle found so far. As we test $O(1)$ critical cases per segment, computing the metric takes linear time per triangle (after computing the full Voronoi diagram in $O(P \log P)$ time once) and thus quadratic time for the entire triangulation.

\section{Comparing metrics}

In the previous section, we defined seven near-Delaunay metrics. These capture in different ways how close to Delaunay a triangulation is. Though future work may endeavor to establish a standard here, it is not a-priori clear which measure is ``the best'': this likely depends on context, that is, the purpose of evaluating a triangulation. The question we ask here, is whether these metrics actually capture different facets of being ``near-Delaunay''. Specifically, given two metrics, do they always evaluate the same triangulation to be closer to Delaunay, for any given pair of triangulations?

Considering a single decomposition element to answer this question, there is clear distinction between the forms of decomposition: quadrilateral-based metrics do not take other points into account, contrasting the other two forms; edge-based metrics use angles and lengths, whereas triangle-based metrics use area ratios instead. We thus focus here on comparisons between metrics using the same form of decomposition. We study how the metrics differ from each other, and what properties they value. Specifically, for each comparison of metrics $\mu$ and $\mu'$, we show that there are triangulations $T$ and $T'$, such that $\mu(T) = \mu(T')$ and $\mu'(T) > \mu'(T')$. Such an example answers the above question of making the same judgments negatively.

\mypar{Quadrilateral-based metrics}
Our three quadrilateral-based metrics use only four points in measuring one element. Yet, we show that each metric evaluates a different facet of being near-Delaunay. Note that a (convex) quadrilateral is immediately a triangulation on this point set.

The Opposing Angles metric does not consider how the two summed angles are distributed over the two triangles. Even a very flat triangle (angle approaching $\pi$) can be offset a very tall triangle (angle approaching 0). Yet, these angles behave very differently from the distances between the circumcenters and thus the Dual Edge Ratio and Dual Area Overlap.
Thus, we can construct two quadrilaterals (see Fig.~\ref{fig:opposing_angles_diff}): one has very similar triangles, while the other has very different triangles, but both have the same edge $\edge{u}{v}$. Whereas the Opposing Angles metric scores the same on both, we see readily that the Dual Edge Ratio and Dual Area Overlap score differently.

\begin{figure}[b]
    \centering
    \includegraphics[page=1]{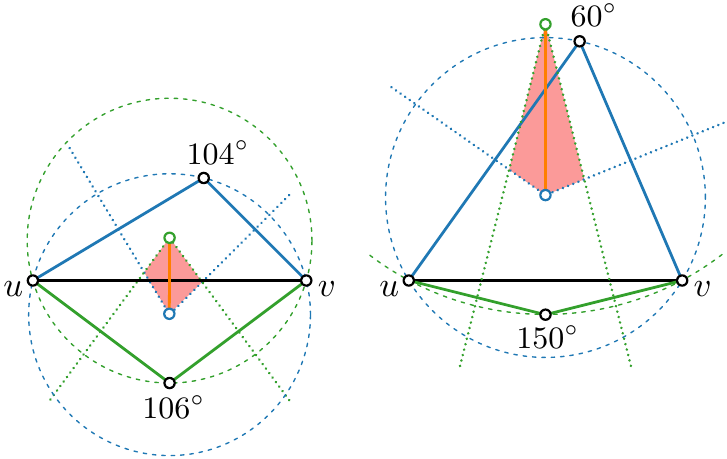}
    \caption{Two quadrilaterals with equal (sum of) Opposing Angles, but different Dual Edge Ratio (orange) and Dual Area Overlap (red).}
    \label{fig:opposing_angles_diff}
\end{figure}
\begin{figure}[b]
    \centering
    \includegraphics[page=2]{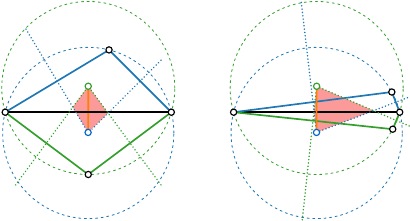}
    \caption{Two quadrilaterals with equal Dual Edge Ratio (orange) and different Dual Area Overlap (Red).}
    \label{fig:dual_metrics_diff}
\end{figure}

Comparing the Dual Edge Ratio to the Dual Area Overlap, we see that the former is based purely on the distance between the circumcenters, whereas the latter depends also on the shape of the triangles. This allows us to construct another two quadrilaterals (see Fig.~\ref{fig:dual_metrics_diff}): we move the opposing points along their defined circumcircles so as to not move the circumcenters, while changing the area of overlap.

\mypar{Edge-based metrics}
We have two metrics here: Lens and Shrunk Circle. As already mentioned in the previous sections, these differ by how they treat points on the different sides of the edge. Whereas the Lens measure treats these independently, the Shrunk Circle measure requires an integrated consideration of all points. This allows us to construct two triangulations again, using only four points in convex position (see Fig.~\ref{fig:edge_metrics_diff}): we keep the Lens measure constant, by moving one of the vertices over the defining arc, thus keeping the tangent at $u$ constant as well; in contrast, the Shrunk Circle metric in the first example uses a circle that encompasses a large part of the defining arc of the Lens -- by placing the point there, we can force the circle to shrink further and cover less of the edge $\edge{u}{v}$.

\begin{figure}[h]
    \centering
    \includegraphics[page=3]{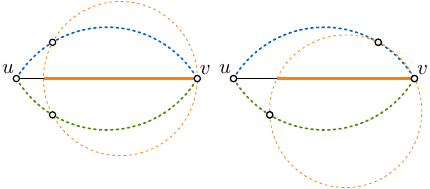}
    \caption{Two edges with equal Lens (blue, green) different Shrunk Circle (orange).}
    \label{fig:edge_metrics_diff}
\end{figure}

\begin{figure}[b]
    \centering
    \includegraphics[page=4]{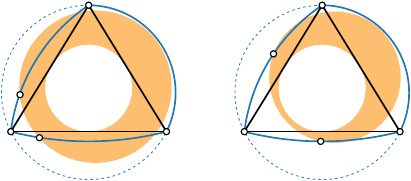}
    \caption{Two triangles with equal Triangular Lens (blue) and different Shrunk Circumcircle (orange).}
    \label{fig:triangle_metrics_diff}
\end{figure}

\mypar{Triangle-based metrics}
Finally, we consider the two triangle-based metrics. They behave somewhat similarly with respect to each other as the edge-based metrics do: whereas the Triangular Lens works with independent arcs per edge, the Shrunk Circumcircle uses a single circle that must overlap each of the three edges. We can thus follow the same principle to show two triangles in a point set that score equally on the Triangular Lens, but differently on the Shrunk Circumcircle, by moving the points along the arcs of the former to force the defining circle of the latter to shrink (see Fig.~\ref{fig:triangle_metrics_diff}).

\section{Experiments}

Here we explore what the most Delaunay-like triangulation of a point set looks like, for each of our metrics, given different constraints that can force a triangulation to be non-Delaunay. For a point set, we try all triangulations that adhere to the given constraints to find the \emph{optimized} triangulation that scores best according to each metric. We use small sets of only 10 points, to ensure that this is feasible.

\mypar{Constraints}
We consider four types of constraints: constrained edges, a lower bound or upper bound on the total edge length, and a maximum degree.

For constrained edges, we are given a set of edges that must be included in the triangulation. The constrained edges are handpicked edges, which we consider to be ``interesting''. Most importantly the constrained edges are not in the Delaunay triangulation, are not chord of the convex hull, and lie somewhat close to another point.
This is the same constraint as for the CDT and hence we may also compare how this structure compares to our result.

The length of the triangulation is the sum over the lengths of its edges. As the Delaunay triangulation neither minimizes nor maximizes the length, we can use an upper bound (maximum length) or a lower bound (minimum length), to constrain the triangulation to be non-Delaunay. Specifically, we use either 1.2 times the length of the Delaunay triangulation as a lower bound, or 0.8 times the length of the Delaunay triangulation as an upper bound. However, since the Delaunay triangulation is inherently short for many point sets, such triangulations often do not exist. We hence also create a special point set where the Delaunay triangulation is longer than most other triangulations, to see the differences between metrics.

Bounding the maximum degree means we consider only triangulations for which all vertices have degree at most a given constant. We use maximum degree 5 in our experiment. This number is sometimes exceeded by the Delaunay triangulation, but still allows for different triangulations of the same point set.

\mypar{Aggregation}
For our optimized triangulations, we have to evaluate each metric on the entire triangulation. So far, we have left the method of aggregation out of our considerations. For the purpose of our experiment, we consider two methods: using the sum and using bottleneck values.

The sum is a natural way to aggregate the values of a triangulation $T$, such that all decomposition elements (quadrilaterals, edges or triangles) have an impact on the triangulation. For quadrilateral-based metrics, we minimize the sum; otherwise, we maximize it.

Using bottleneck values means we focus on the worst-case element (maximum for quadrilateral-based metrics, minimum for our other metrics). The goal of this approach is to let the worst value be the deciding factor. A problem with this method is that there can be two different triangulations that score the same, as they both include the same bottleneck. Hence, we compare triangulations lexicographically: this means that the elements are sorted (increasingly for quadrilateral-based, decreasingly for other metrics) and the first element in which the triangulations differ, determines which triangulation is closer to Delaunay.

\mypar{Results}
In the appendix, we show various tables with our experimental results; the most relevant excerpts are shown in Table~\ref{tab:overview}. For the constrained edges, we show the CDT for comparison, using red edges to indicate the constraints. For the other constraints, we show the Delaunay triangulation for comparison. In both cases, the optimized triangulations use green markings to show edges that are different from the comparison.

Table~\ref{tab:cdt_sum} shows optimized triangulations with constrained edges using sum aggregation. We observe that most metrics are always similar to the CDT for such random point sets. The only differences are for Dual Edge Ratio and Dual Area Overlap: each time Dual Edge Ratio is different from the CDT, Dual Area Overlap is also different, though not necessarily vice versa. One exception is for Shrunk Circle. This suggests to us that, at least in small random cases, our metrics capture near-Delaunay quite well, as the CDT is an established way of getting a Delaunay-like triangulation, for these constraints.
Using bottleneck aggregation (Table~\ref{tab:cdt_bottleneck}), we observe that the quadrilateral metrics often behave similarly and are regularly different from the CDT. With one exception, the other optimized triangulations are again all identical to the CDT.

Tables~\ref{tab:minlength_sum} and~\ref{tab:minlength_bottleneck} show the optimized triangulations with a minimum-length constraint. We can compare the different optimizations to empirically evaluate the behavior of the metrics. We observe that in the sum aggregation quadrilateral-based metrics generally behave differently than the other metrics. However, this distinction is less clear using the bottleneck aggregation. Furthermore, with bottleneck aggregation the Shrunk Circle metric often produces a unique triangulation. Note that this is the same case that was excepting from the general trend for constrained edges.

Table~\ref{tab:maxlength_sum} and~\ref{tab:maxlength_bottleneck} show the optimized triangulation with a maximum-length constraint. If no such triangulation exists, we simply show the Delaunay triangulation. In every point set except the one specifically created for this constraint, there was no such triangulation. For this case, we observe that the Dual Edge Ratio and Dual Area overlap produce different triangulations from the other metrics. In the bottleneck variant, we also see that the Lens measure procedures another distinct triangulation.

Table~\ref{tab:maxdegree_sum} and~\ref{tab:maxdegree_bottleneck} shows the optimization with a maximum-degree constraint. In only few of the random point sets, the maximum degree exceeds the constraint. Most metrics flip the same edge, the notable exception being the Shrunk Circumcircle metric which flips two different edges for one case (shown). The two other point sets shown were specifically created for this case.
We see considerable differences here between metrics, as well as between sum and bottleneck aggregation.

With the wheel example here, we may perhaps see the one case where there is a somewhat clear case of a ``visually nice'' optimized triangulation. Specifically, a somewhat regular pattern emerges for Lens and Shrunk Circle using sum aggregation as well as all metrics except for the Dual Edge Ratio and Dual Area Overlap using bottleneck aggregation -- in the other cases we see very skinny triangles occurring. Whether this form of being near-Delaunay, however, is the most useful remains possibly context dependent. This hints at bottleneck aggregation perhaps being more useful -- indeed, it matches the traditional lexicographic optimization of the minimum angle that the Delaunay triangulation achieves.

\section{Discussion}

With a suite of metrics, we now have a common framework to think about situations where we need a good (Delaunay-like) triangulation to compute with, but constraints such as bounded degree prevent us from actually using the Delaunay triangulation itself. We have shown how these metrics differ among themselves as well as how they result in different triangulations when considering optimizing under the constraint of including given edges (like the Constrained Delaunay Triangulation).

We leave to future work to establish efficient algorithms to compute the best triangulation (given one of the metrics) given a set of constraints. It may further be interesting to investigate how humans (or computational geometers) assess the quality of a triangulation, how close it is to being Delaunay, and how this relates to the metrics provided here. This may uncover that at least to match intuition, we may need combinations of metrics, or possibly a different metric altogether.

A possible avenue for further metrics is to consider the empty ellipse as the generalization of the empty circle, trying to optimize for the ellipse's aspect ratio. This seems mostly relevant for edge-based or quadrilateral-based metrics, since triangles do not necessarily allow for empty ellipses passing through their corners. But for a single edge or quadrilateral, a general ellipse seems to provide too much freedom, allowing for somewhat arbitrary-seeming results. Restricting one of the axes of the ellipse to be parallel to the defining edge may offer a solution; yet, this seems to cause a counterintuitive relation between the two sides of the edge.

\begin{table*}[b]
    \centering
    \caption{Excerpt of the most prominent results.}
    \label{tab:overview}
    \small
    \begin{tabular}{c}
        \textbf{Optimizing with constrained edges} (Sum aggregation)
    \end{tabular}
    \begin{tabu} to \textwidth {X[1,c,m]|X[1,c,m]|X[1,c,m]|X[1,c,m]|X[1,c,m]|X[1,c,m]|X[1,c,m]|X[1,c,m]}
            \toprule
            CDT & Opposing Angles & Dual Edge Ratio & Dual Area overlap & Lens & Shrunk Circle & Triangular Lens & Shrunk Circumcircle\\
            \midrule
            \includegraphics[page=1, width=1\linewidth]{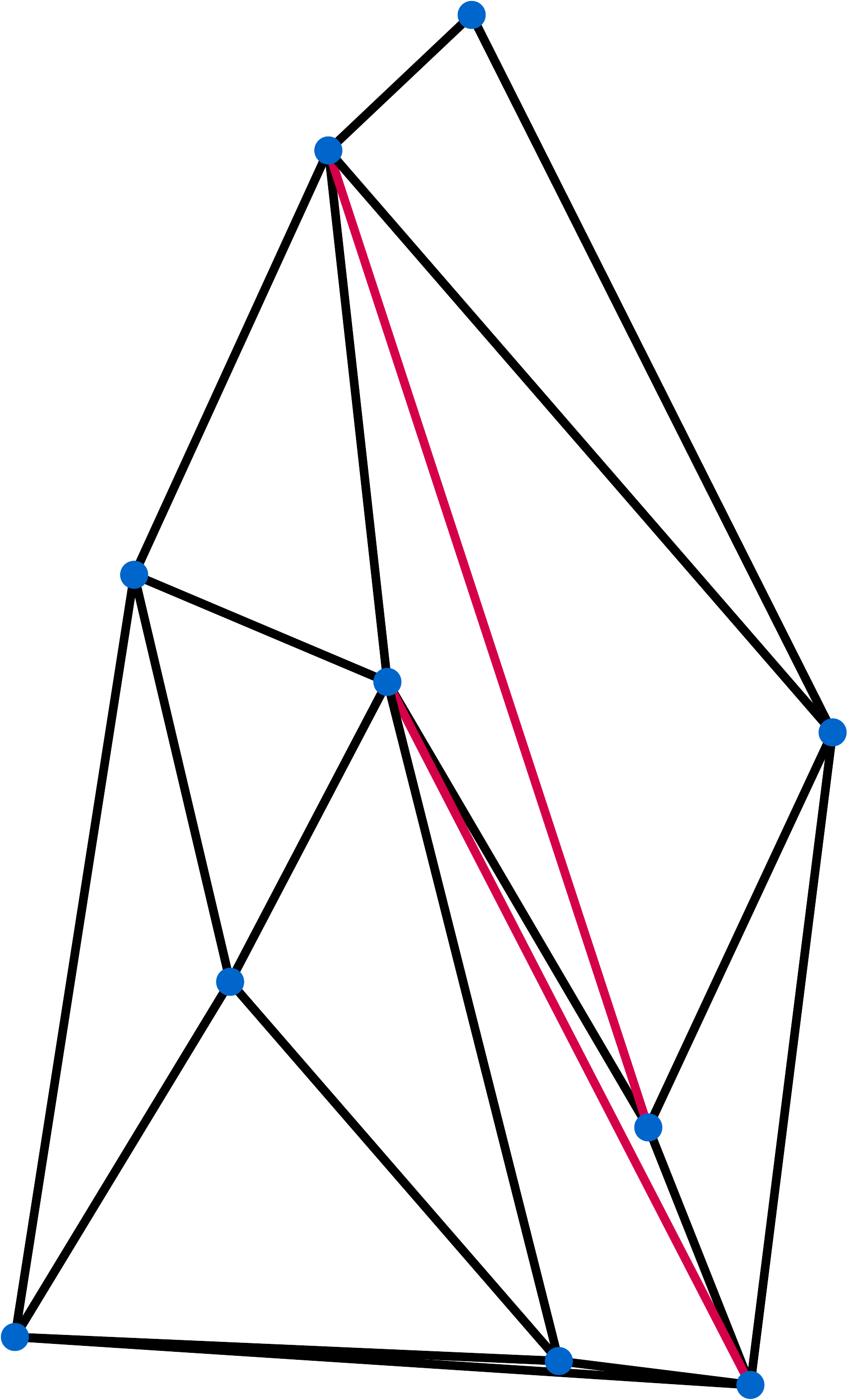} &
            \includegraphics[page=2, width=1\linewidth]{figures/cdt6sum.pdf} &
            \includegraphics[page=3, width=1\linewidth]{figures/cdt6sum.pdf} &
            \includegraphics[page=4, width=1\linewidth]{figures/cdt6sum.pdf} &
            \includegraphics[page=5, width=1\linewidth]{figures/cdt6sum.pdf} &
            \includegraphics[page=6, width=1\linewidth]{figures/cdt6sum.pdf} &
            \includegraphics[page=7, width=1\linewidth]{figures/cdt6sum.pdf} &
            \includegraphics[page=8, width=1\linewidth]{figures/cdt6sum.pdf}\\
            \bottomrule
    \end{tabu}
    \begin{tabular}{c}
         \textbf{Optimizing with constrained edges} (Bottleneck aggregation)
    \end{tabular}
    \begin{tabu} to \textwidth {X[1,c,m]|X[1,c,m]|X[1,c,m]|X[1,c,m]|X[1,c,m]|X[1,c,m]|X[1,c,m]|X[1,c,m]}
            \toprule
            CDT & Opposing Angles & Dual Edge Ratio & Dual Area overlap & Lens & Shrunk Circle & Triangular Lens & Shrunk Circumcircle\\
            \midrule
            \includegraphics[page=1, width=1\linewidth]{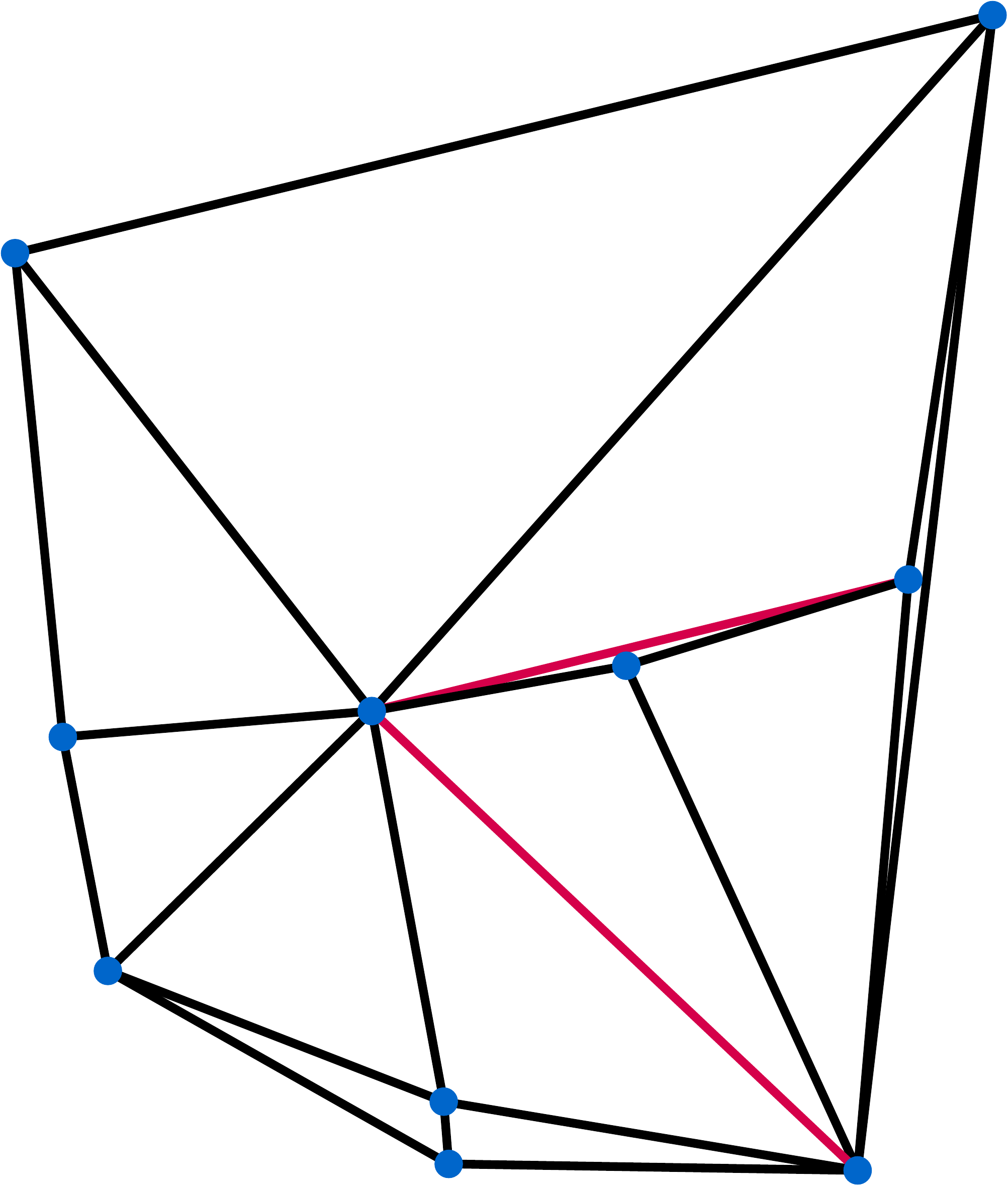} &
            \includegraphics[page=2, width=1\linewidth]{figures/cdt4max.pdf} &
            \includegraphics[page=3, width=1\linewidth]{figures/cdt4max.pdf} &
            \includegraphics[page=4, width=1\linewidth]{figures/cdt4max.pdf} &
            \includegraphics[page=5, width=1\linewidth]{figures/cdt4max.pdf} &
            \includegraphics[page=6, width=1\linewidth]{figures/cdt4max.pdf} &
            \includegraphics[page=7, width=1\linewidth]{figures/cdt4max.pdf} &
            \includegraphics[page=8, width=1\linewidth]{figures/cdt4max.pdf}\\
            \bottomrule
    \end{tabu}
    \begin{tabular}{c}
         \textbf{Optimizing with minimum length 1.2 Delaunay length} (Sum aggregation)
    \end{tabular}
    \begin{tabu} to \textwidth {X[1,c,m]|X[1,c,m]|X[1,c,m]|X[1,c,m]|X[1,c,m]|X[1,c,m]|X[1,c,m]|X[1,c,m]}
            \toprule
            DT & Opposing Angles & Dual Edge Ratio & Dual Area overlap & Lens & Shrunk Circle & Triangular Lens & Shrunk Circumcircle\\
            \midrule
            \includegraphics[page=1, width=1\linewidth]{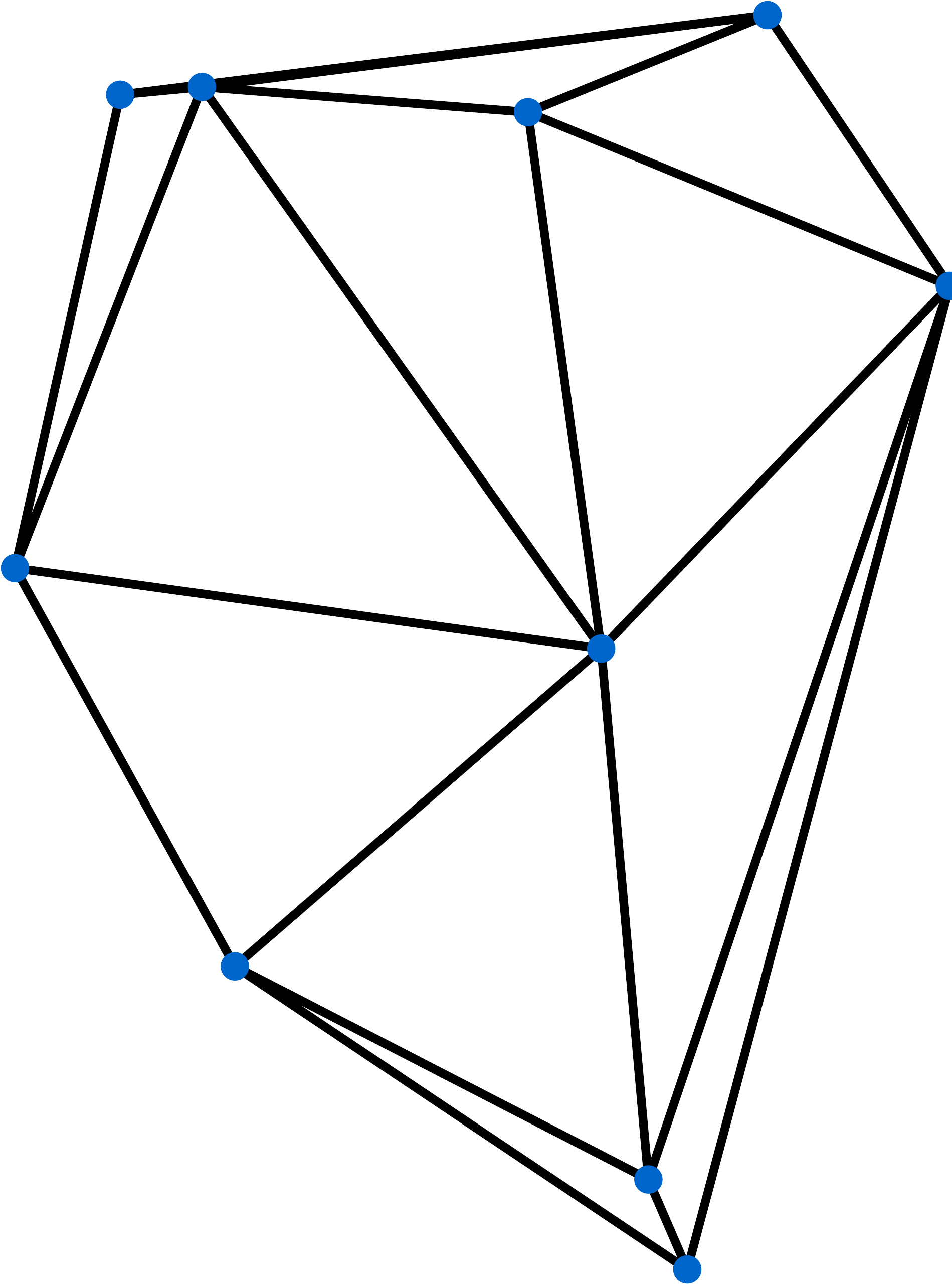} &
            \includegraphics[page=2, width=1\linewidth]{figures/minLength0sum.pdf} &
            \includegraphics[page=3, width=1\linewidth]{figures/minLength0sum.pdf} &
            \includegraphics[page=4, width=1\linewidth]{figures/minLength0sum.pdf} &
            \includegraphics[page=5, width=1\linewidth]{figures/minLength0sum.pdf} &
            \includegraphics[page=6, width=1\linewidth]{figures/minLength0sum.pdf} &
            \includegraphics[page=7, width=1\linewidth]{figures/minLength0sum.pdf} &
            \includegraphics[page=8, width=1\linewidth]{figures/minLength0sum.pdf}\\
            \bottomrule
    \end{tabu}
    \begin{tabular}{c}
         \textbf{Optimizing with maximum length 0.8 Delaunay length} (Bottleneck aggregation)
    \end{tabular}
    \begin{tabu} to \textwidth {X[1,c,m]|X[1,c,m]|X[1,c,m]|X[1,c,m]|X[1,c,m]|X[1,c,m]|X[1,c,m]|X[1,c,m]}
            \toprule
            DT & Opposing Angles & Dual Edge Ratio & Dual Area overlap & Lens & Shrunk Circle & Triangular Lens & Shrunk Circumcircle\\
            \midrule
            \includegraphics[page=1, width=1\linewidth]{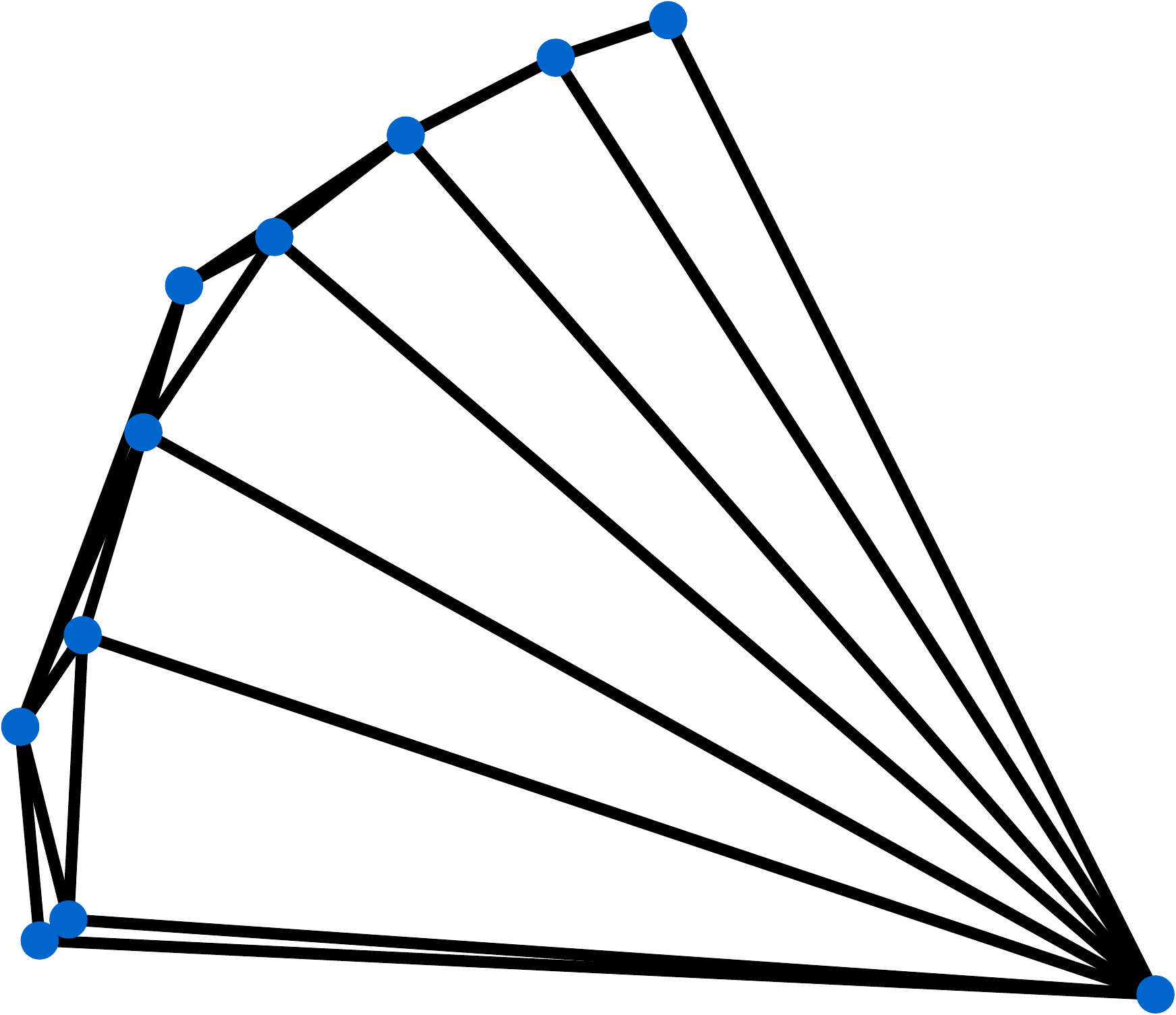} &
            \includegraphics[page=2, width=1\linewidth]{figures/maxLength8max.pdf} &
            \includegraphics[page=3, width=1\linewidth]{figures/maxLength8max.pdf} &
            \includegraphics[page=4, width=1\linewidth]{figures/maxLength8max.pdf} &
            \includegraphics[page=5, width=1\linewidth]{figures/maxLength8max.pdf} &
            \includegraphics[page=6, width=1\linewidth]{figures/maxLength8max.pdf} &
            \includegraphics[page=7, width=1\linewidth]{figures/maxLength8max.pdf} &
            \includegraphics[page=8, width=1\linewidth]{figures/maxLength8max.pdf}\\
            \bottomrule
    \end{tabu}
    \begin{tabular}{c}
        \textbf{Optimizing with maximum degree 5} (Sum aggregation)
    \end{tabular}
    \begin{tabu} to \textwidth {X[1,c,m]|X[1,c,m]|X[1,c,m]|X[1,c,m]|X[1,c,m]|X[1,c,m]|X[1,c,m]|X[1,c,m]}
            \toprule
            DT & Opposing Angles & Dual Edge Ratio & Dual Area overlap & Lens & Shrunk Circle & Triangular Lens & Shrunk Circumcircle\\
            \midrule
            \includegraphics[page=1, width=1\linewidth]{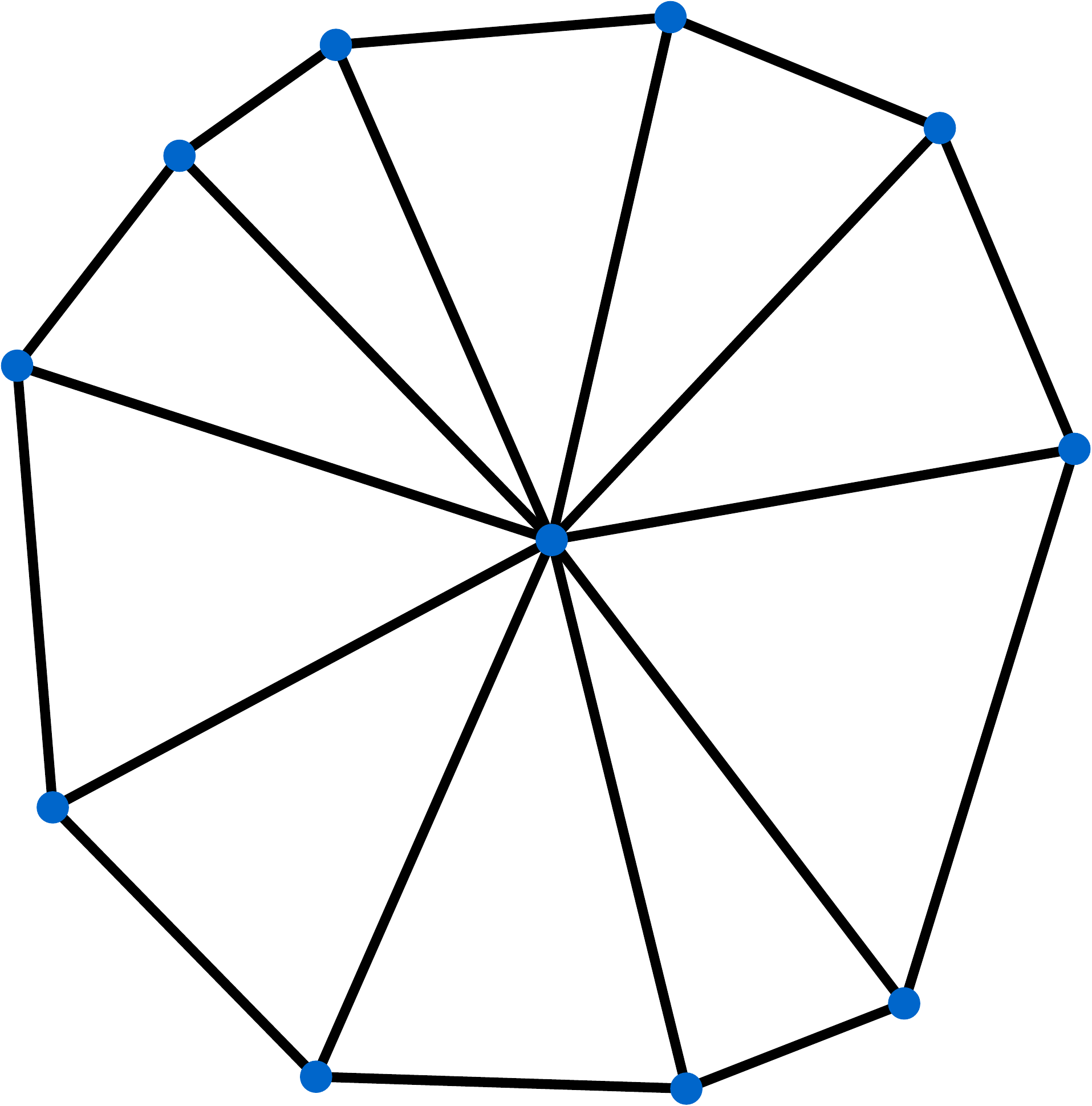} &
            \includegraphics[page=2, width=1\linewidth]{figures/maxDegree9sum.pdf} &
            \includegraphics[page=3, width=1\linewidth]{figures/maxDegree9sum.pdf} &
            \includegraphics[page=4, width=1\linewidth]{figures/maxDegree9sum.pdf} &
            \includegraphics[page=5, width=1\linewidth]{figures/maxDegree9sum.pdf} &
            \includegraphics[page=6, width=1\linewidth]{figures/maxDegree9sum.pdf} &
            \includegraphics[page=7, width=1\linewidth]{figures/maxDegree9sum.pdf} &
            \includegraphics[page=8, width=1\linewidth]{figures/maxDegree9sum.pdf}\\
            \bottomrule
    \end{tabu}
    \begin{tabular}{c}
         \textbf{Optimizing with maximum degree 5} (Bottleneck aggregation)
    \end{tabular}
    \begin{tabu} to \textwidth {X[1,c,m]|X[1,c,m]|X[1,c,m]|X[1,c,m]|X[1,c,m]|X[1,c,m]|X[1,c,m]|X[1,c,m]}
            \toprule
            DT & Opposing Angles & Dual Edge Ratio & Dual Area overlap & Lens & Shrunk Circle & Triangular Lens & Shrunk Circumcircle\\
            \midrule
            \includegraphics[page=1, width=1\linewidth]{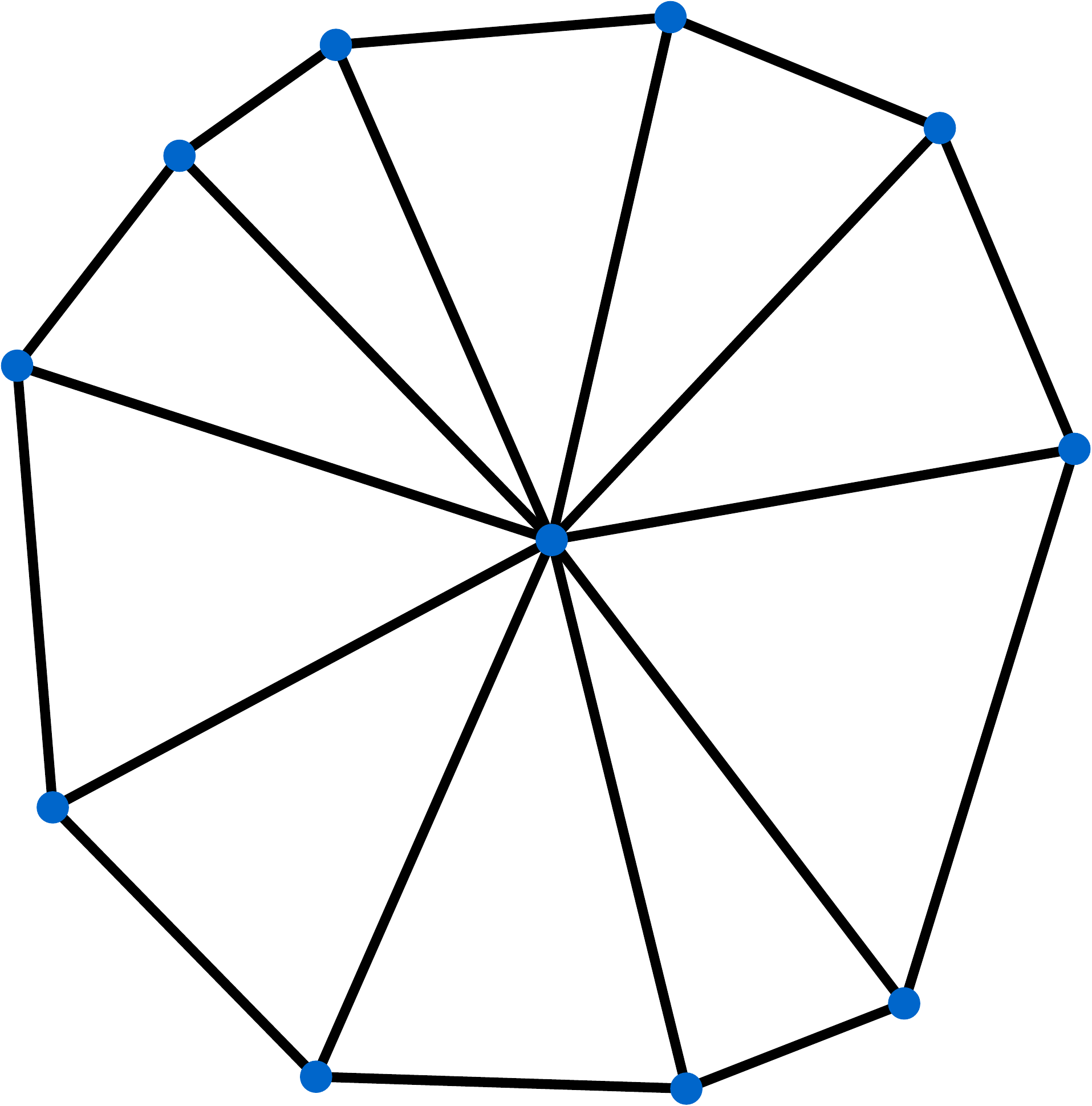} &
            \includegraphics[page=2, width=1\linewidth]{figures/maxDegree9max.pdf} &
            \includegraphics[page=3, width=1\linewidth]{figures/maxDegree9max.pdf} &
            \includegraphics[page=4, width=1\linewidth]{figures/maxDegree9max.pdf} &
            \includegraphics[page=5, width=1\linewidth]{figures/maxDegree9max.pdf} &
            \includegraphics[page=6, width=1\linewidth]{figures/maxDegree9max.pdf} &
            \includegraphics[page=7, width=1\linewidth]{figures/maxDegree9max.pdf} &
            \includegraphics[page=8, width=1\linewidth]{figures/maxDegree9max.pdf}\\
            \bottomrule
    \end{tabu}
\end{table*}

\begin{table*}[t]
\caption{\hspace{-14em} \textbf{\textsf{Appendix}} \hspace{9em} Optimizing with constrained edges (Sum aggregation)}
\label{tab:cdt_sum}
\small
\begin{tabu} to \textwidth {X[1,c,m]|X[1,c,m]|X[1,c,m]|X[1,c,m]|X[1,c,m]|X[1,c,m]|X[1,c,m]|X[1,c,m]}
        \toprule
        CDT & Opposing Angles & Dual Edge Ratio & Dual Area overlap & Lens & Shrunk Circle & Triangular Lens & Shrunk Circumcircle\\
        \midrule
        \includegraphics[page=1, width=1\linewidth]{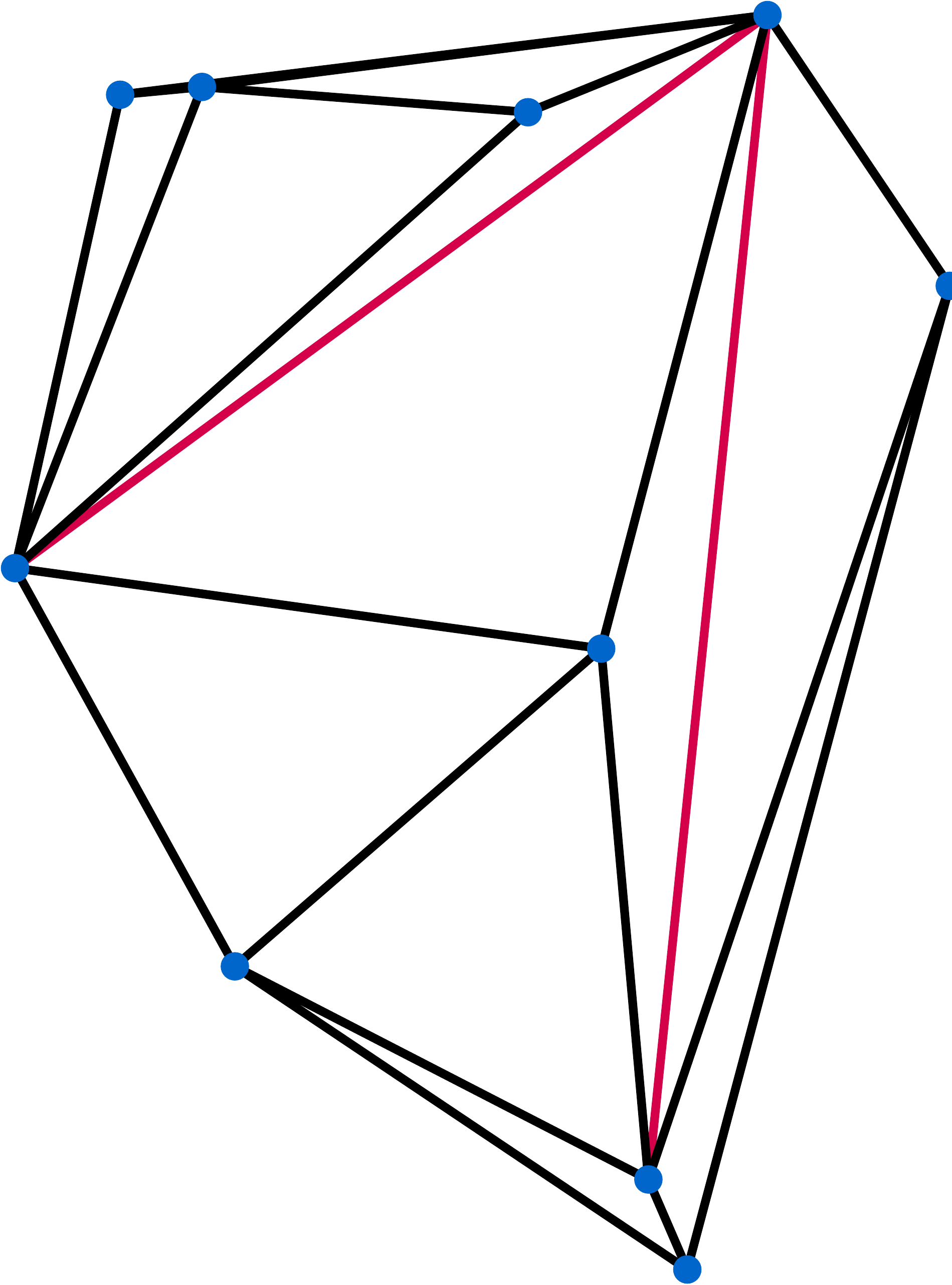} &
        \includegraphics[page=2, width=1\linewidth]{figures/cdt0sum.pdf} &
        \includegraphics[page=3, width=1\linewidth]{figures/cdt0sum.pdf} &
        \includegraphics[page=4, width=1\linewidth]{figures/cdt0sum.pdf} &
        \includegraphics[page=5, width=1\linewidth]{figures/cdt0sum.pdf} &
        \includegraphics[page=6, width=1\linewidth]{figures/cdt0sum.pdf} &
        \includegraphics[page=7, width=1\linewidth]{figures/cdt0sum.pdf} &
        \includegraphics[page=8, width=1\linewidth]{figures/cdt0sum.pdf}\\\midrule
        \includegraphics[page=1, width=1\linewidth]{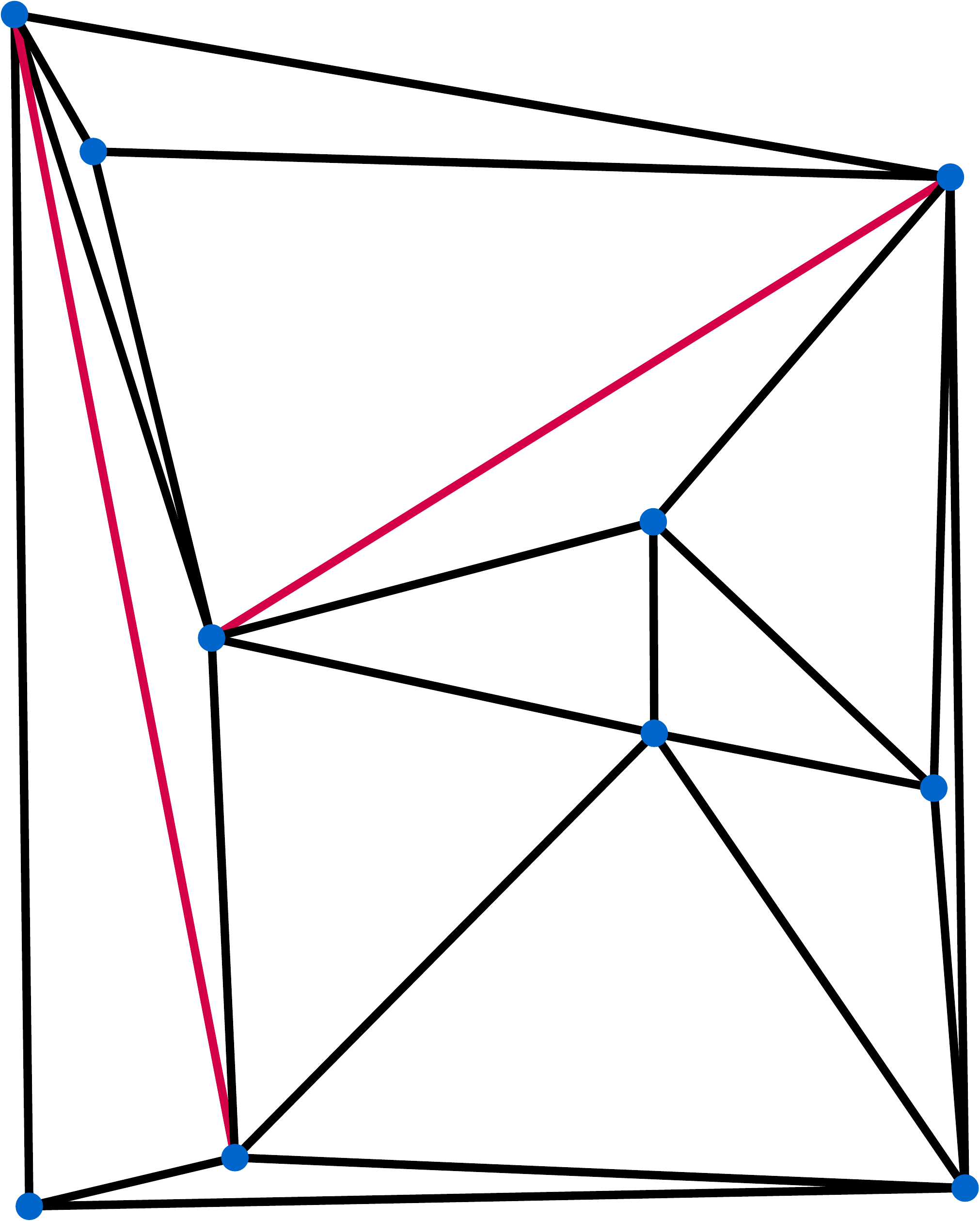} &
        \includegraphics[page=2, width=1\linewidth]{figures/cdt1sum.pdf} &
        \includegraphics[page=3, width=1\linewidth]{figures/cdt1sum.pdf} &
        \includegraphics[page=4, width=1\linewidth]{figures/cdt1sum.pdf} &
        \includegraphics[page=5, width=1\linewidth]{figures/cdt1sum.pdf} &
        \includegraphics[page=6, width=1\linewidth]{figures/cdt1sum.pdf} &
        \includegraphics[page=7, width=1\linewidth]{figures/cdt1sum.pdf} &
        \includegraphics[page=8, width=1\linewidth]{figures/cdt1sum.pdf}\\\midrule
        \includegraphics[page=1, width=1\linewidth]{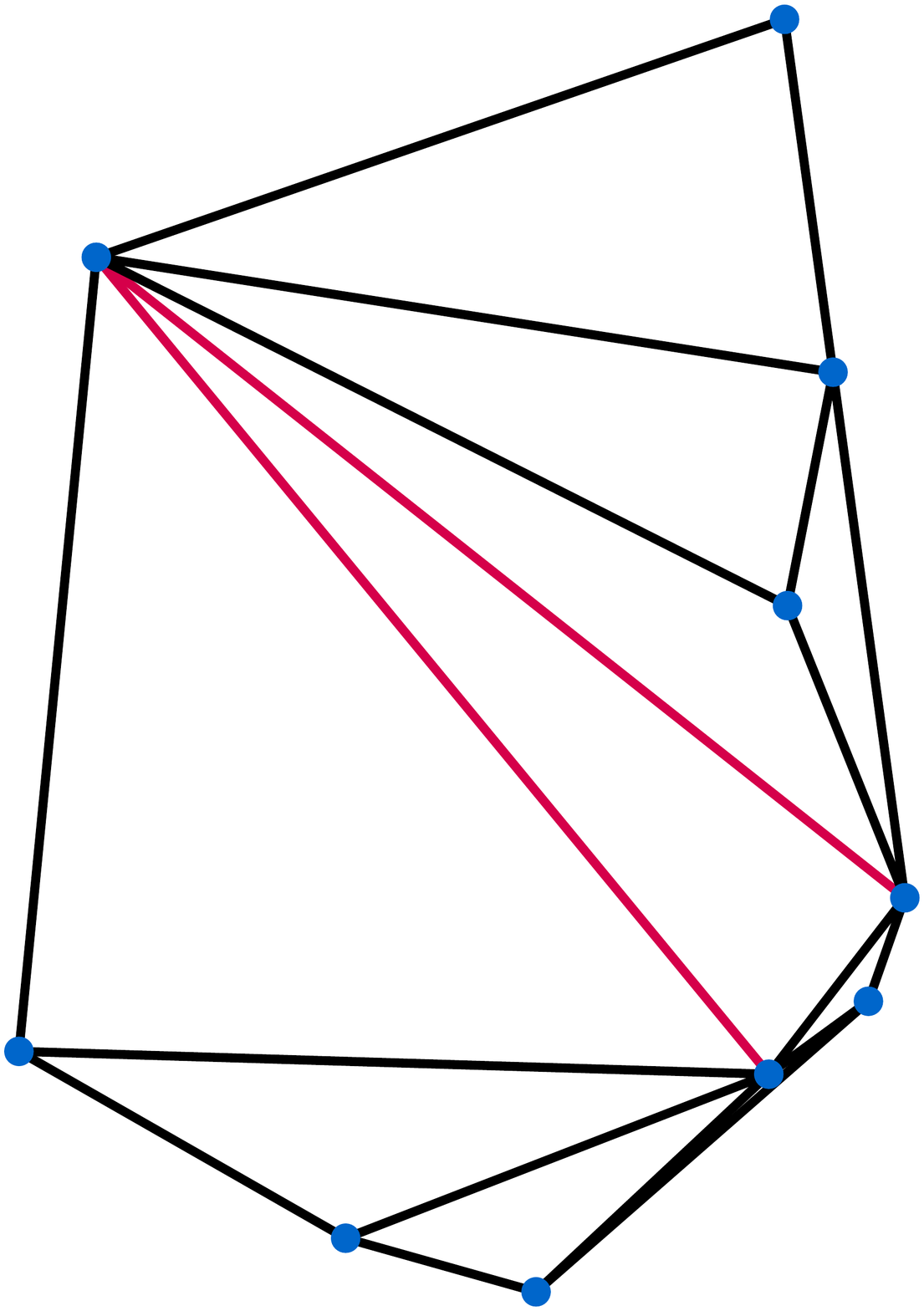} &
        \includegraphics[page=2, width=1\linewidth]{figures/cdt2sum.pdf} &
        \includegraphics[page=3, width=1\linewidth]{figures/cdt2sum.pdf} &
        \includegraphics[page=4, width=1\linewidth]{figures/cdt2sum.pdf} &
        \includegraphics[page=5, width=1\linewidth]{figures/cdt2sum.pdf} &
        \includegraphics[page=6, width=1\linewidth]{figures/cdt2sum.pdf} &
        \includegraphics[page=7, width=1\linewidth]{figures/cdt2sum.pdf} &
        \includegraphics[page=8, width=1\linewidth]{figures/cdt2sum.pdf}\\\midrule
        \includegraphics[page=1, width=1\linewidth]{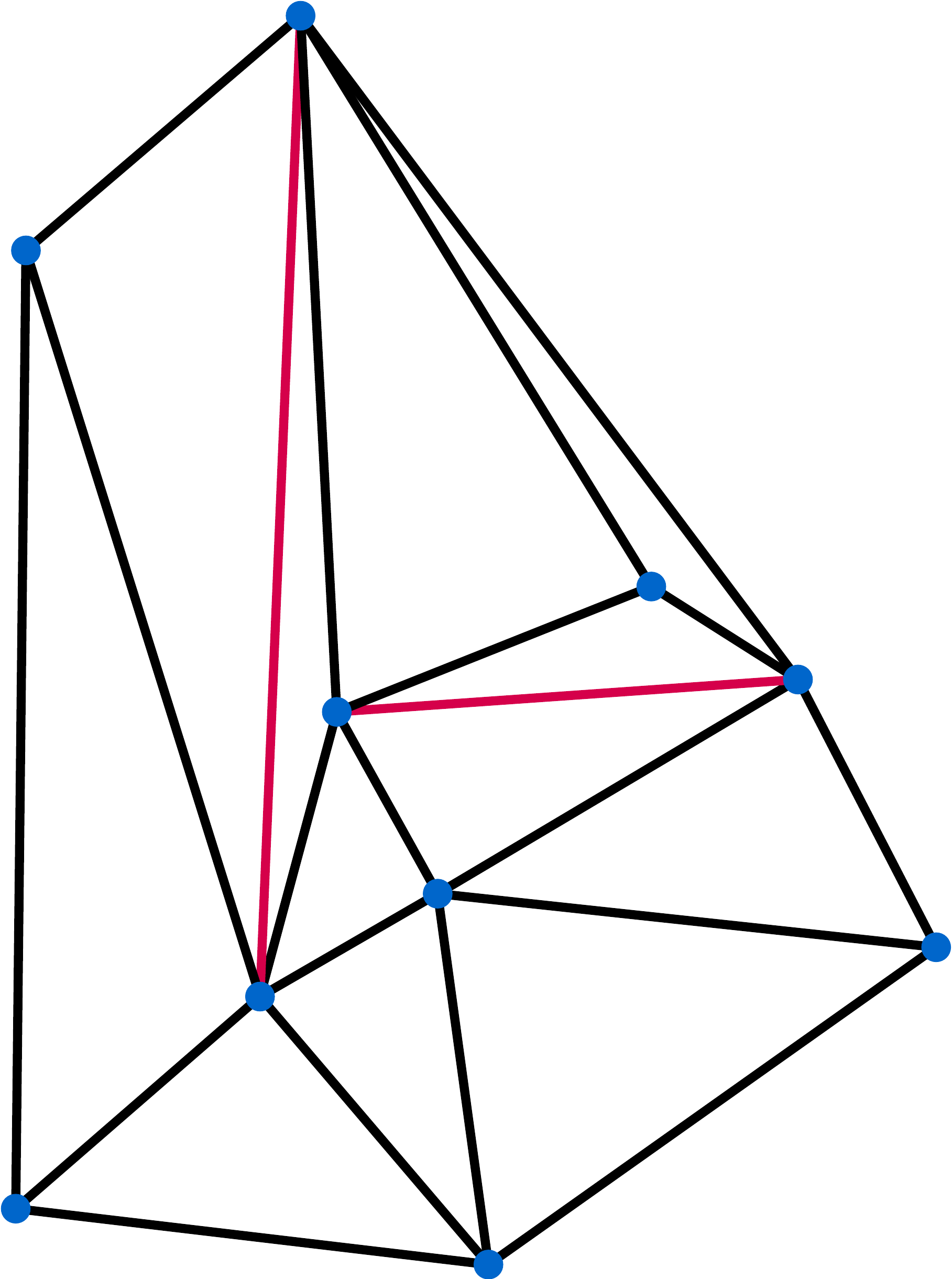} &
        \includegraphics[page=2, width=1\linewidth]{figures/cdt3sum.pdf} &
        \includegraphics[page=3, width=1\linewidth]{figures/cdt3sum.pdf} &
        \includegraphics[page=4, width=1\linewidth]{figures/cdt3sum.pdf} &
        \includegraphics[page=5, width=1\linewidth]{figures/cdt3sum.pdf} &
        \includegraphics[page=6, width=1\linewidth]{figures/cdt3sum.pdf} &
        \includegraphics[page=7, width=1\linewidth]{figures/cdt3sum.pdf} &
        \includegraphics[page=8, width=1\linewidth]{figures/cdt3sum.pdf}\\\midrule
        \includegraphics[page=1, width=1\linewidth]{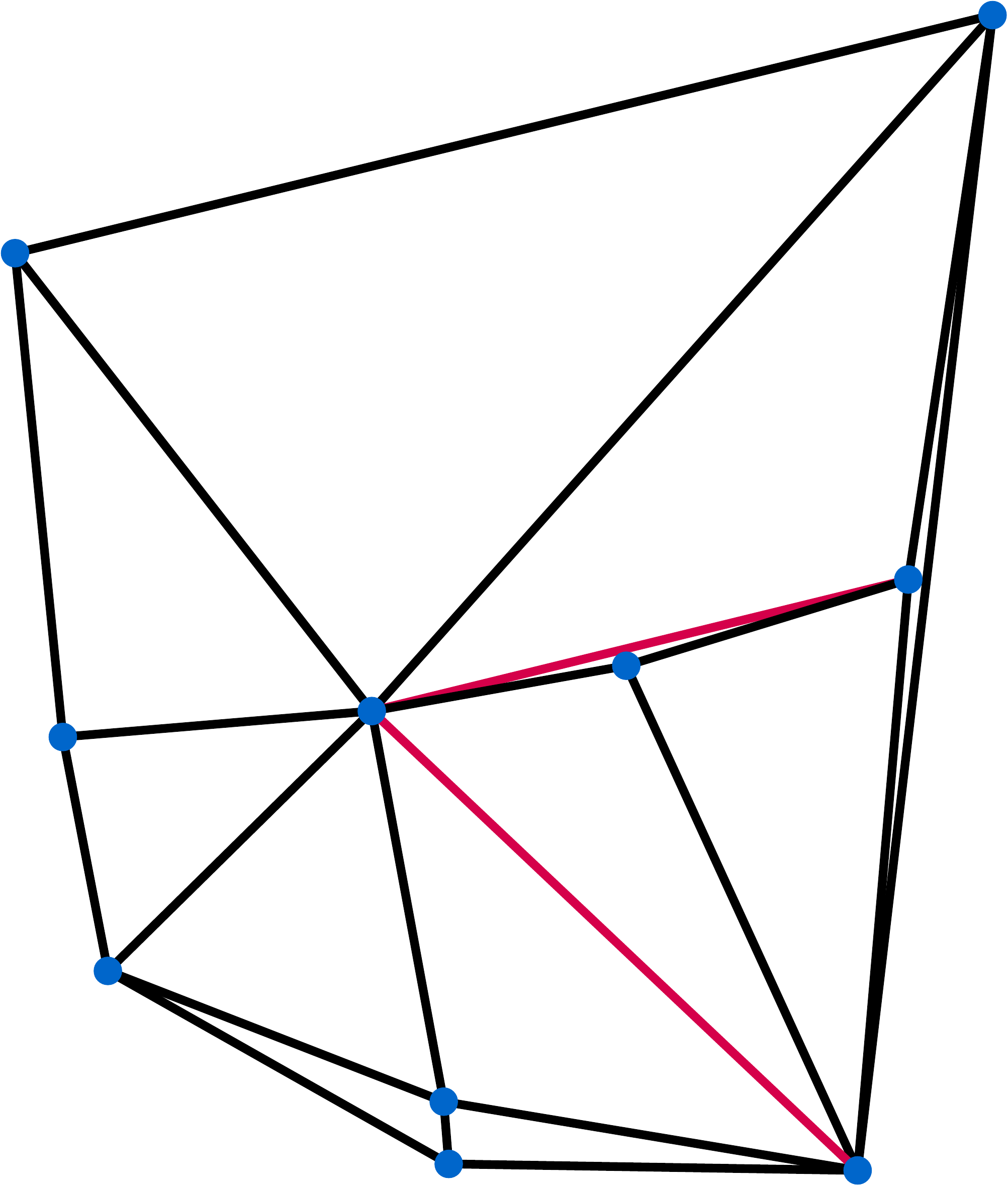} &
        \includegraphics[page=2, width=1\linewidth]{figures/cdt4sum.pdf} &
        \includegraphics[page=3, width=1\linewidth]{figures/cdt4sum.pdf} &
        \includegraphics[page=4, width=1\linewidth]{figures/cdt4sum.pdf} &
        \includegraphics[page=5, width=1\linewidth]{figures/cdt4sum.pdf} &
        \includegraphics[page=6, width=1\linewidth]{figures/cdt4sum.pdf} &
        \includegraphics[page=7, width=1\linewidth]{figures/cdt4sum.pdf} &
        \includegraphics[page=8, width=1\linewidth]{figures/cdt4sum.pdf}\\\midrule
        \includegraphics[page=1, width=1\linewidth]{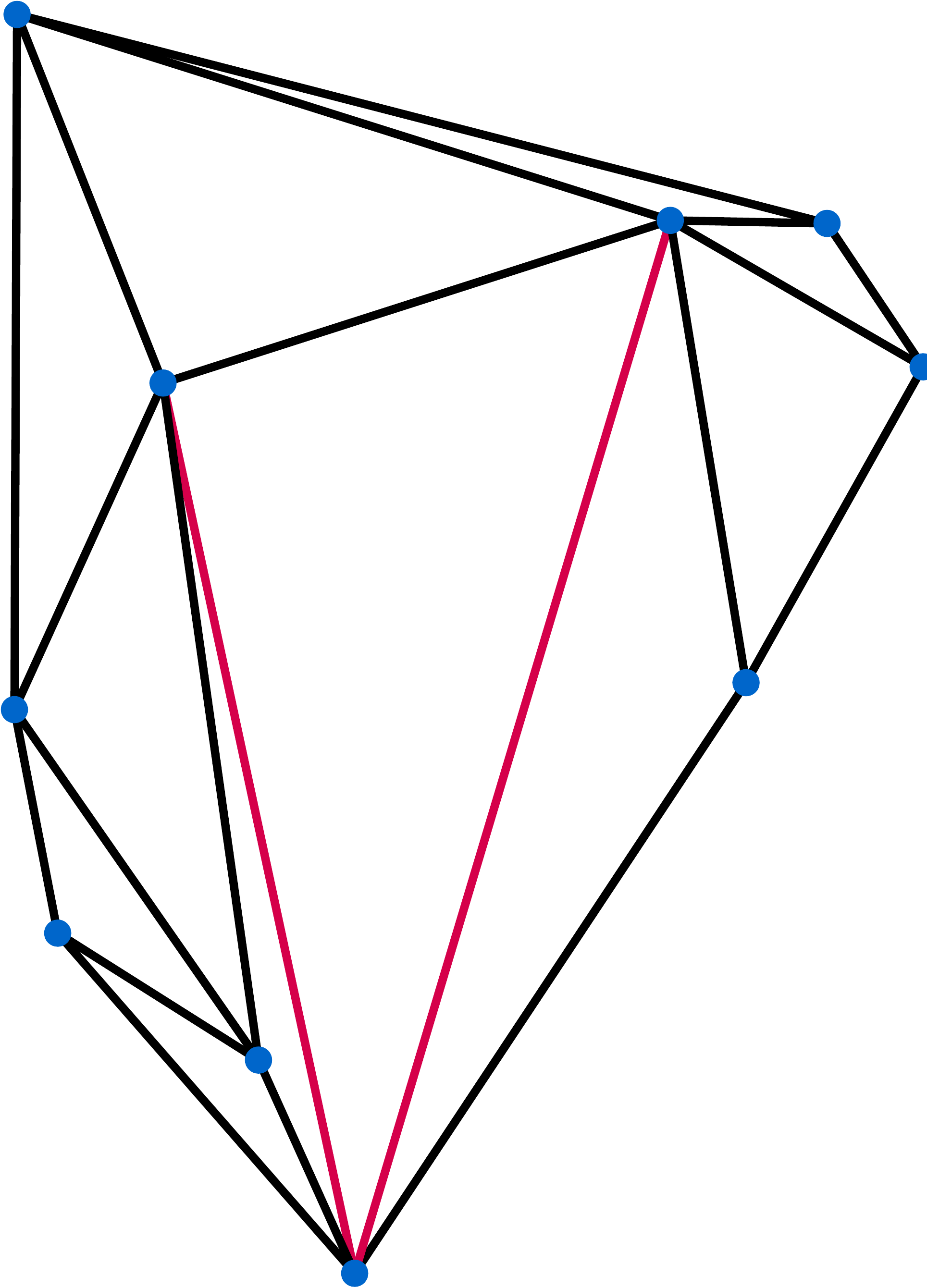} &
        \includegraphics[page=2, width=1\linewidth]{figures/cdt5sum.pdf} &
        \includegraphics[page=3, width=1\linewidth]{figures/cdt5sum.pdf} &
        \includegraphics[page=4, width=1\linewidth]{figures/cdt5sum.pdf} &
        \includegraphics[page=5, width=1\linewidth]{figures/cdt5sum.pdf} &
        \includegraphics[page=6, width=1\linewidth]{figures/cdt5sum.pdf} &
        \includegraphics[page=7, width=1\linewidth]{figures/cdt5sum.pdf} &
        \includegraphics[page=8, width=1\linewidth]{figures/cdt5sum.pdf}\\\midrule
        \includegraphics[page=1, width=1\linewidth]{figures/cdt6sum.pdf} &
        \includegraphics[page=2, width=1\linewidth]{figures/cdt6sum.pdf} &
        \includegraphics[page=3, width=1\linewidth]{figures/cdt6sum.pdf} &
        \includegraphics[page=4, width=1\linewidth]{figures/cdt6sum.pdf} &
        \includegraphics[page=5, width=1\linewidth]{figures/cdt6sum.pdf} &
        \includegraphics[page=6, width=1\linewidth]{figures/cdt6sum.pdf} &
        \includegraphics[page=7, width=1\linewidth]{figures/cdt6sum.pdf} &
        \includegraphics[page=8, width=1\linewidth]{figures/cdt6sum.pdf}\\\midrule
        \includegraphics[page=1, width=1\linewidth]{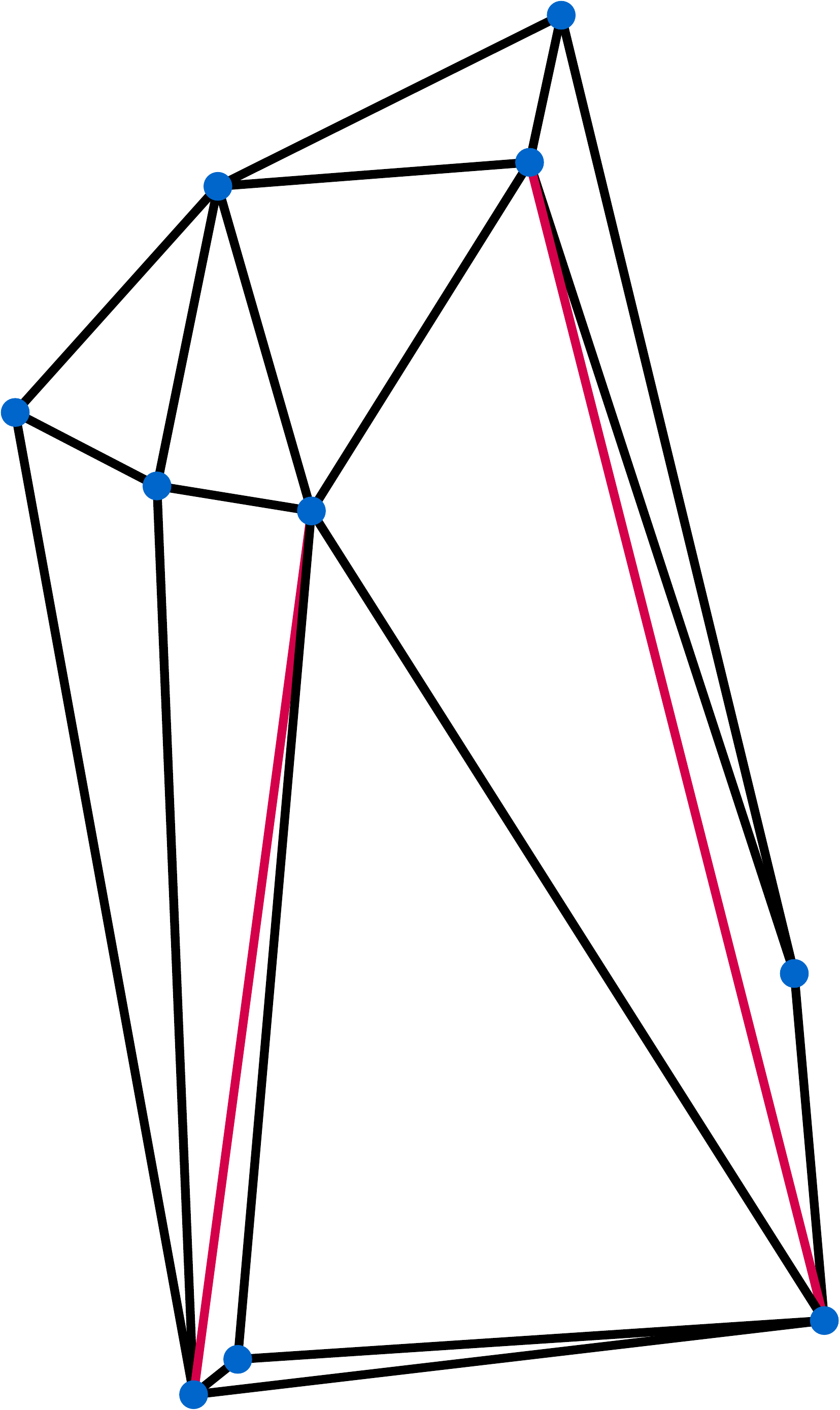} &
        \includegraphics[page=2, width=1\linewidth]{figures/cdt7sum.pdf} &
        \includegraphics[page=3, width=1\linewidth]{figures/cdt7sum.pdf} &
        \includegraphics[page=4, width=1\linewidth]{figures/cdt7sum.pdf} &
        \includegraphics[page=5, width=1\linewidth]{figures/cdt7sum.pdf} &
        \includegraphics[page=6, width=1\linewidth]{figures/cdt7sum.pdf} &
        \includegraphics[page=7, width=1\linewidth]{figures/cdt7sum.pdf} &
        \includegraphics[page=8, width=1\linewidth]{figures/cdt7sum.pdf}\\
        \bottomrule
\end{tabu}
\end{table*}

\begin{table*}[t]
\caption{Optimizing with constrained edges (Bottleneck aggregation)}
\label{tab:cdt_bottleneck}
\small
\begin{tabu} to \textwidth {X[1,c,m]|X[1,c,m]|X[1,c,m]|X[1,c,m]|X[1,c,m]|X[1,c,m]|X[1,c,m]|X[1,c,m]}
        \toprule
        CDT & Opposing Angles & Dual Edge Ratio & Dual Area overlap & Lens & Shrunk Circle & Triangular Lens & Shrunk Circumcircle\\
        \midrule
        \includegraphics[page=1, width=1\linewidth]{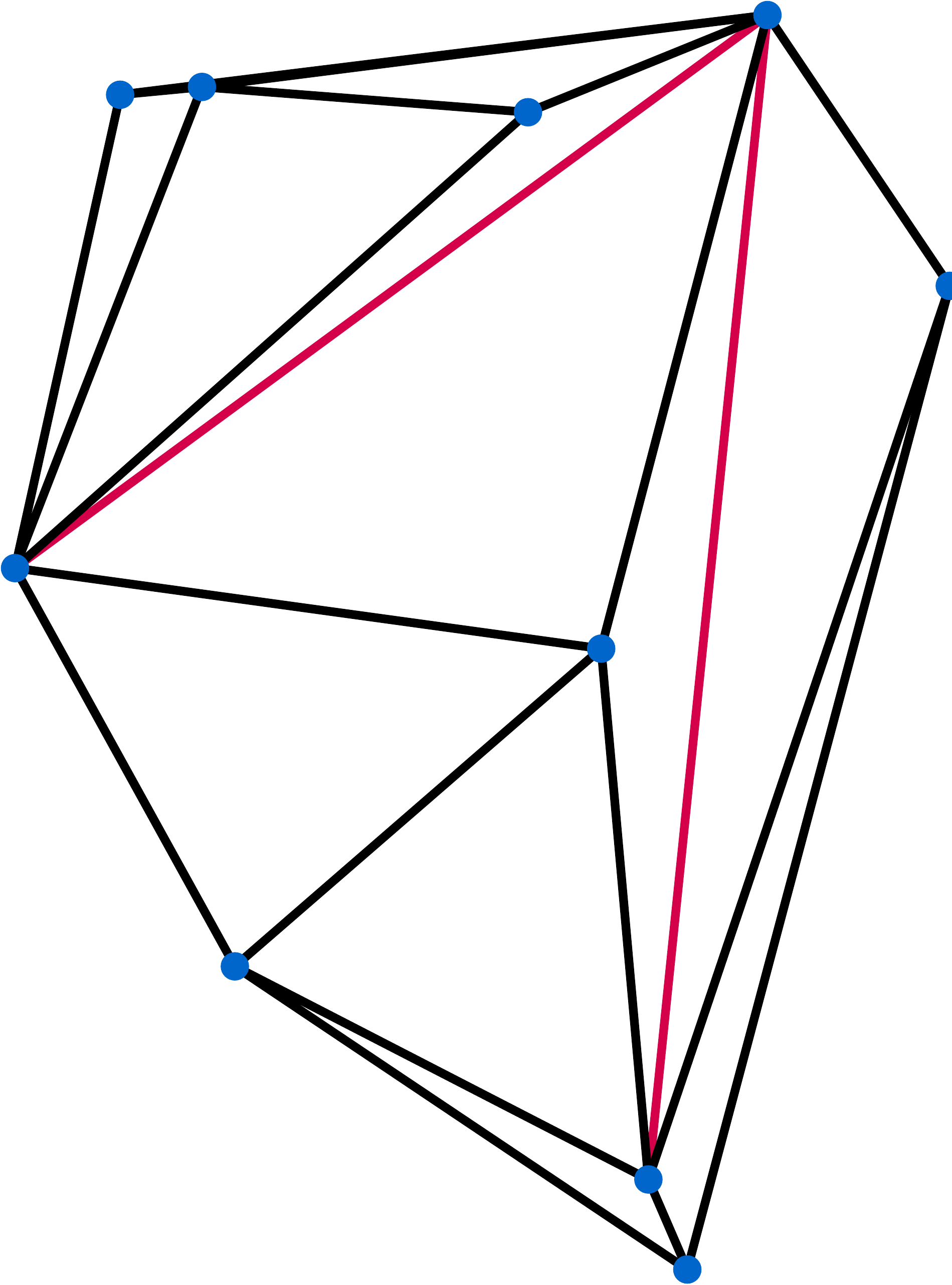} &
        \includegraphics[page=2, width=1\linewidth]{figures/cdt0max.pdf} &
        \includegraphics[page=3, width=1\linewidth]{figures/cdt0max.pdf} &
        \includegraphics[page=4, width=1\linewidth]{figures/cdt0max.pdf} &
        \includegraphics[page=5, width=1\linewidth]{figures/cdt0max.pdf} &
        \includegraphics[page=6, width=1\linewidth]{figures/cdt0max.pdf} &
        \includegraphics[page=7, width=1\linewidth]{figures/cdt0max.pdf} &
        \includegraphics[page=8, width=1\linewidth]{figures/cdt0max.pdf}\\\midrule
        \includegraphics[page=1, width=1\linewidth]{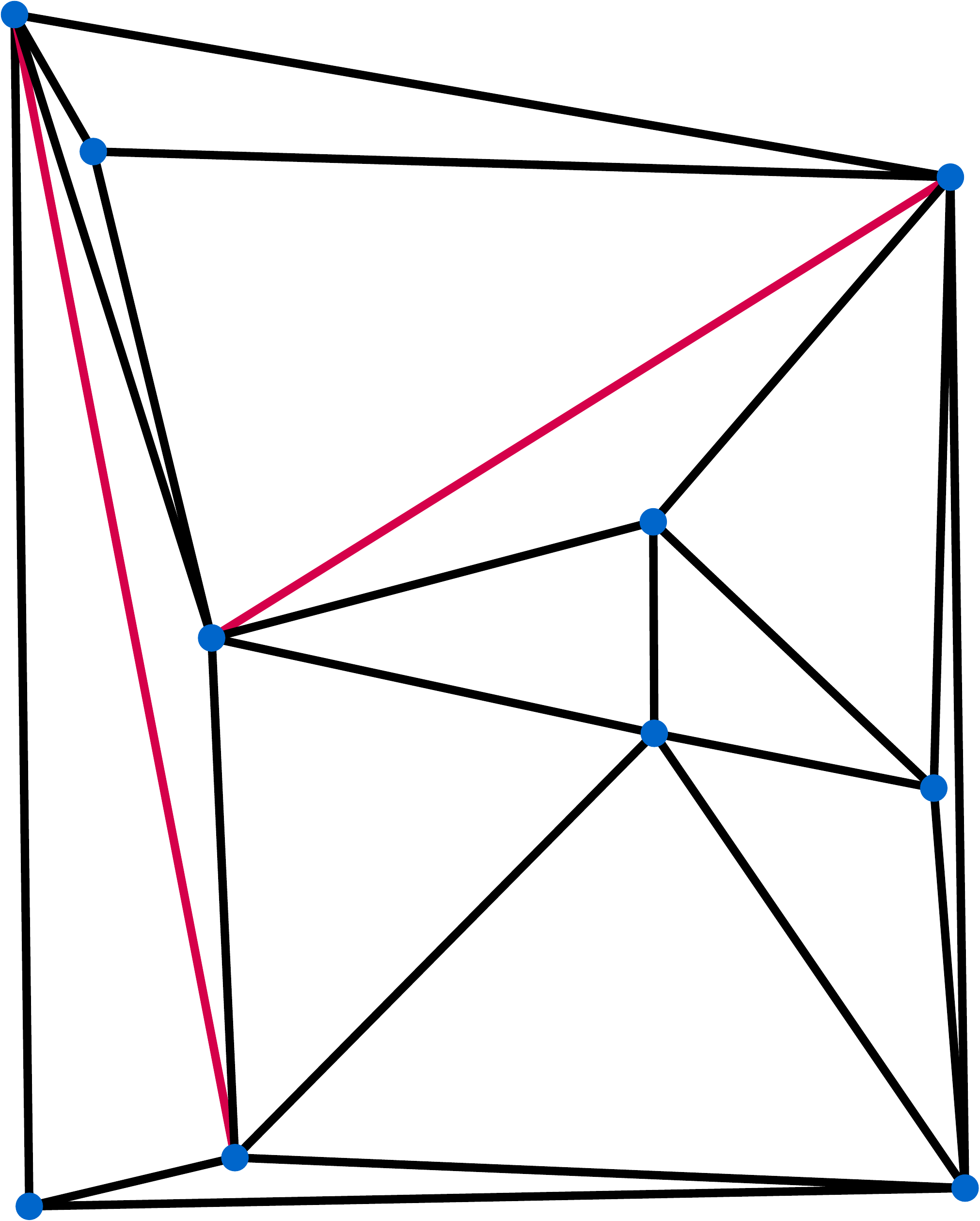} &
        \includegraphics[page=2, width=1\linewidth]{figures/cdt1max.pdf} &
        \includegraphics[page=3, width=1\linewidth]{figures/cdt1max.pdf} &
        \includegraphics[page=4, width=1\linewidth]{figures/cdt1max.pdf} &
        \includegraphics[page=5, width=1\linewidth]{figures/cdt1max.pdf} &
        \includegraphics[page=6, width=1\linewidth]{figures/cdt1max.pdf} &
        \includegraphics[page=7, width=1\linewidth]{figures/cdt1max.pdf} &
        \includegraphics[page=8, width=1\linewidth]{figures/cdt1max.pdf}\\\midrule
        \includegraphics[page=1, width=1\linewidth]{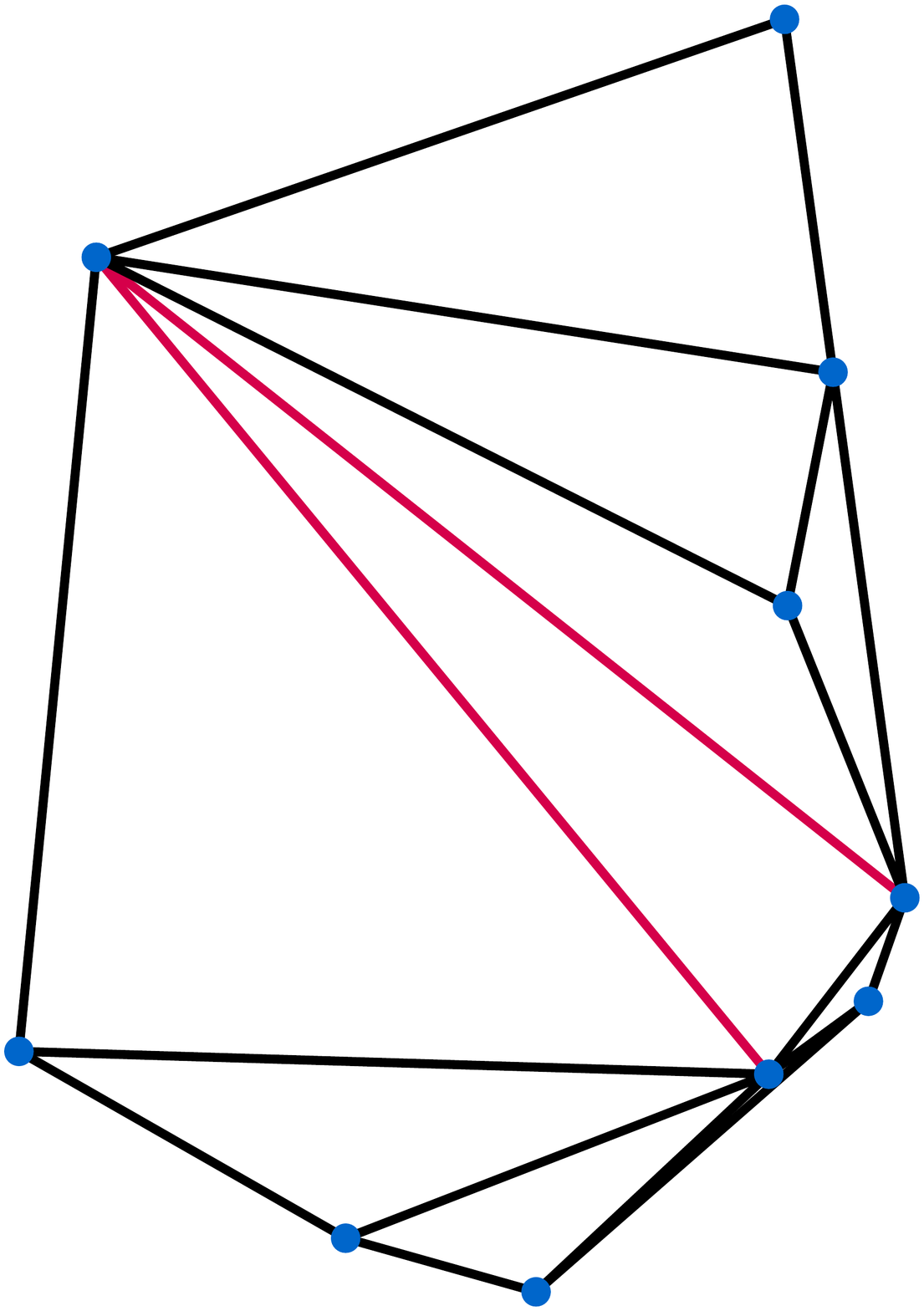} &
        \includegraphics[page=2, width=1\linewidth]{figures/cdt2max.pdf} &
        \includegraphics[page=3, width=1\linewidth]{figures/cdt2max.pdf} &
        \includegraphics[page=4, width=1\linewidth]{figures/cdt2max.pdf} &
        \includegraphics[page=5, width=1\linewidth]{figures/cdt2max.pdf} &
        \includegraphics[page=6, width=1\linewidth]{figures/cdt2max.pdf} &
        \includegraphics[page=7, width=1\linewidth]{figures/cdt2max.pdf} &
        \includegraphics[page=8, width=1\linewidth]{figures/cdt2max.pdf}\\\midrule
        \includegraphics[page=1, width=1\linewidth]{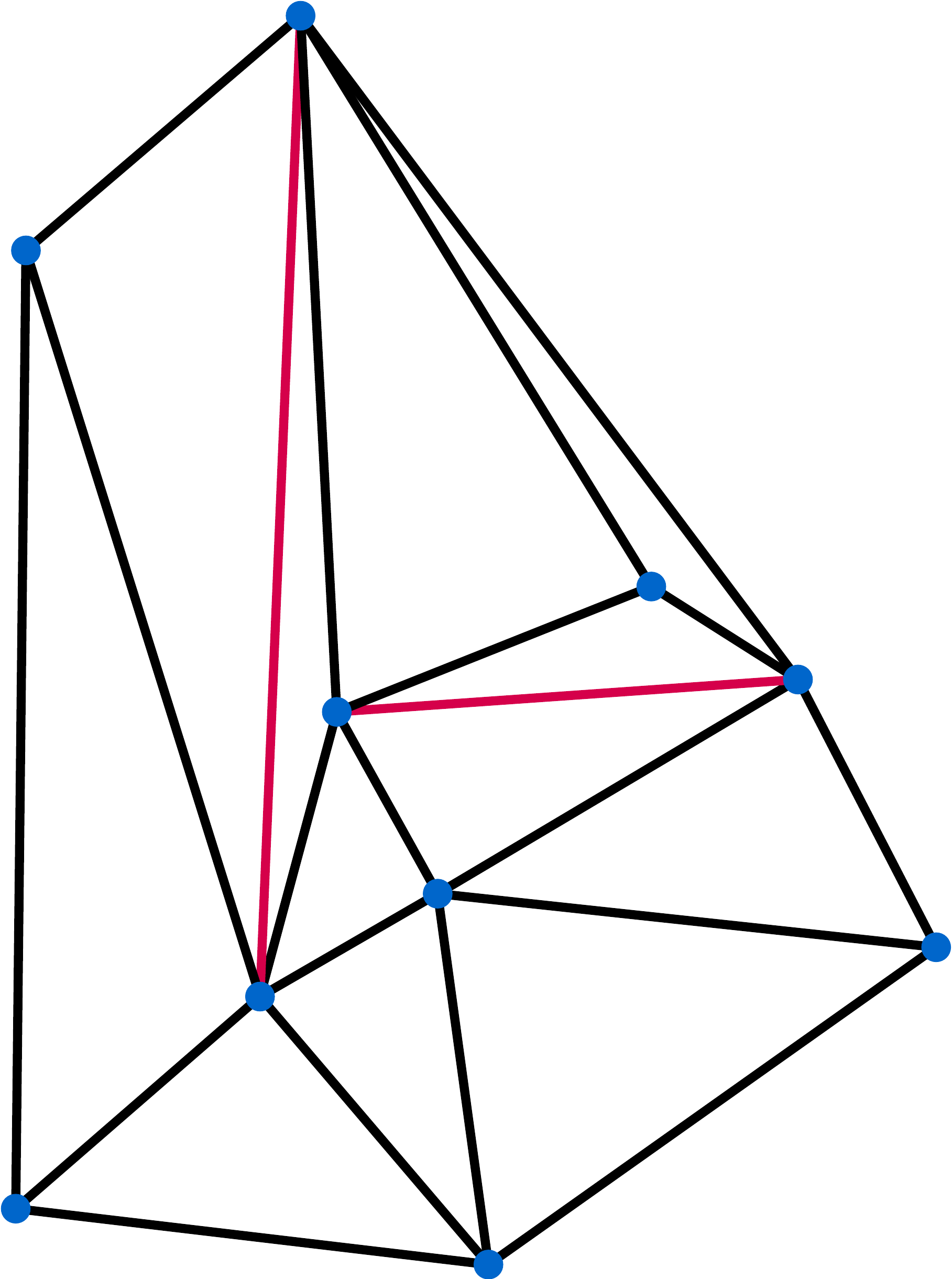} &
        \includegraphics[page=2, width=1\linewidth]{figures/cdt3max.pdf} &
        \includegraphics[page=3, width=1\linewidth]{figures/cdt3max.pdf} &
        \includegraphics[page=4, width=1\linewidth]{figures/cdt3max.pdf} &
        \includegraphics[page=5, width=1\linewidth]{figures/cdt3max.pdf} &
        \includegraphics[page=6, width=1\linewidth]{figures/cdt3max.pdf} &
        \includegraphics[page=7, width=1\linewidth]{figures/cdt3max.pdf} &
        \includegraphics[page=8, width=1\linewidth]{figures/cdt3max.pdf}\\\midrule
        \includegraphics[page=1, width=1\linewidth]{figures/cdt4max.pdf} &
        \includegraphics[page=2, width=1\linewidth]{figures/cdt4max.pdf} &
        \includegraphics[page=3, width=1\linewidth]{figures/cdt4max.pdf} &
        \includegraphics[page=4, width=1\linewidth]{figures/cdt4max.pdf} &
        \includegraphics[page=5, width=1\linewidth]{figures/cdt4max.pdf} &
        \includegraphics[page=6, width=1\linewidth]{figures/cdt4max.pdf} &
        \includegraphics[page=7, width=1\linewidth]{figures/cdt4max.pdf} &
        \includegraphics[page=8, width=1\linewidth]{figures/cdt4max.pdf}\\\midrule
        \includegraphics[page=1, width=1\linewidth]{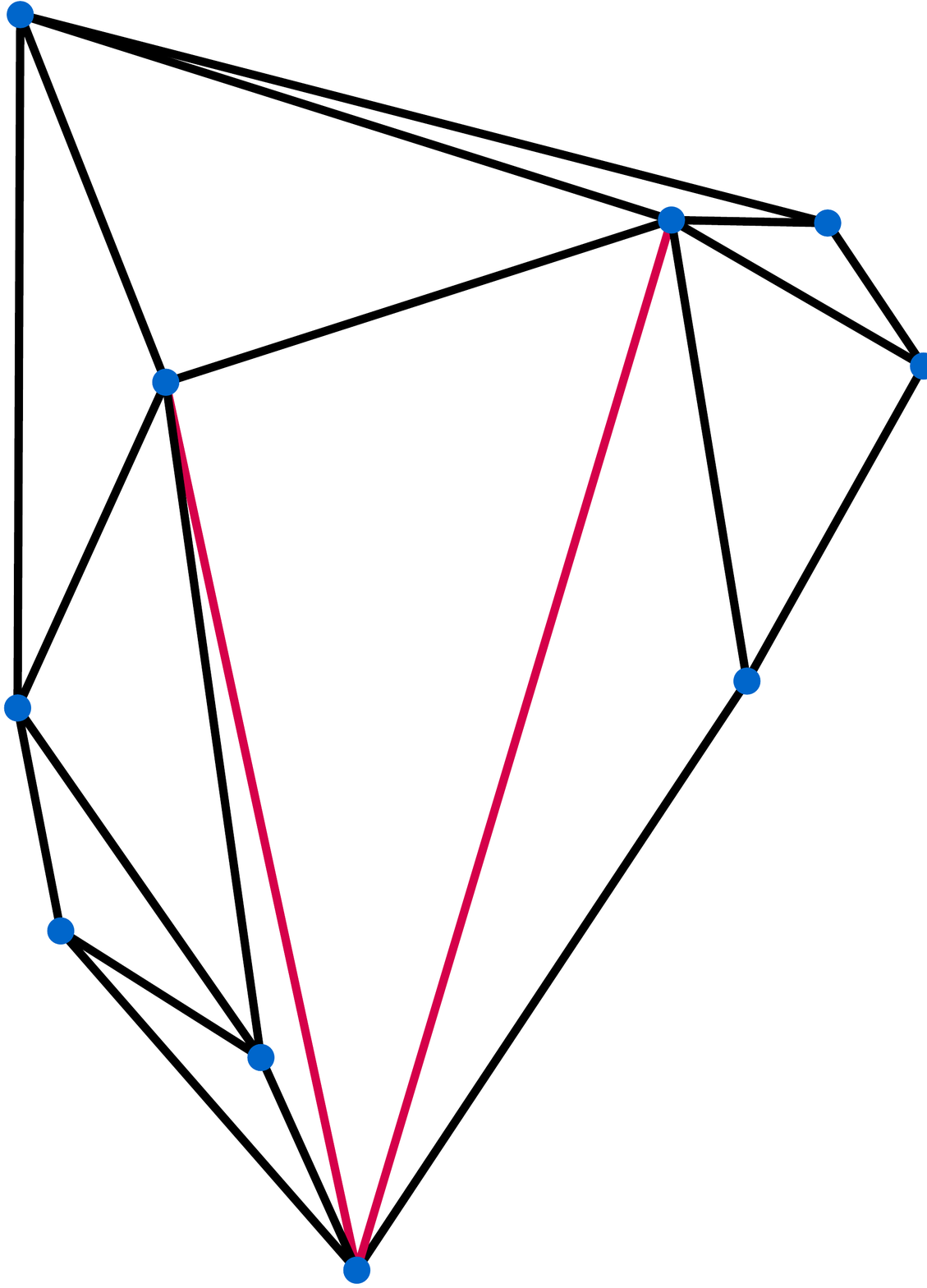} &
        \includegraphics[page=2, width=1\linewidth]{figures/cdt5max.pdf} &
        \includegraphics[page=3, width=1\linewidth]{figures/cdt5max.pdf} &
        \includegraphics[page=4, width=1\linewidth]{figures/cdt5max.pdf} &
        \includegraphics[page=5, width=1\linewidth]{figures/cdt5max.pdf} &
        \includegraphics[page=6, width=1\linewidth]{figures/cdt5max.pdf} &
        \includegraphics[page=7, width=1\linewidth]{figures/cdt5max.pdf} &
        \includegraphics[page=8, width=1\linewidth]{figures/cdt5max.pdf}\\\midrule
        \includegraphics[page=1, width=1\linewidth]{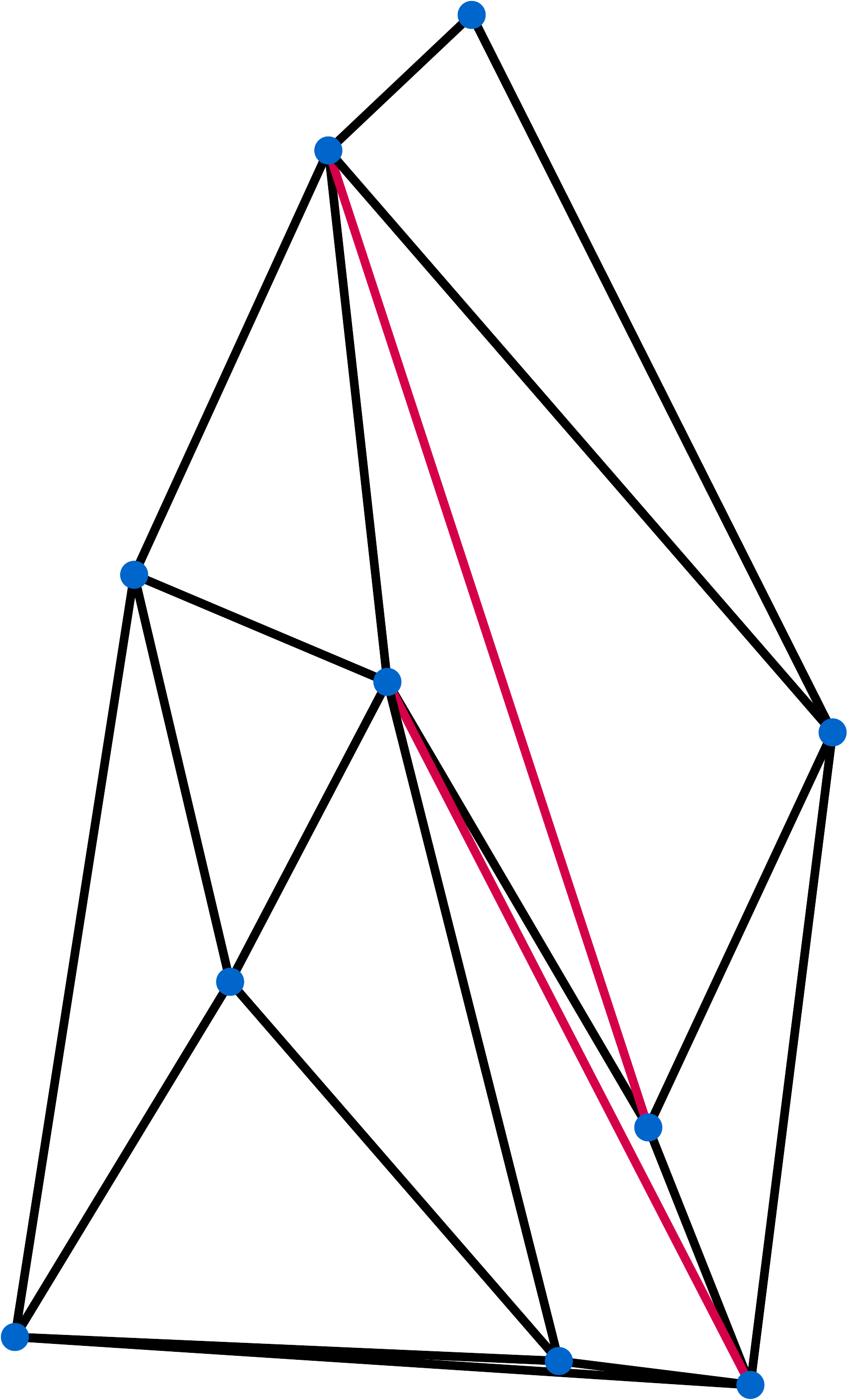} &
        \includegraphics[page=2, width=1\linewidth]{figures/cdt6max.pdf} &
        \includegraphics[page=3, width=1\linewidth]{figures/cdt6max.pdf} &
        \includegraphics[page=4, width=1\linewidth]{figures/cdt6max.pdf} &
        \includegraphics[page=5, width=1\linewidth]{figures/cdt6max.pdf} &
        \includegraphics[page=6, width=1\linewidth]{figures/cdt6max.pdf} &
        \includegraphics[page=7, width=1\linewidth]{figures/cdt6max.pdf} &
        \includegraphics[page=8, width=1\linewidth]{figures/cdt6max.pdf}\\\midrule
        \includegraphics[page=1, width=1\linewidth]{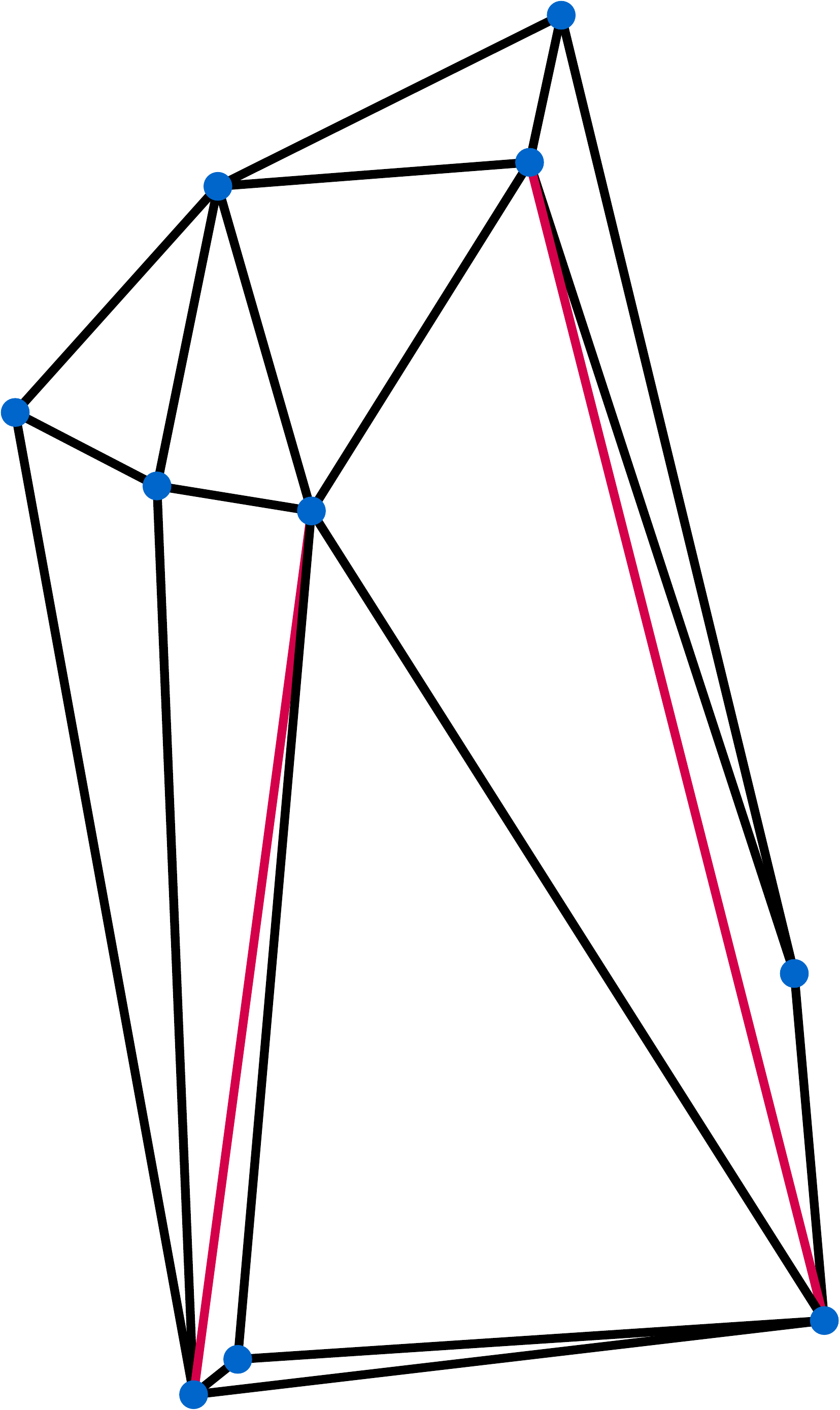} &
        \includegraphics[page=2, width=1\linewidth]{figures/cdt7max.pdf} &
        \includegraphics[page=3, width=1\linewidth]{figures/cdt7max.pdf} &
        \includegraphics[page=4, width=1\linewidth]{figures/cdt7max.pdf} &
        \includegraphics[page=5, width=1\linewidth]{figures/cdt7max.pdf} &
        \includegraphics[page=6, width=1\linewidth]{figures/cdt7max.pdf} &
        \includegraphics[page=7, width=1\linewidth]{figures/cdt7max.pdf} &
        \includegraphics[page=8, width=1\linewidth]{figures/cdt7max.pdf}\\
        \bottomrule
\end{tabu}
\end{table*}

\begin{table*}[t]
\caption{Optimizing with minimum length 1.2 Delaunay length (Sum aggregation)}
\label{tab:minlength_sum}
\small
\begin{tabu} to \textwidth {X[1,c,m]|X[1,c,m]|X[1,c,m]|X[1,c,m]|X[1,c,m]|X[1,c,m]|X[1,c,m]|X[1,c,m]}
        \toprule
        DT & Opposing Angles & Dual Edge Ratio & Dual Area overlap & Lens & Shrunk Circle & Triangular Lens & Shrunk Circumcircle\\
        \midrule
        \includegraphics[page=1, width=1\linewidth]{figures/minLength0sum.pdf} &
        \includegraphics[page=2, width=1\linewidth]{figures/minLength0sum.pdf} &
        \includegraphics[page=3, width=1\linewidth]{figures/minLength0sum.pdf} &
        \includegraphics[page=4, width=1\linewidth]{figures/minLength0sum.pdf} &
        \includegraphics[page=5, width=1\linewidth]{figures/minLength0sum.pdf} &
        \includegraphics[page=6, width=1\linewidth]{figures/minLength0sum.pdf} &
        \includegraphics[page=7, width=1\linewidth]{figures/minLength0sum.pdf} &
        \includegraphics[page=8, width=1\linewidth]{figures/minLength0sum.pdf}\\\midrule
        \includegraphics[page=1, width=1\linewidth]{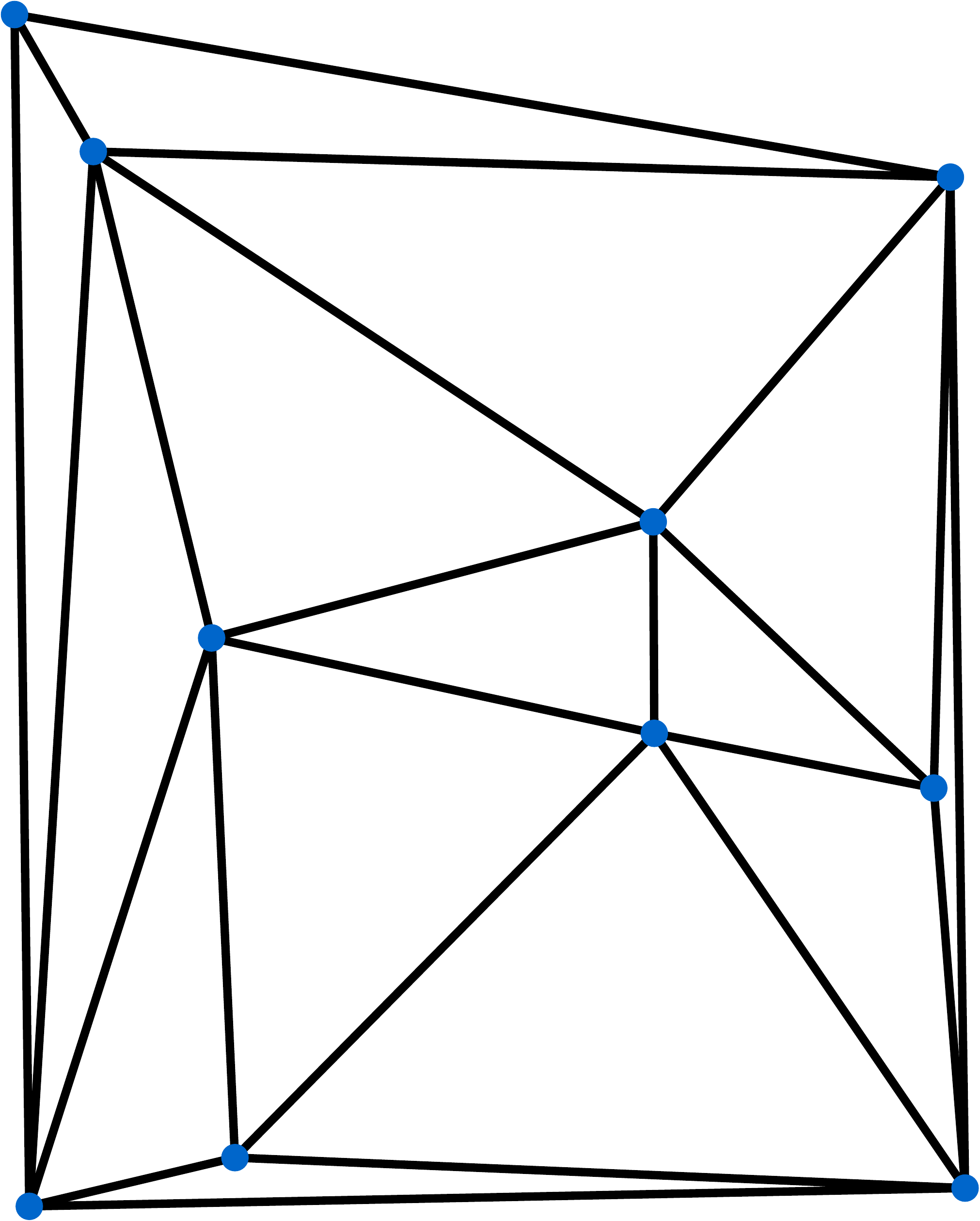} &
        \includegraphics[page=2, width=1\linewidth]{figures/minLength1sum.pdf} &
        \includegraphics[page=3, width=1\linewidth]{figures/minLength1sum.pdf} &
        \includegraphics[page=4, width=1\linewidth]{figures/minLength1sum.pdf} &
        \includegraphics[page=5, width=1\linewidth]{figures/minLength1sum.pdf} &
        \includegraphics[page=6, width=1\linewidth]{figures/minLength1sum.pdf} &
        \includegraphics[page=7, width=1\linewidth]{figures/minLength1sum.pdf} &
        \includegraphics[page=8, width=1\linewidth]{figures/minLength1sum.pdf}\\\midrule
        \includegraphics[page=1, width=1\linewidth]{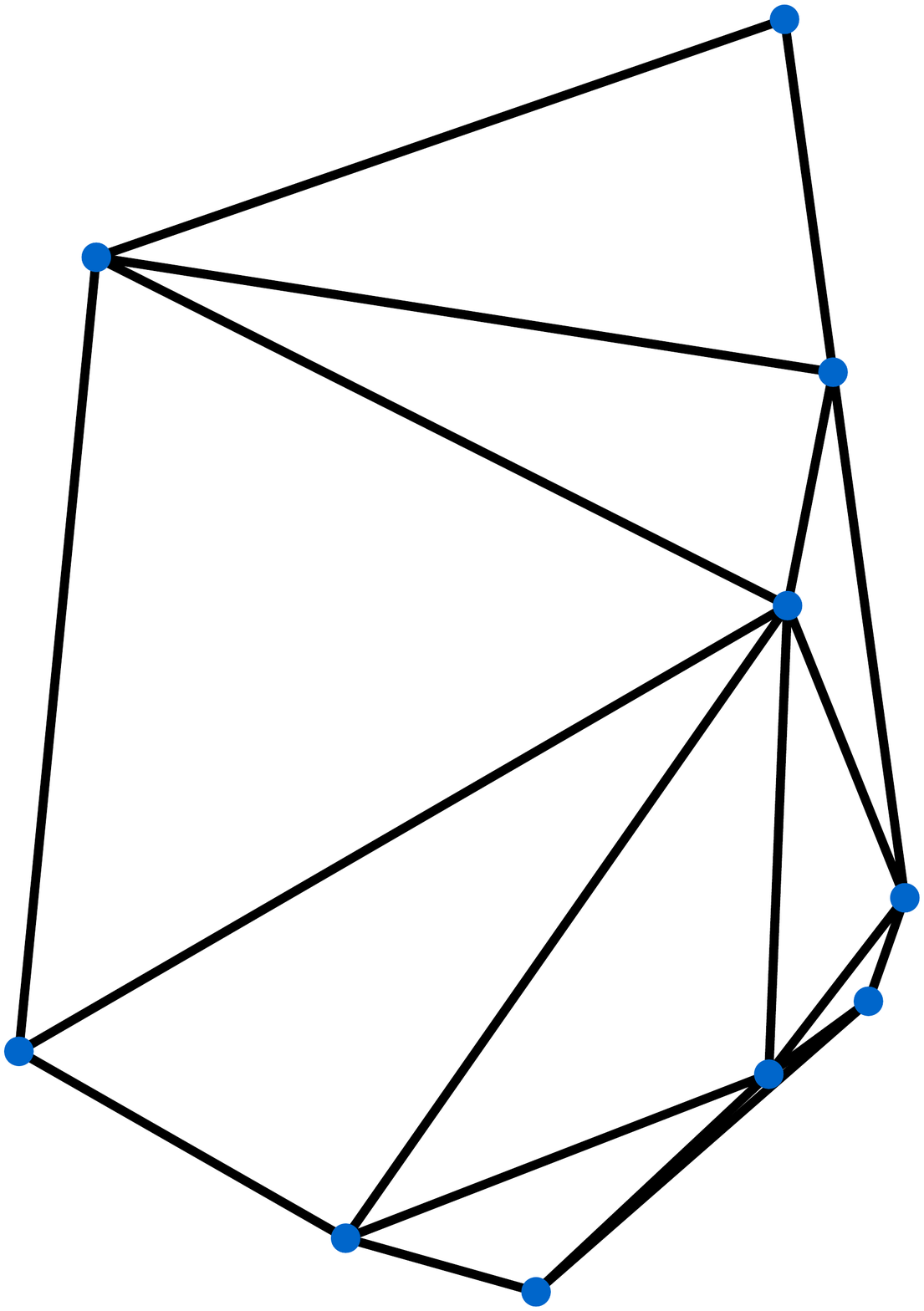} &
        \includegraphics[page=2, width=1\linewidth]{figures/minLength2sum.pdf} &
        \includegraphics[page=3, width=1\linewidth]{figures/minLength2sum.pdf} &
        \includegraphics[page=4, width=1\linewidth]{figures/minLength2sum.pdf} &
        \includegraphics[page=5, width=1\linewidth]{figures/minLength2sum.pdf} &
        \includegraphics[page=6, width=1\linewidth]{figures/minLength2sum.pdf} &
        \includegraphics[page=7, width=1\linewidth]{figures/minLength2sum.pdf} &
        \includegraphics[page=8, width=1\linewidth]{figures/minLength2sum.pdf}\\\midrule
        \includegraphics[page=1, width=1\linewidth]{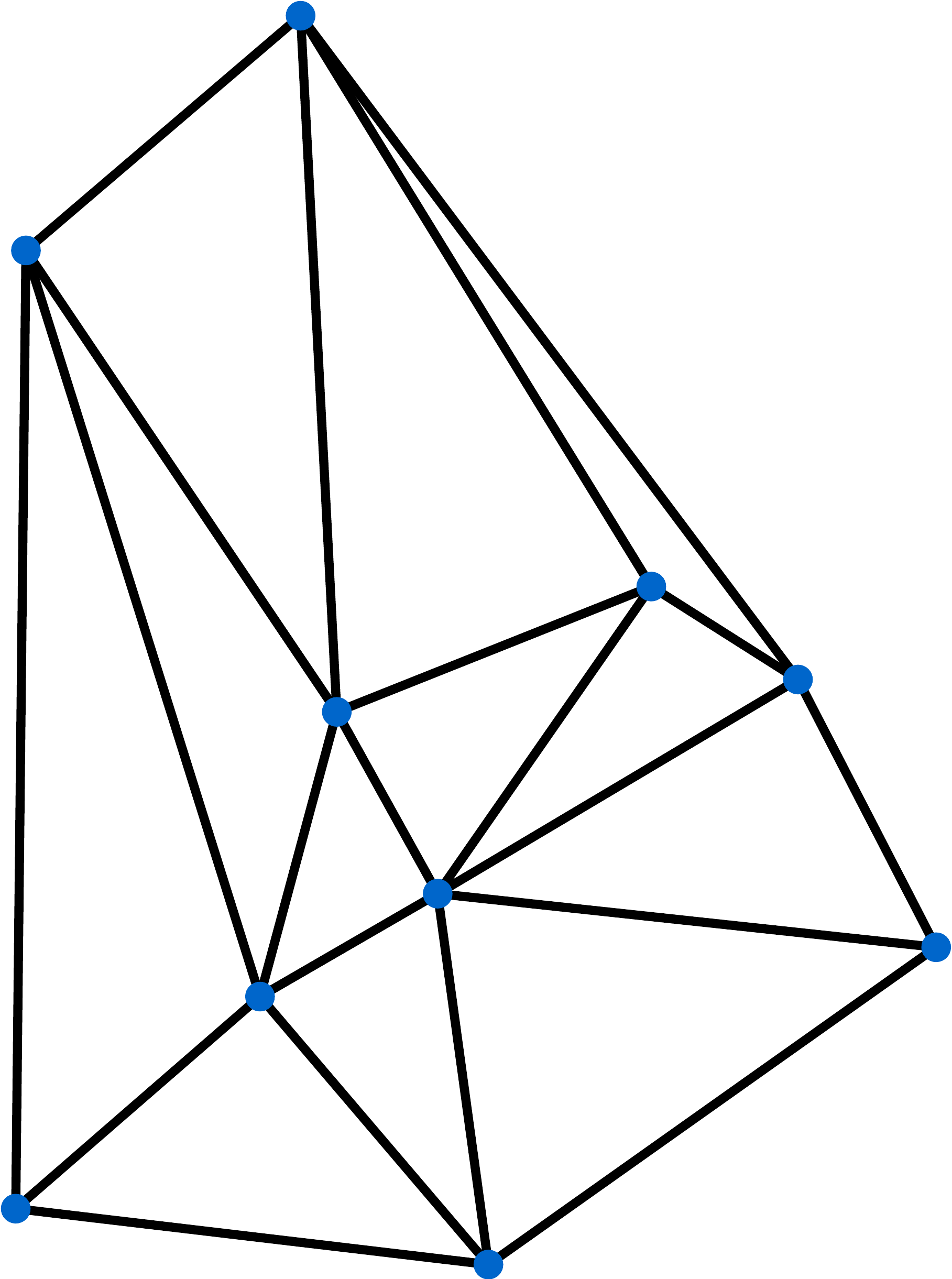} &
        \includegraphics[page=2, width=1\linewidth]{figures/minLength3sum.pdf} &
        \includegraphics[page=3, width=1\linewidth]{figures/minLength3sum.pdf} &
        \includegraphics[page=4, width=1\linewidth]{figures/minLength3sum.pdf} &
        \includegraphics[page=5, width=1\linewidth]{figures/minLength3sum.pdf} &
        \includegraphics[page=6, width=1\linewidth]{figures/minLength3sum.pdf} &
        \includegraphics[page=7, width=1\linewidth]{figures/minLength3sum.pdf} &
        \includegraphics[page=8, width=1\linewidth]{figures/minLength3sum.pdf}\\\midrule
        \includegraphics[page=1, width=1\linewidth]{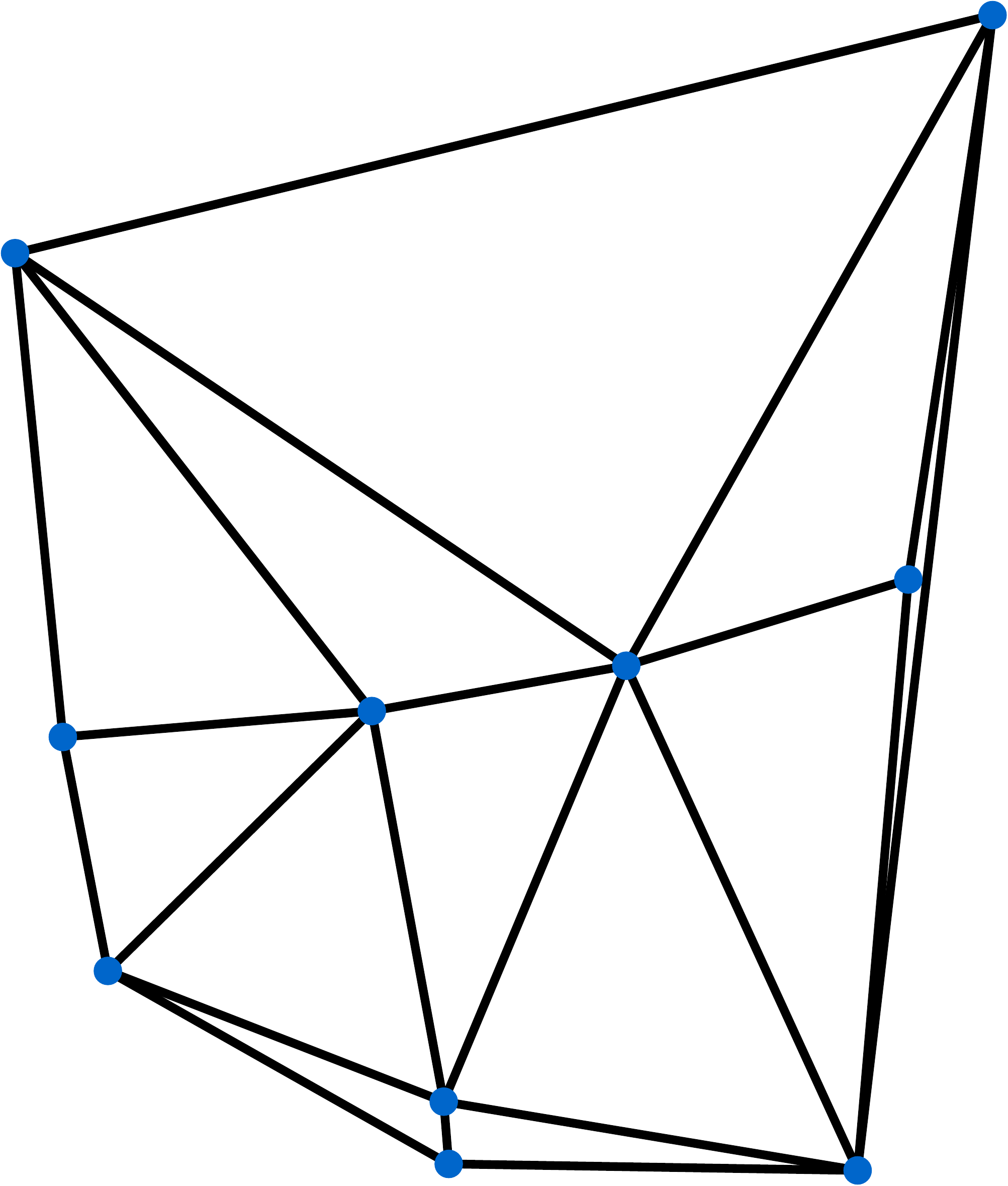} &
        \includegraphics[page=2, width=1\linewidth]{figures/minLength4sum.pdf} &
        \includegraphics[page=3, width=1\linewidth]{figures/minLength4sum.pdf} &
        \includegraphics[page=4, width=1\linewidth]{figures/minLength4sum.pdf} &
        \includegraphics[page=5, width=1\linewidth]{figures/minLength4sum.pdf} &
        \includegraphics[page=6, width=1\linewidth]{figures/minLength4sum.pdf} &
        \includegraphics[page=7, width=1\linewidth]{figures/minLength4sum.pdf} &
        \includegraphics[page=8, width=1\linewidth]{figures/minLength4sum.pdf}\\\midrule
        \includegraphics[page=1, width=1\linewidth]{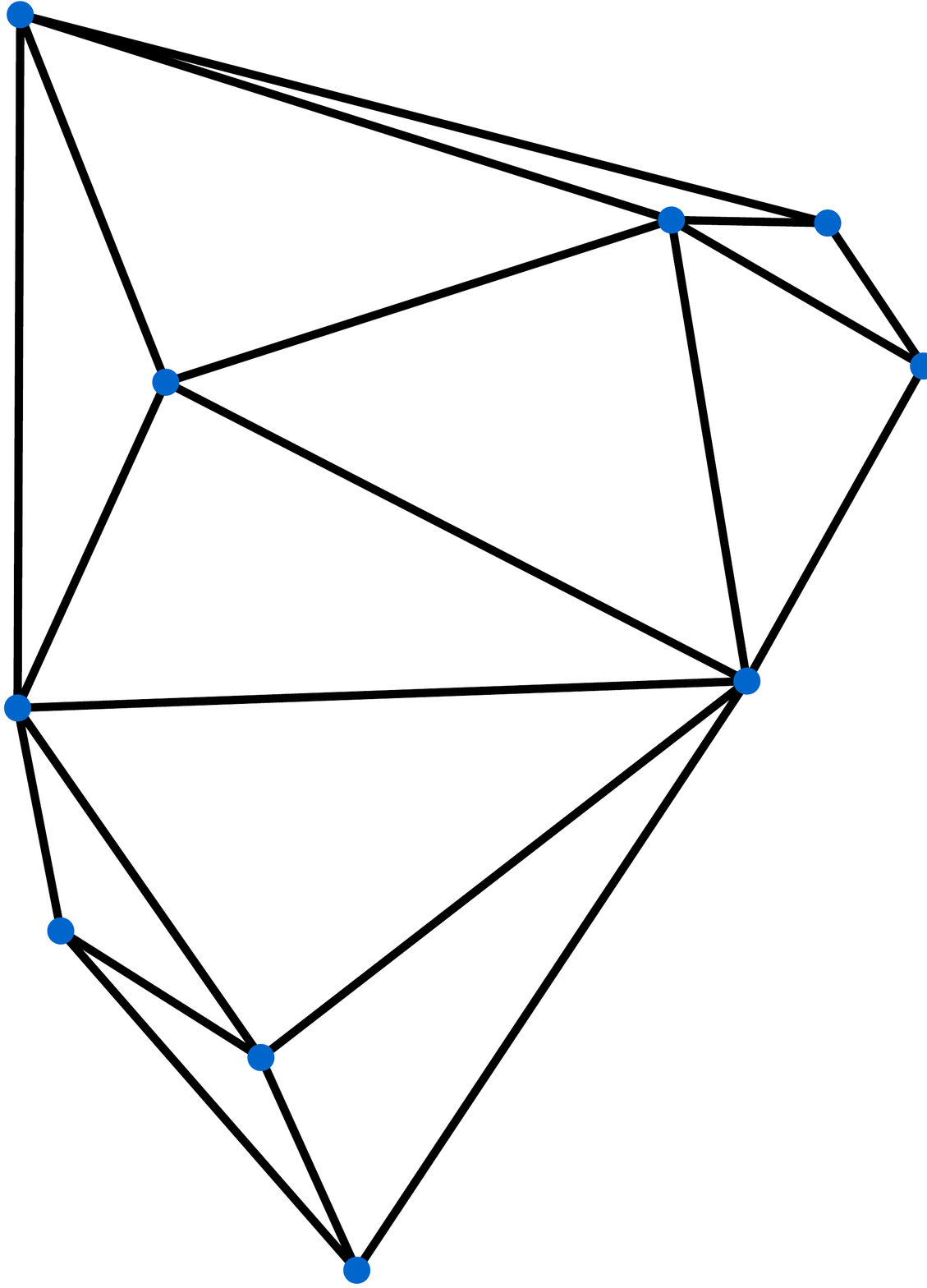} &
        \includegraphics[page=2, width=1\linewidth]{figures/minLength5sum.pdf} &
        \includegraphics[page=3, width=1\linewidth]{figures/minLength5sum.pdf} &
        \includegraphics[page=4, width=1\linewidth]{figures/minLength5sum.pdf} &
        \includegraphics[page=5, width=1\linewidth]{figures/minLength5sum.pdf} &
        \includegraphics[page=6, width=1\linewidth]{figures/minLength5sum.pdf} &
        \includegraphics[page=7, width=1\linewidth]{figures/minLength5sum.pdf} &
        \includegraphics[page=8, width=1\linewidth]{figures/minLength5sum.pdf}\\\midrule
        \includegraphics[page=1, width=1\linewidth]{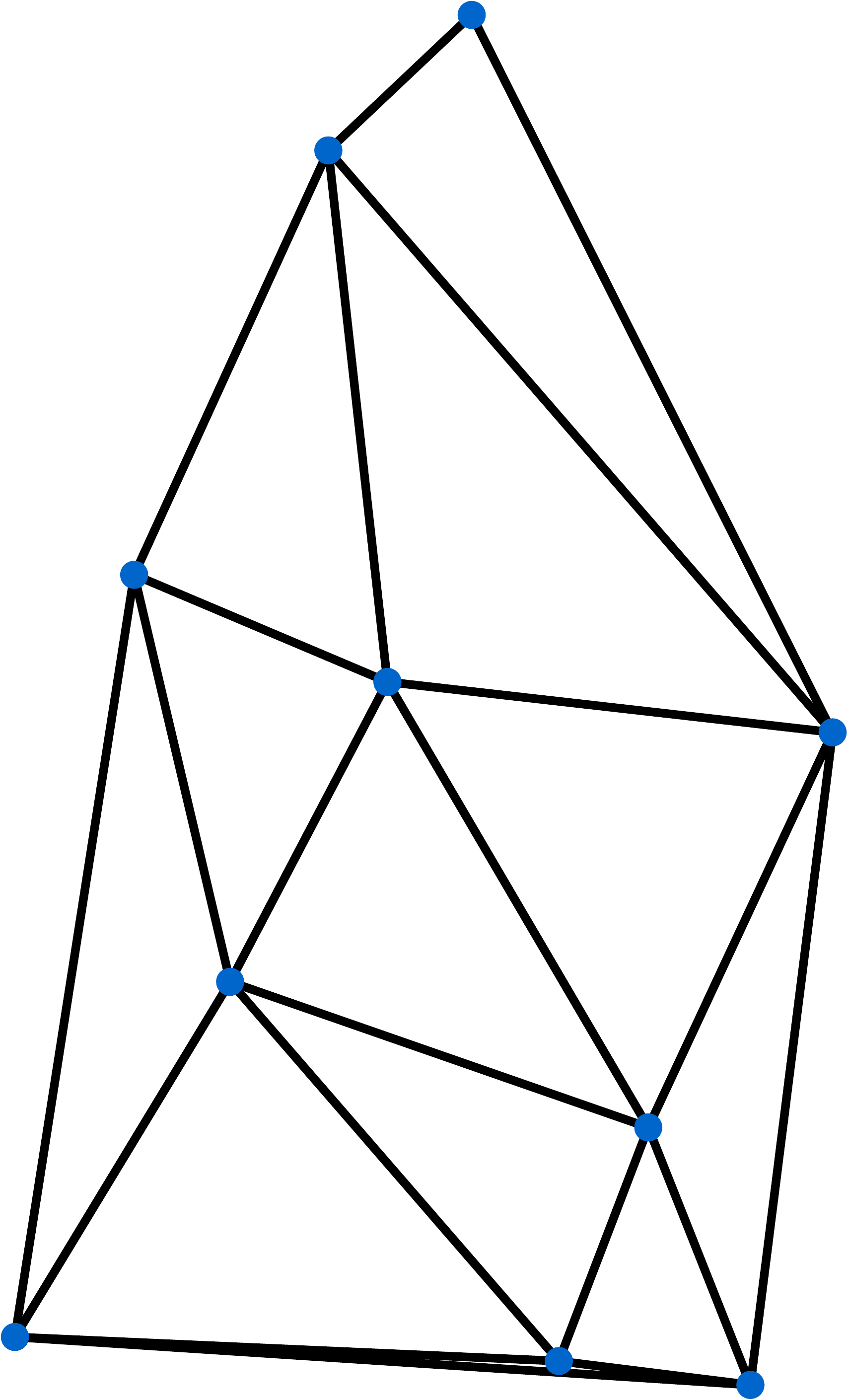} &
        \includegraphics[page=2, width=1\linewidth]{figures/minLength6sum.pdf} &
        \includegraphics[page=3, width=1\linewidth]{figures/minLength6sum.pdf} &
        \includegraphics[page=4, width=1\linewidth]{figures/minLength6sum.pdf} &
        \includegraphics[page=5, width=1\linewidth]{figures/minLength6sum.pdf} &
        \includegraphics[page=6, width=1\linewidth]{figures/minLength6sum.pdf} &
        \includegraphics[page=7, width=1\linewidth]{figures/minLength6sum.pdf} &
        \includegraphics[page=8, width=1\linewidth]{figures/minLength6sum.pdf}\\\midrule
        \includegraphics[page=1, width=1\linewidth]{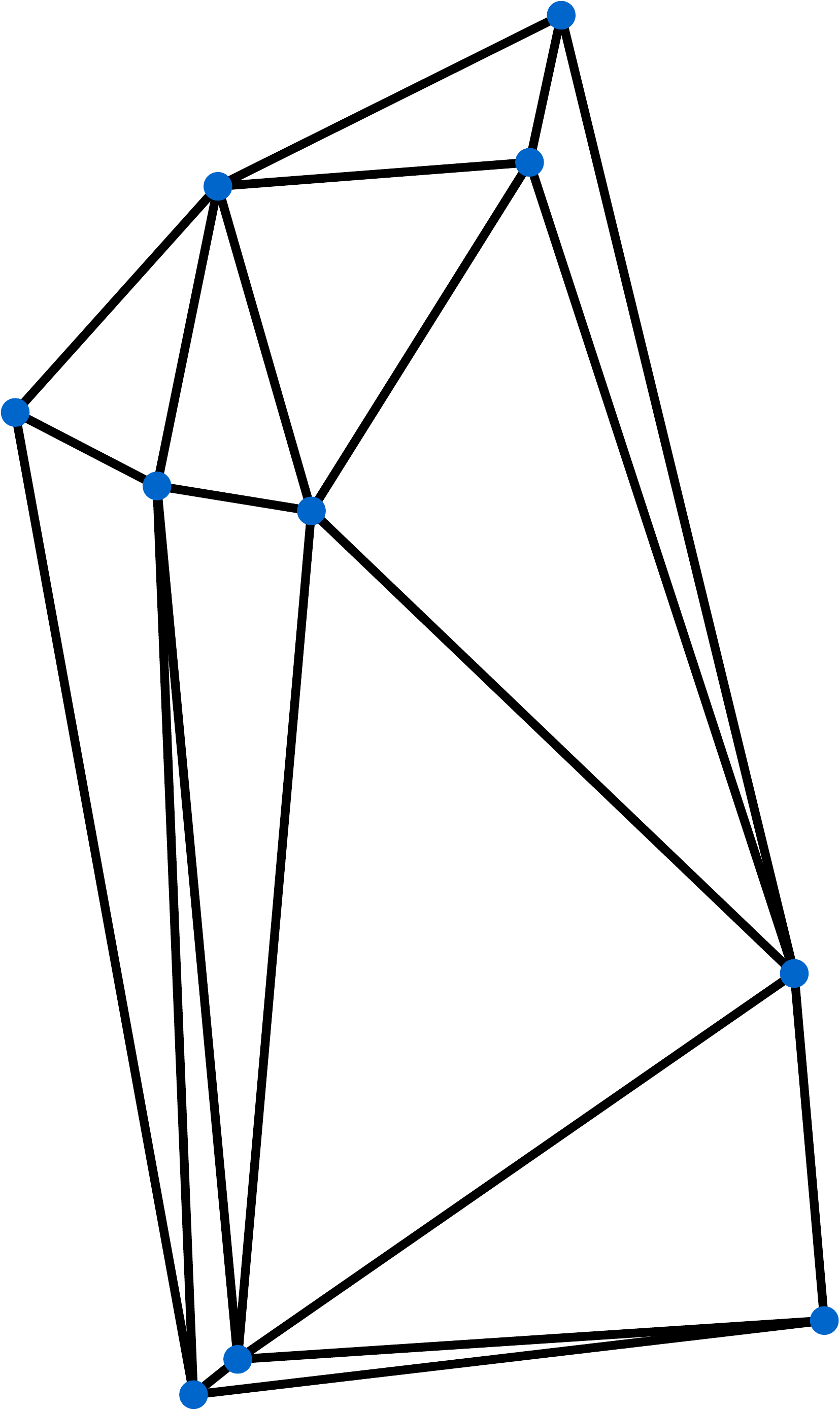} &
        \includegraphics[page=2, width=1\linewidth]{figures/minLength7sum.pdf} &
        \includegraphics[page=3, width=1\linewidth]{figures/minLength7sum.pdf} &
        \includegraphics[page=4, width=1\linewidth]{figures/minLength7sum.pdf} &
        \includegraphics[page=5, width=1\linewidth]{figures/minLength7sum.pdf} &
        \includegraphics[page=6, width=1\linewidth]{figures/minLength7sum.pdf} &
        \includegraphics[page=7, width=1\linewidth]{figures/minLength7sum.pdf} &
        \includegraphics[page=8, width=1\linewidth]{figures/minLength7sum.pdf}\\
        \bottomrule
\end{tabu}
\end{table*}

\begin{table*}[t]
\caption{Optimizing with minimum length 1.2 Delaunay length (Bottleneck aggregation)}
\label{tab:minlength_bottleneck}
\small
\begin{tabu} to \textwidth {X[1,c,m]|X[1,c,m]|X[1,c,m]|X[1,c,m]|X[1,c,m]|X[1,c,m]|X[1,c,m]|X[1,c,m]}
        \toprule
        DT & Opposing Angles & Dual Edge Ratio & Dual Area overlap & Lens & Shrunk Circle & Triangular Lens & Shrunk Circumcircle\\
        \midrule
        \includegraphics[page=1, width=1\linewidth]{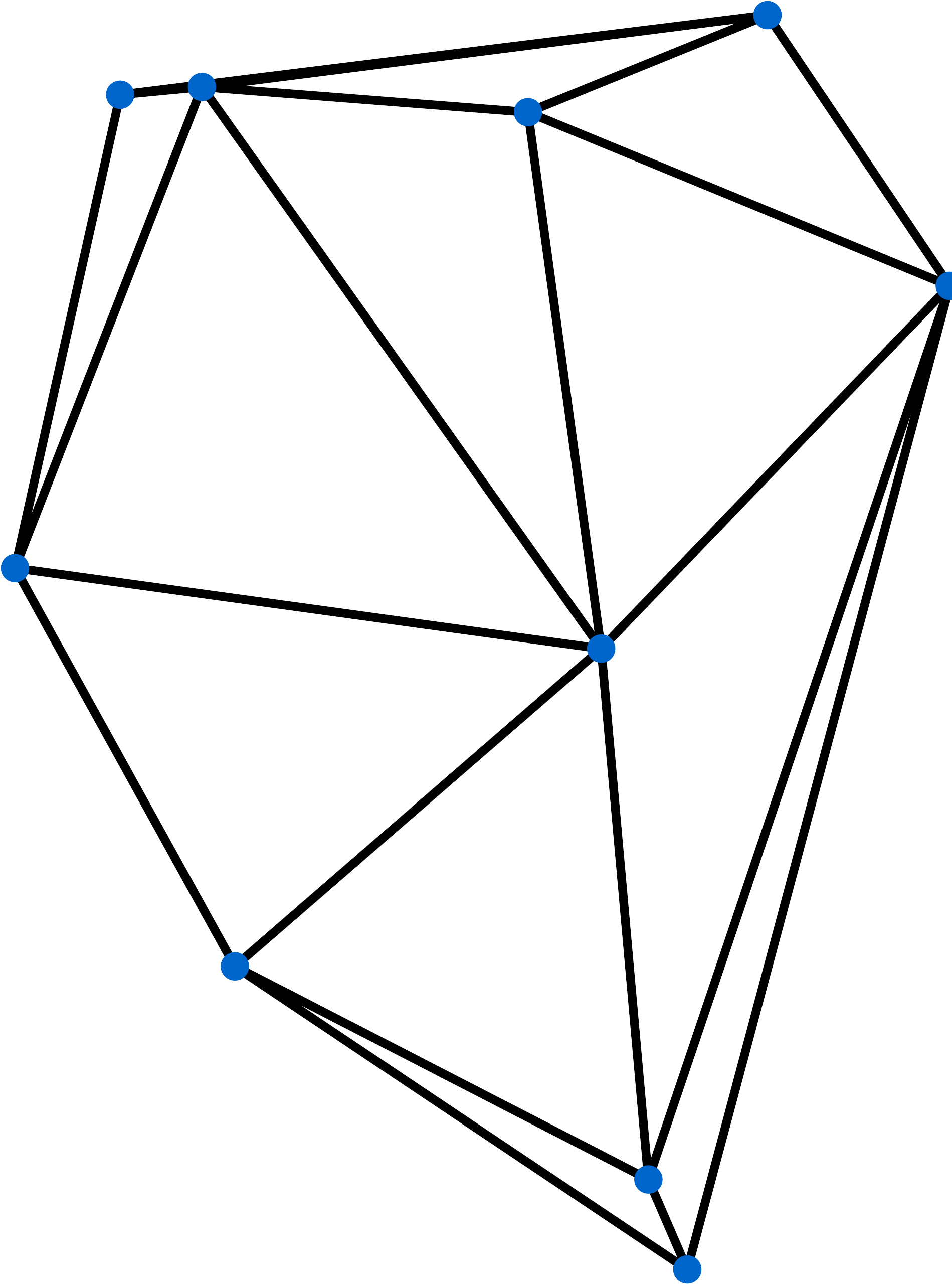} &
        \includegraphics[page=2, width=1\linewidth]{figures/minLength0max.pdf} &
        \includegraphics[page=3, width=1\linewidth]{figures/minLength0max.pdf} &
        \includegraphics[page=4, width=1\linewidth]{figures/minLength0max.pdf} &
        \includegraphics[page=5, width=1\linewidth]{figures/minLength0max.pdf} &
        \includegraphics[page=6, width=1\linewidth]{figures/minLength0max.pdf} &
        \includegraphics[page=7, width=1\linewidth]{figures/minLength0max.pdf} &
        \includegraphics[page=8, width=1\linewidth]{figures/minLength0max.pdf}\\\midrule
        \includegraphics[page=1, width=1\linewidth]{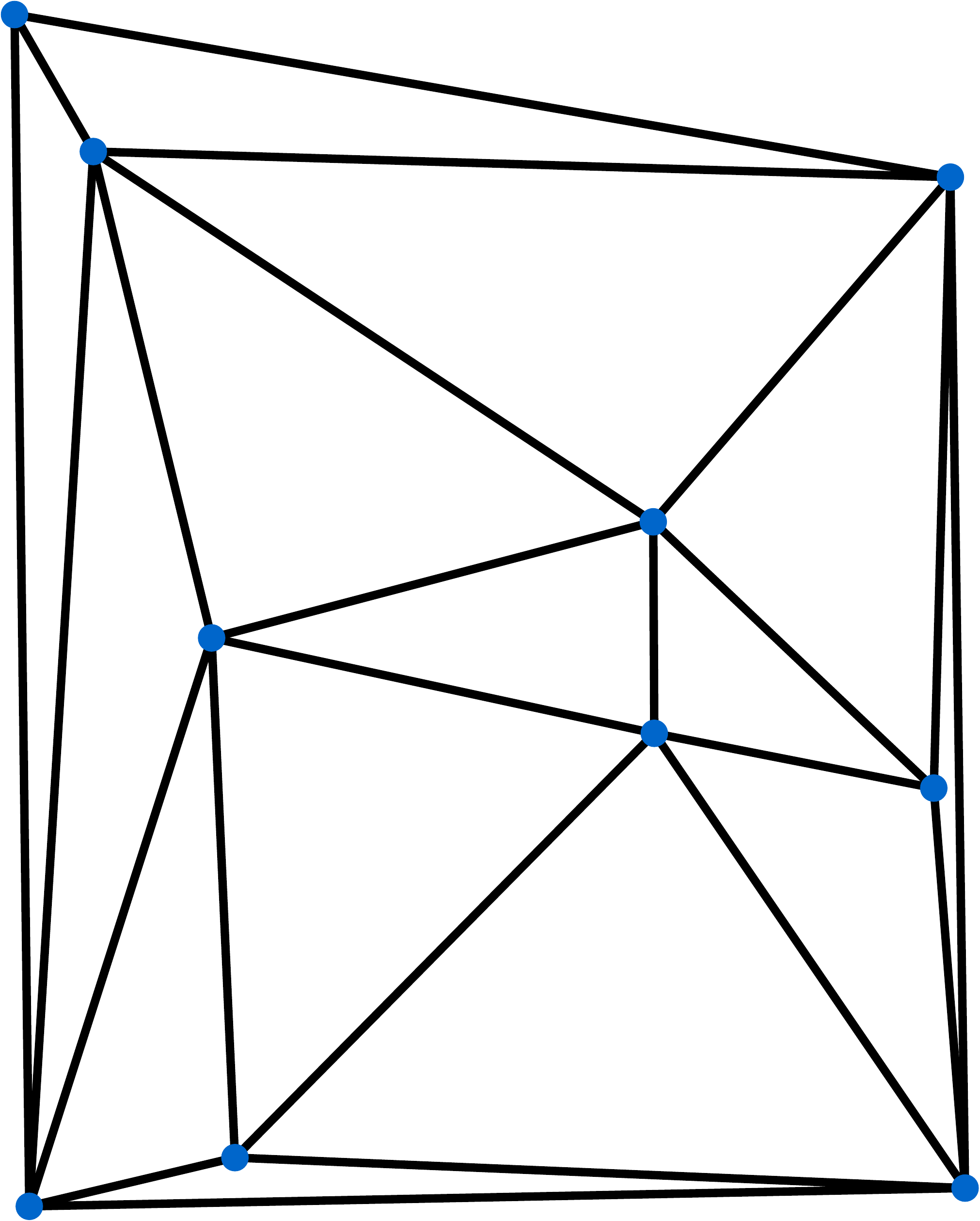} &
        \includegraphics[page=2, width=1\linewidth]{figures/minLength1max.pdf} &
        \includegraphics[page=3, width=1\linewidth]{figures/minLength1max.pdf} &
        \includegraphics[page=4, width=1\linewidth]{figures/minLength1max.pdf} &
        \includegraphics[page=5, width=1\linewidth]{figures/minLength1max.pdf} &
        \includegraphics[page=6, width=1\linewidth]{figures/minLength1max.pdf} &
        \includegraphics[page=7, width=1\linewidth]{figures/minLength1max.pdf} &
        \includegraphics[page=8, width=1\linewidth]{figures/minLength1max.pdf}\\\midrule
        \includegraphics[page=1, width=1\linewidth]{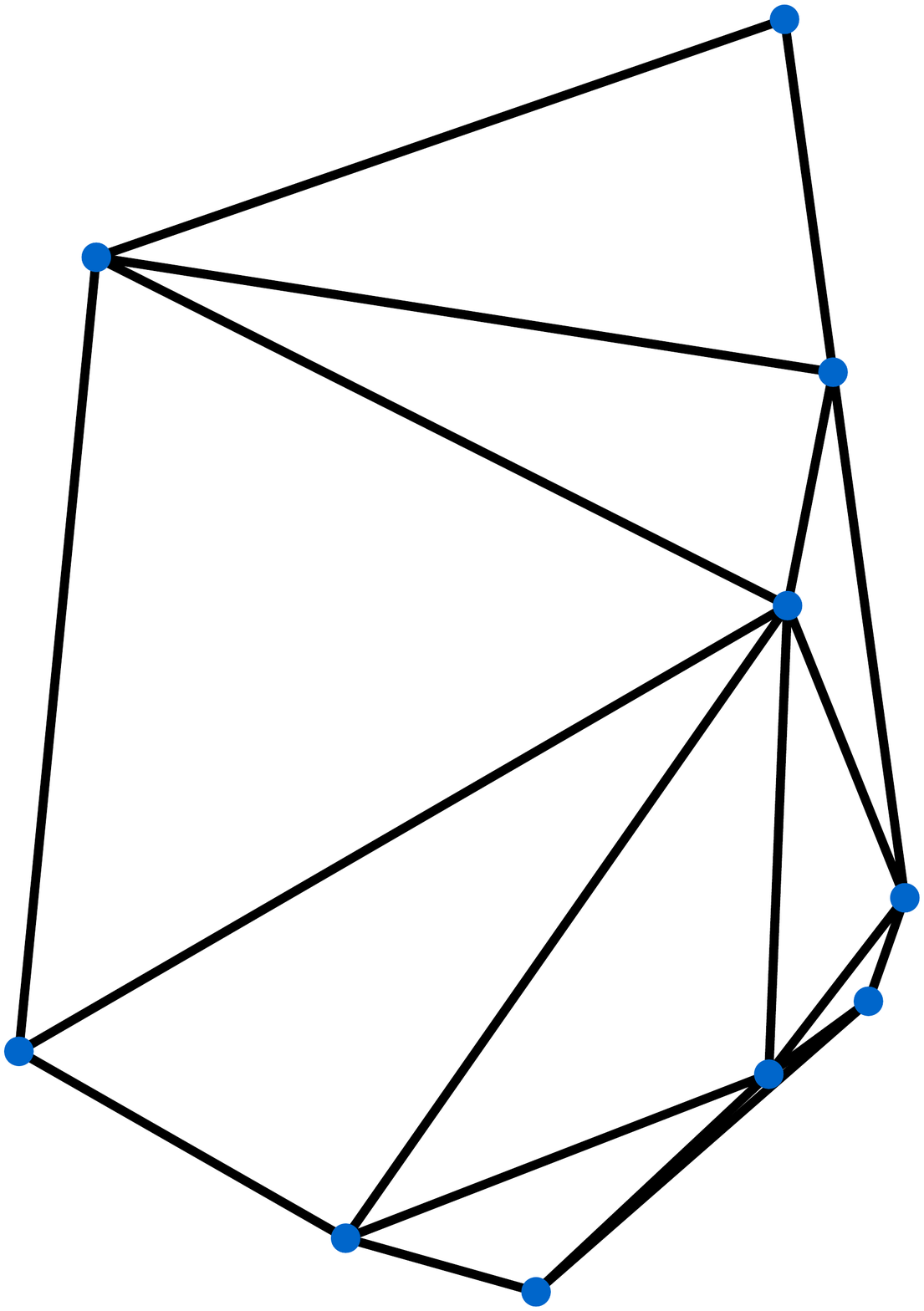} &
        \includegraphics[page=2, width=1\linewidth]{figures/minLength2max.pdf} &
        \includegraphics[page=3, width=1\linewidth]{figures/minLength2max.pdf} &
        \includegraphics[page=4, width=1\linewidth]{figures/minLength2max.pdf} &
        \includegraphics[page=5, width=1\linewidth]{figures/minLength2max.pdf} &
        \includegraphics[page=6, width=1\linewidth]{figures/minLength2max.pdf} &
        \includegraphics[page=7, width=1\linewidth]{figures/minLength2max.pdf} &
        \includegraphics[page=8, width=1\linewidth]{figures/minLength2max.pdf}\\\midrule
        \includegraphics[page=1, width=1\linewidth]{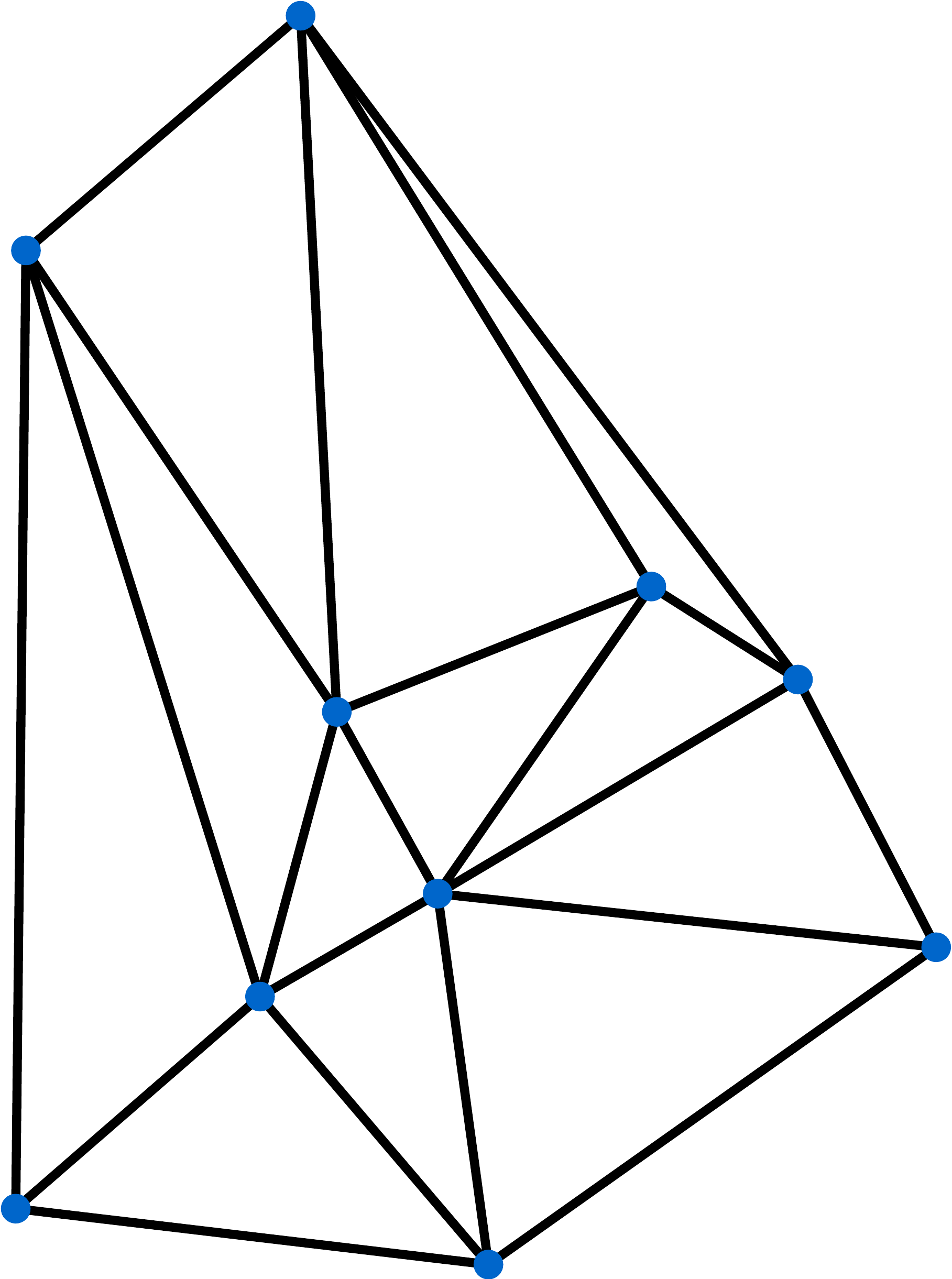} &
        \includegraphics[page=2, width=1\linewidth]{figures/minLength3max.pdf} &
        \includegraphics[page=3, width=1\linewidth]{figures/minLength3max.pdf} &
        \includegraphics[page=4, width=1\linewidth]{figures/minLength3max.pdf} &
        \includegraphics[page=5, width=1\linewidth]{figures/minLength3max.pdf} &
        \includegraphics[page=6, width=1\linewidth]{figures/minLength3max.pdf} &
        \includegraphics[page=7, width=1\linewidth]{figures/minLength3max.pdf} &
        \includegraphics[page=8, width=1\linewidth]{figures/minLength3max.pdf}\\\midrule
        \includegraphics[page=1, width=1\linewidth]{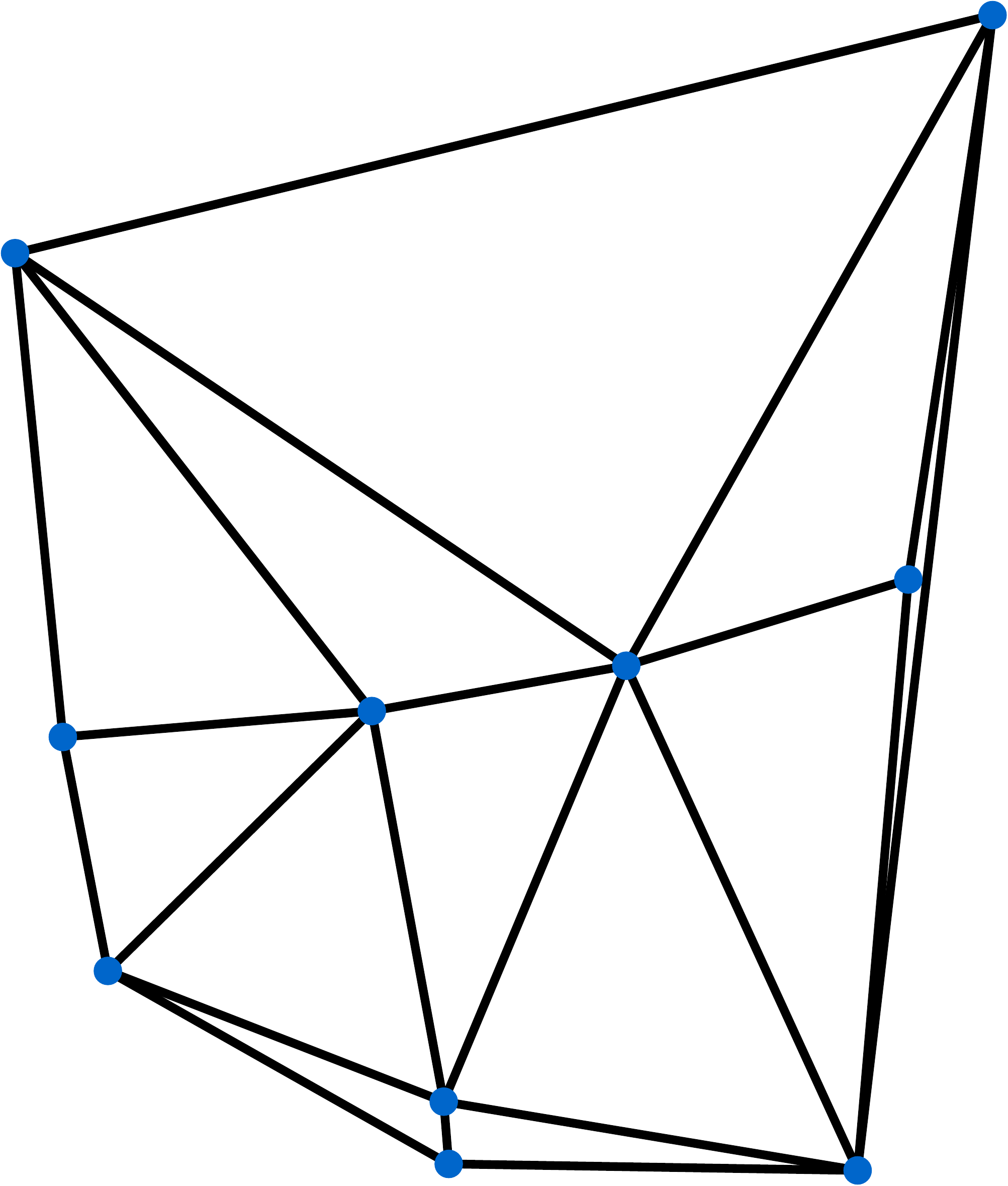} &
        \includegraphics[page=2, width=1\linewidth]{figures/minLength4max.pdf} &
        \includegraphics[page=3, width=1\linewidth]{figures/minLength4max.pdf} &
        \includegraphics[page=4, width=1\linewidth]{figures/minLength4max.pdf} &
        \includegraphics[page=5, width=1\linewidth]{figures/minLength4max.pdf} &
        \includegraphics[page=6, width=1\linewidth]{figures/minLength4max.pdf} &
        \includegraphics[page=7, width=1\linewidth]{figures/minLength4max.pdf} &
        \includegraphics[page=8, width=1\linewidth]{figures/minLength4max.pdf}\\\midrule
        \includegraphics[page=1, width=1\linewidth]{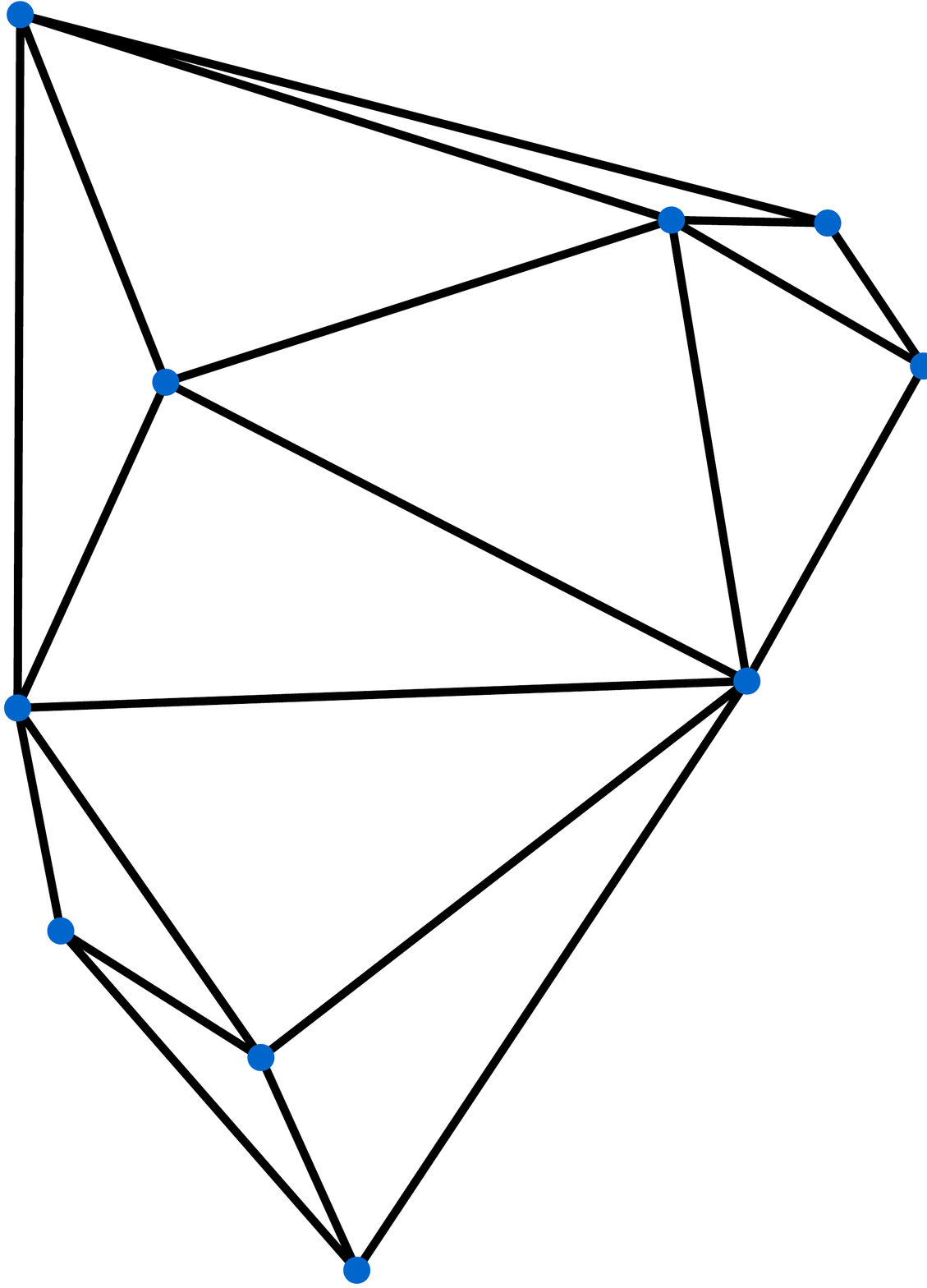} &
        \includegraphics[page=2, width=1\linewidth]{figures/minLength5max.pdf} &
        \includegraphics[page=3, width=1\linewidth]{figures/minLength5max.pdf} &
        \includegraphics[page=4, width=1\linewidth]{figures/minLength5max.pdf} &
        \includegraphics[page=5, width=1\linewidth]{figures/minLength5max.pdf} &
        \includegraphics[page=6, width=1\linewidth]{figures/minLength5max.pdf} &
        \includegraphics[page=7, width=1\linewidth]{figures/minLength5max.pdf} &
        \includegraphics[page=8, width=1\linewidth]{figures/minLength5max.pdf}\\\midrule
        \includegraphics[page=1, width=1\linewidth]{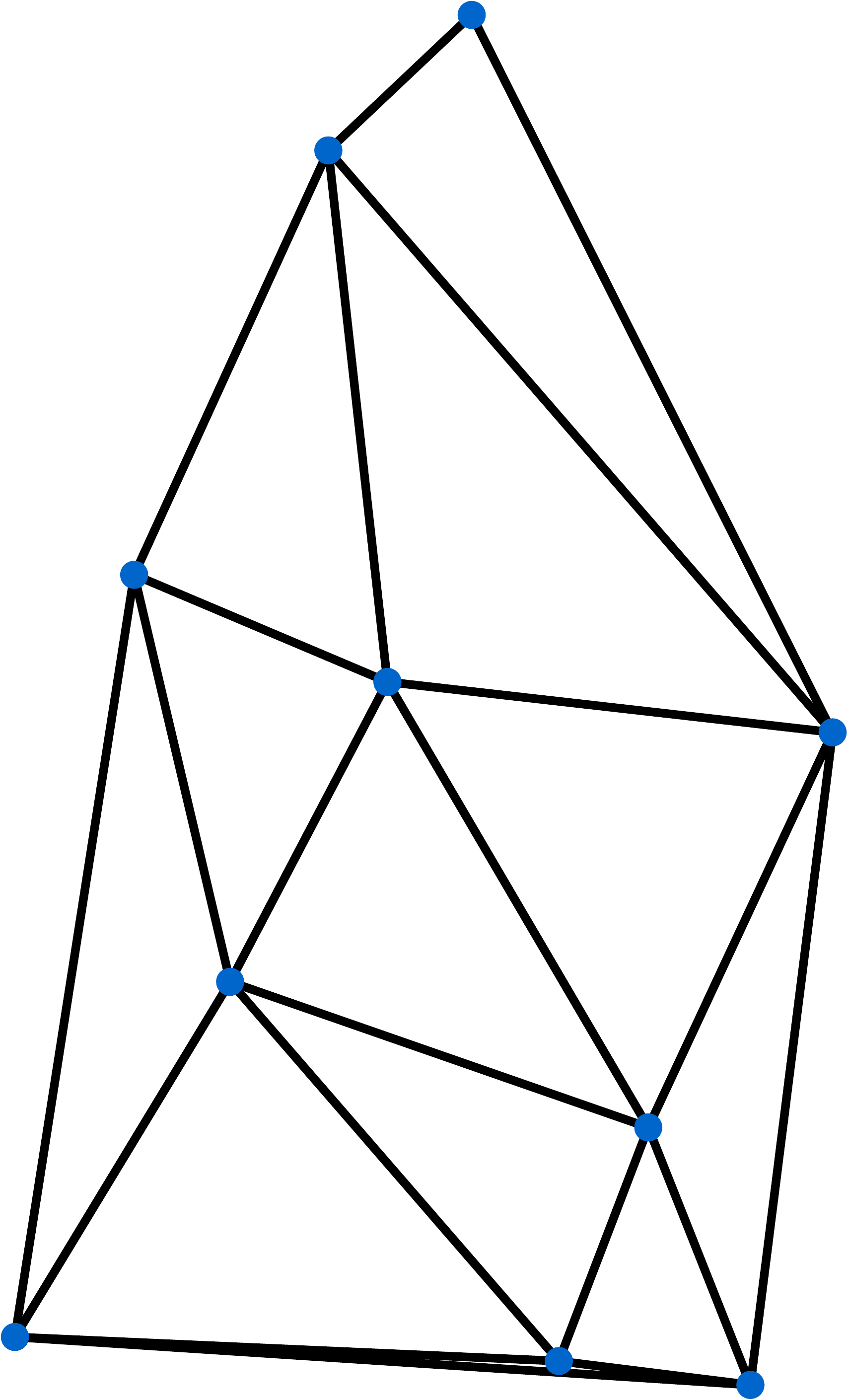} &
        \includegraphics[page=2, width=1\linewidth]{figures/minLength6max.pdf} &
        \includegraphics[page=3, width=1\linewidth]{figures/minLength6max.pdf} &
        \includegraphics[page=4, width=1\linewidth]{figures/minLength6max.pdf} &
        \includegraphics[page=5, width=1\linewidth]{figures/minLength6max.pdf} &
        \includegraphics[page=6, width=1\linewidth]{figures/minLength6max.pdf} &
        \includegraphics[page=7, width=1\linewidth]{figures/minLength6max.pdf} &
        \includegraphics[page=8, width=1\linewidth]{figures/minLength6max.pdf}\\\midrule
        \includegraphics[page=1, width=1\linewidth]{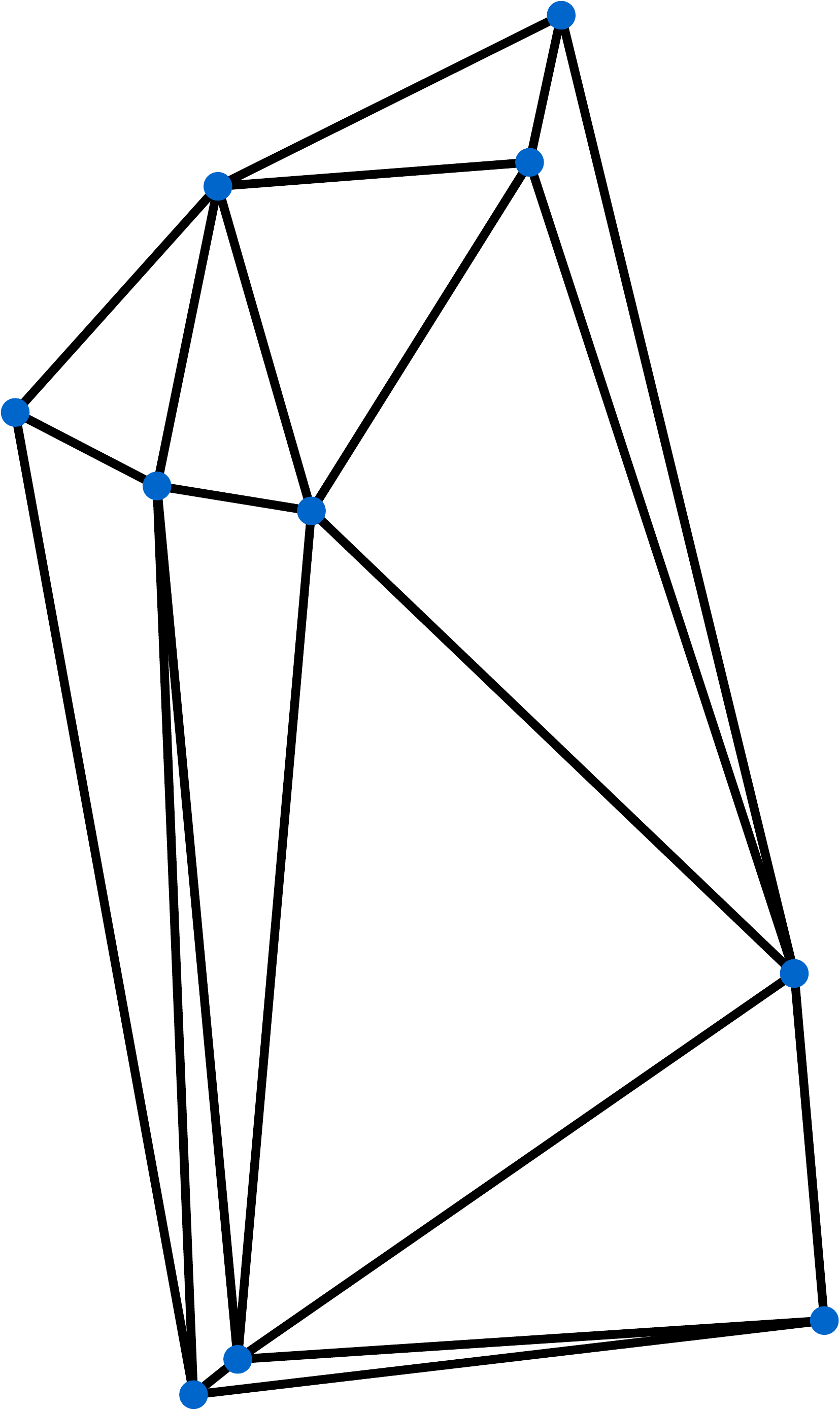} &
        \includegraphics[page=2, width=1\linewidth]{figures/minLength7max.pdf} &
        \includegraphics[page=3, width=1\linewidth]{figures/minLength7max.pdf} &
        \includegraphics[page=4, width=1\linewidth]{figures/minLength7max.pdf} &
        \includegraphics[page=5, width=1\linewidth]{figures/minLength7max.pdf} &
        \includegraphics[page=6, width=1\linewidth]{figures/minLength7max.pdf} &
        \includegraphics[page=7, width=1\linewidth]{figures/minLength7max.pdf} &
        \includegraphics[page=8, width=1\linewidth]{figures/minLength7max.pdf}\\
        \bottomrule
\end{tabu}
\end{table*}

\begin{table*}[t]
\caption{Optimizing with maximum length 0.8 Delaunay length (Sum aggregation)}
\label{tab:maxlength_sum}
\small
\begin{tabu} to \textwidth {X[1,c,m]|X[1,c,m]|X[1,c,m]|X[1,c,m]|X[1,c,m]|X[1,c,m]|X[1,c,m]|X[1,c,m]}
        \toprule
        DT & Opposing Angles & Dual Edge Ratio & Dual Area overlap & Lens & Shrunk Circle & Triangular Lens & Shrunk Circumcircle\\
        \midrule
        \includegraphics[page=1, width=1\linewidth]{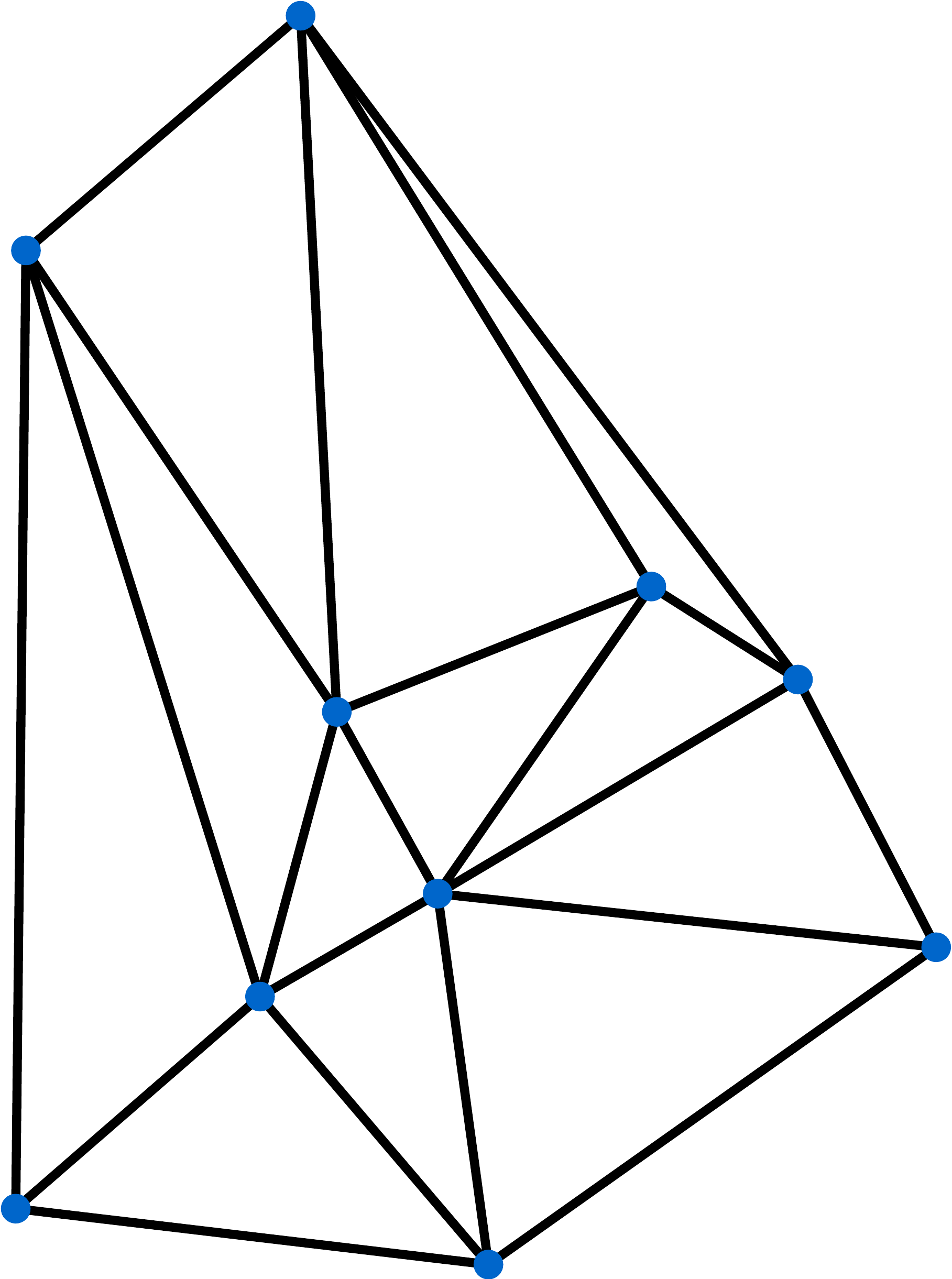} &
        \includegraphics[page=2, width=1\linewidth]{figures/maxLength3sum.pdf} &
        \includegraphics[page=3, width=1\linewidth]{figures/maxLength3sum.pdf} &
        \includegraphics[page=4, width=1\linewidth]{figures/maxLength3sum.pdf} &
        \includegraphics[page=5, width=1\linewidth]{figures/maxLength3sum.pdf} &
        \includegraphics[page=6, width=1\linewidth]{figures/maxLength3sum.pdf} &
        \includegraphics[page=7, width=1\linewidth]{figures/maxLength3sum.pdf} &
        \includegraphics[page=8, width=1\linewidth]{figures/maxLength3sum.pdf}\\\midrule
        \includegraphics[page=1, width=1\linewidth]{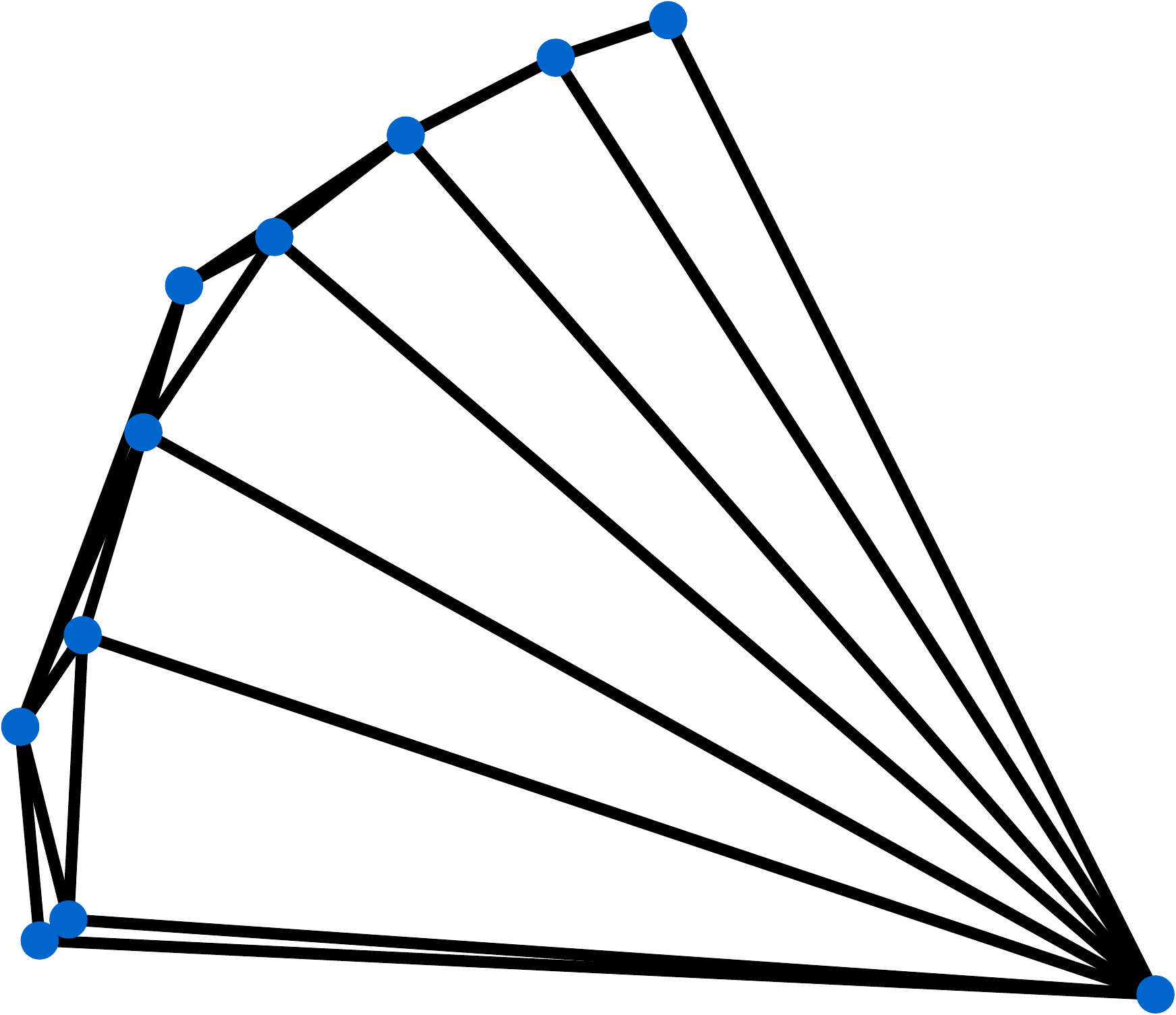} &
        \includegraphics[page=2, width=1\linewidth]{figures/maxLength8sum.pdf} &
        \includegraphics[page=3, width=1\linewidth]{figures/maxLength8sum.pdf} &
        \includegraphics[page=4, width=1\linewidth]{figures/maxLength8sum.pdf} &
        \includegraphics[page=5, width=1\linewidth]{figures/maxLength8sum.pdf} &
        \includegraphics[page=6, width=1\linewidth]{figures/maxLength8sum.pdf} &
        \includegraphics[page=7, width=1\linewidth]{figures/maxLength8sum.pdf} &
        \includegraphics[page=8, width=1\linewidth]{figures/maxLength8sum.pdf}\\\midrule
        \includegraphics[page=1, width=1\linewidth]{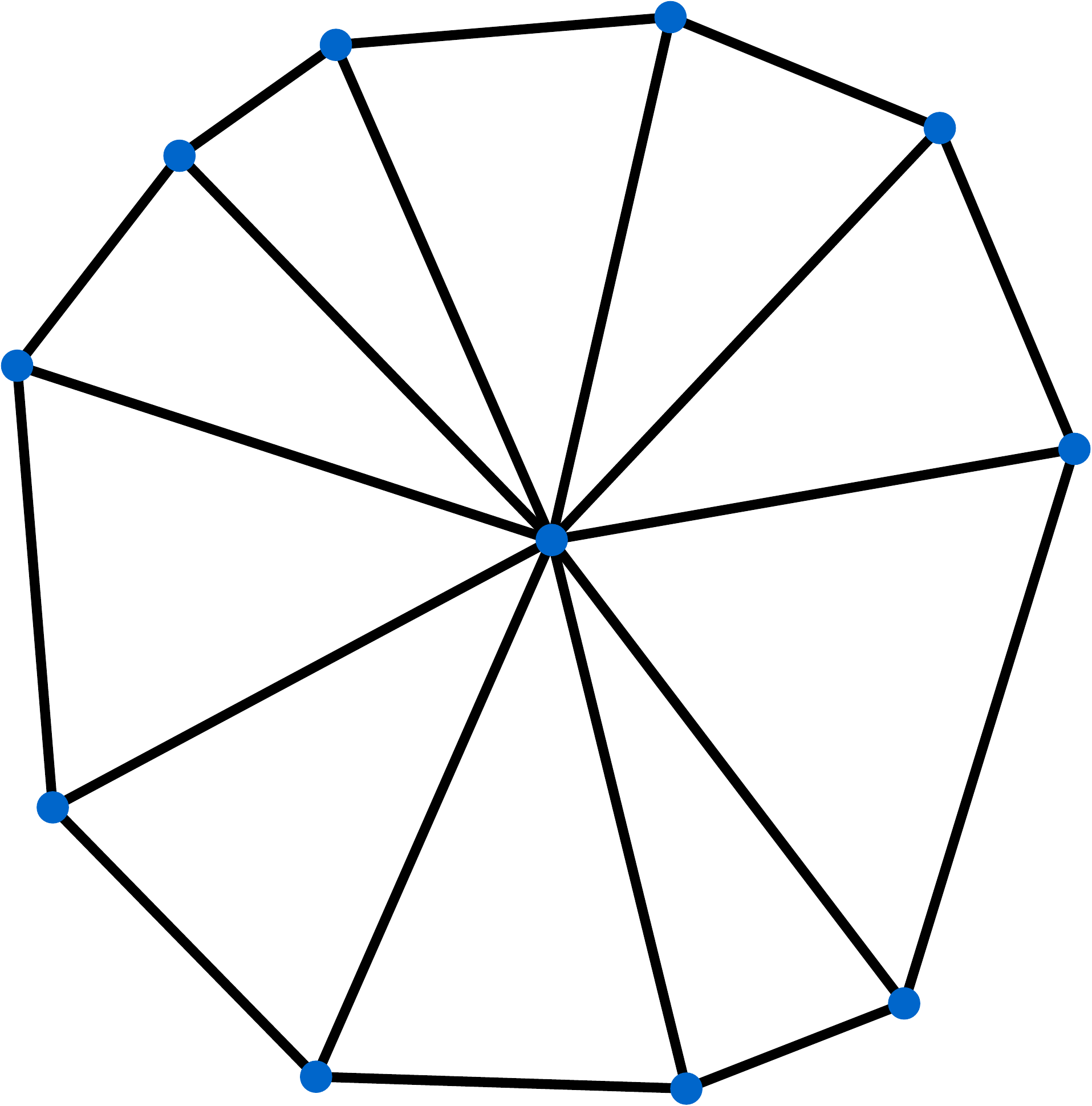} &
        \includegraphics[page=2, width=1\linewidth]{figures/maxLength9sum.pdf} &
        \includegraphics[page=3, width=1\linewidth]{figures/maxLength9sum.pdf} &
        \includegraphics[page=4, width=1\linewidth]{figures/maxLength9sum.pdf} &
        \includegraphics[page=5, width=1\linewidth]{figures/maxLength9sum.pdf} &
        \includegraphics[page=6, width=1\linewidth]{figures/maxLength9sum.pdf} &
        \includegraphics[page=7, width=1\linewidth]{figures/maxLength9sum.pdf} &
        \includegraphics[page=8, width=1\linewidth]{figures/maxLength9sum.pdf}\\
        \bottomrule
\end{tabu}
\vspace{\baselineskip}
\end{table*}

\begin{table*}[t]
\caption{Optimizing with maximum length 0.8 Delaunay length (Bottleneck aggregation)}
\label{tab:maxlength_bottleneck}
\small
\begin{tabu} to \textwidth {X[1,c,m]|X[1,c,m]|X[1,c,m]|X[1,c,m]|X[1,c,m]|X[1,c,m]|X[1,c,m]|X[1,c,m]}
        \toprule
        DT & Opposing Angles & Dual Edge Ratio & Dual Area overlap & Lens & Shrunk Circle & Triangular Lens & Shrunk Circumcircle\\
        \midrule
        \includegraphics[page=1, width=1\linewidth]{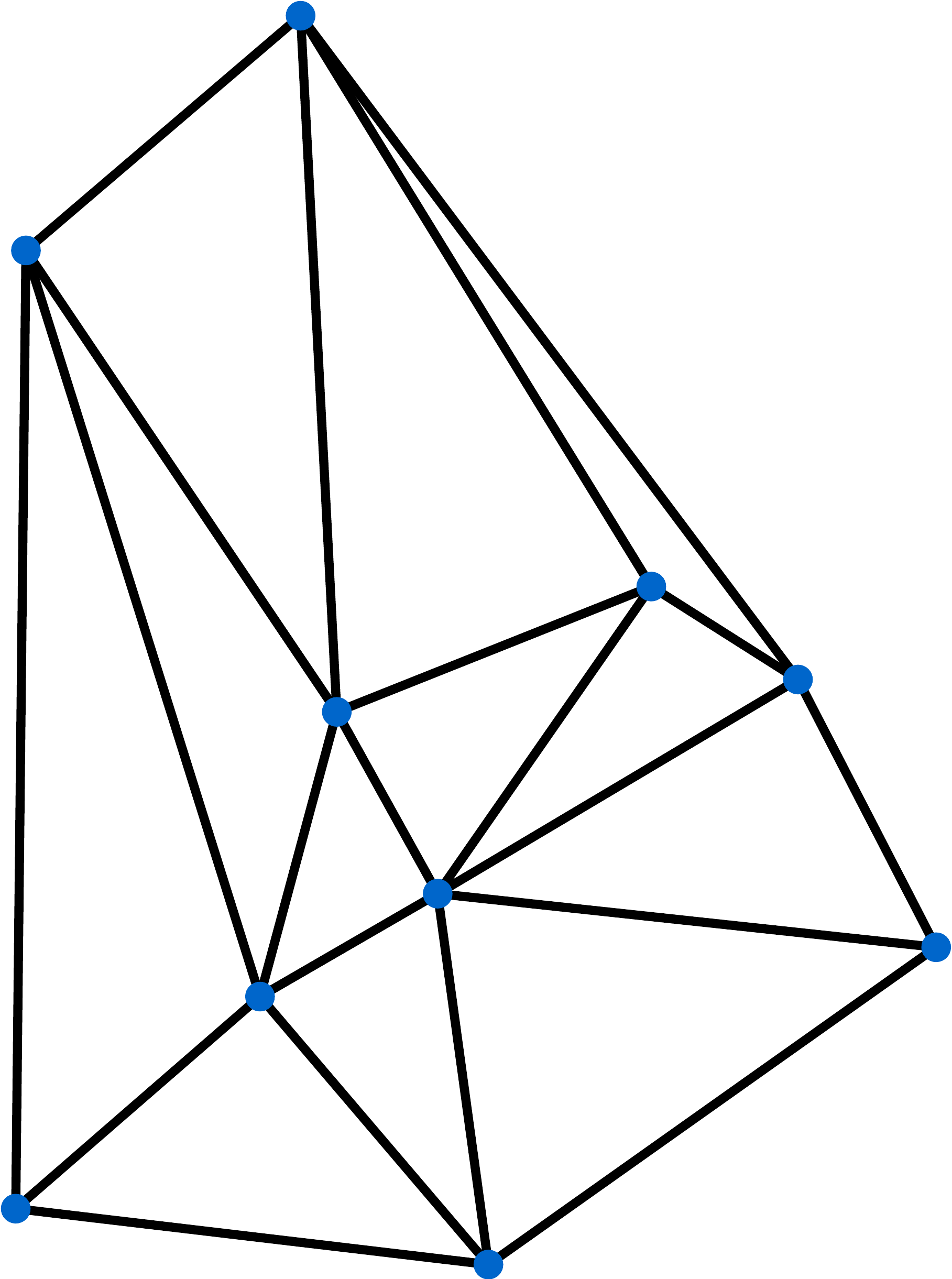} &
        \includegraphics[page=2, width=1\linewidth]{figures/maxLength3max.pdf} &
        \includegraphics[page=3, width=1\linewidth]{figures/maxLength3max.pdf} &
        \includegraphics[page=4, width=1\linewidth]{figures/maxLength3max.pdf} &
        \includegraphics[page=5, width=1\linewidth]{figures/maxLength3max.pdf} &
        \includegraphics[page=6, width=1\linewidth]{figures/maxLength3max.pdf} &
        \includegraphics[page=7, width=1\linewidth]{figures/maxLength3max.pdf} &
        \includegraphics[page=8, width=1\linewidth]{figures/maxLength3max.pdf}\\\midrule
        \includegraphics[page=1, width=1\linewidth]{figures/maxLength8max.pdf} &
        \includegraphics[page=2, width=1\linewidth]{figures/maxLength8max.pdf} &
        \includegraphics[page=3, width=1\linewidth]{figures/maxLength8max.pdf} &
        \includegraphics[page=4, width=1\linewidth]{figures/maxLength8max.pdf} &
        \includegraphics[page=5, width=1\linewidth]{figures/maxLength8max.pdf} &
        \includegraphics[page=6, width=1\linewidth]{figures/maxLength8max.pdf} &
        \includegraphics[page=7, width=1\linewidth]{figures/maxLength8max.pdf} &
        \includegraphics[page=8, width=1\linewidth]{figures/maxLength8max.pdf}\\\midrule
        \includegraphics[page=1, width=1\linewidth]{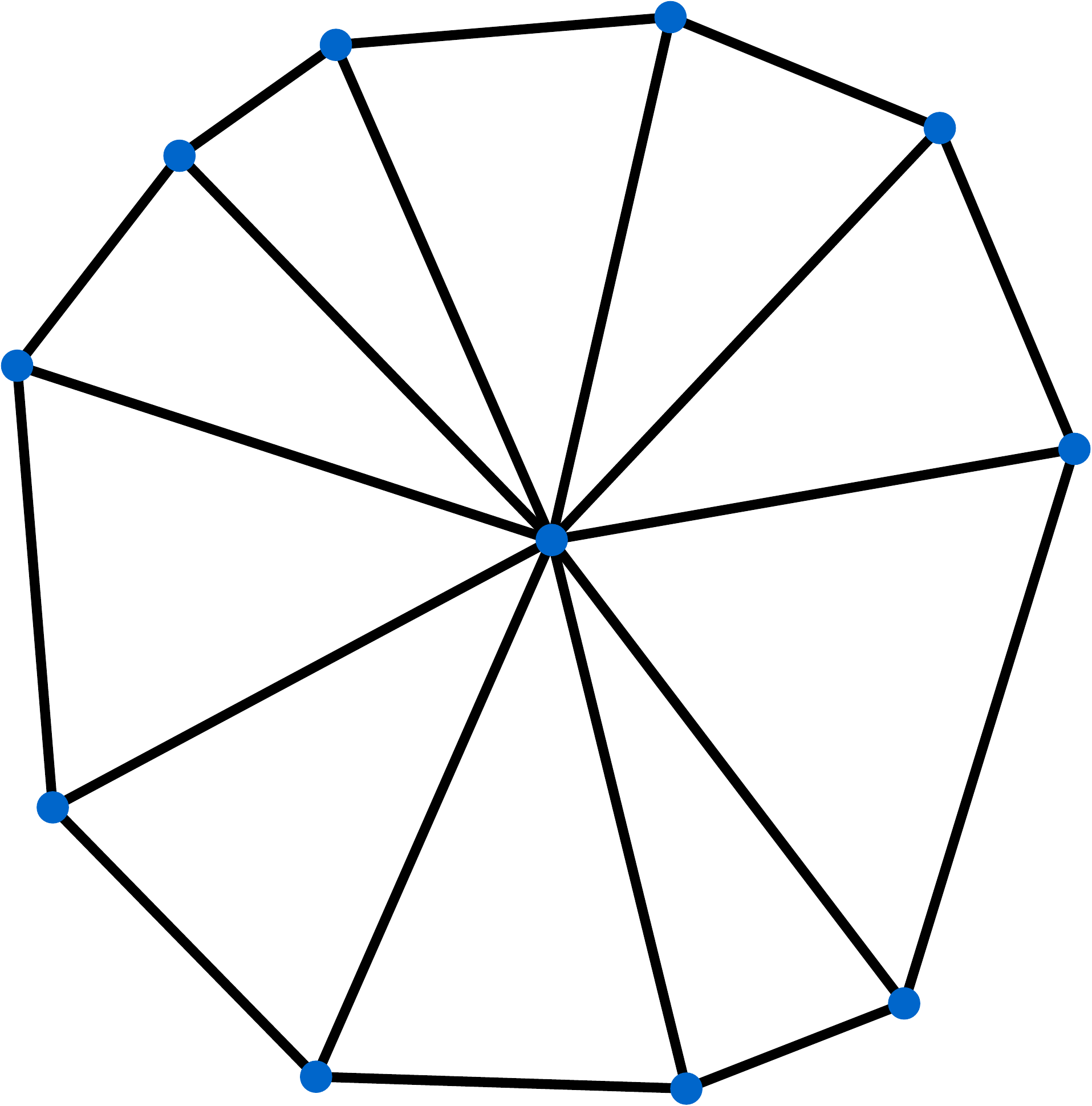} &
        \includegraphics[page=2, width=1\linewidth]{figures/maxLength9max.pdf} &
        \includegraphics[page=3, width=1\linewidth]{figures/maxLength9max.pdf} &
        \includegraphics[page=4, width=1\linewidth]{figures/maxLength9max.pdf} &
        \includegraphics[page=5, width=1\linewidth]{figures/maxLength9max.pdf} &
        \includegraphics[page=6, width=1\linewidth]{figures/maxLength9max.pdf} &
        \includegraphics[page=7, width=1\linewidth]{figures/maxLength9max.pdf} &
        \includegraphics[page=8, width=1\linewidth]{figures/maxLength9max.pdf}\\
        \bottomrule
\end{tabu}
\vspace{6cm}~
\end{table*}

\begin{table*}[t]
\caption{Optimizing with maximum degree 5 (Sum aggregation)}
\label{tab:maxdegree_sum}
\small
\begin{tabu} to \textwidth {X[1,c,m]|X[1,c,m]|X[1,c,m]|X[1,c,m]|X[1,c,m]|X[1,c,m]|X[1,c,m]|X[1,c,m]}
        \toprule
        DT & Opposing Angles & Dual Edge Ratio & Dual Area overlap & Lens & Shrunk Circle & Triangular Lens & Shrunk Circumcircle\\
        \midrule
        \includegraphics[page=1, width=1\linewidth]{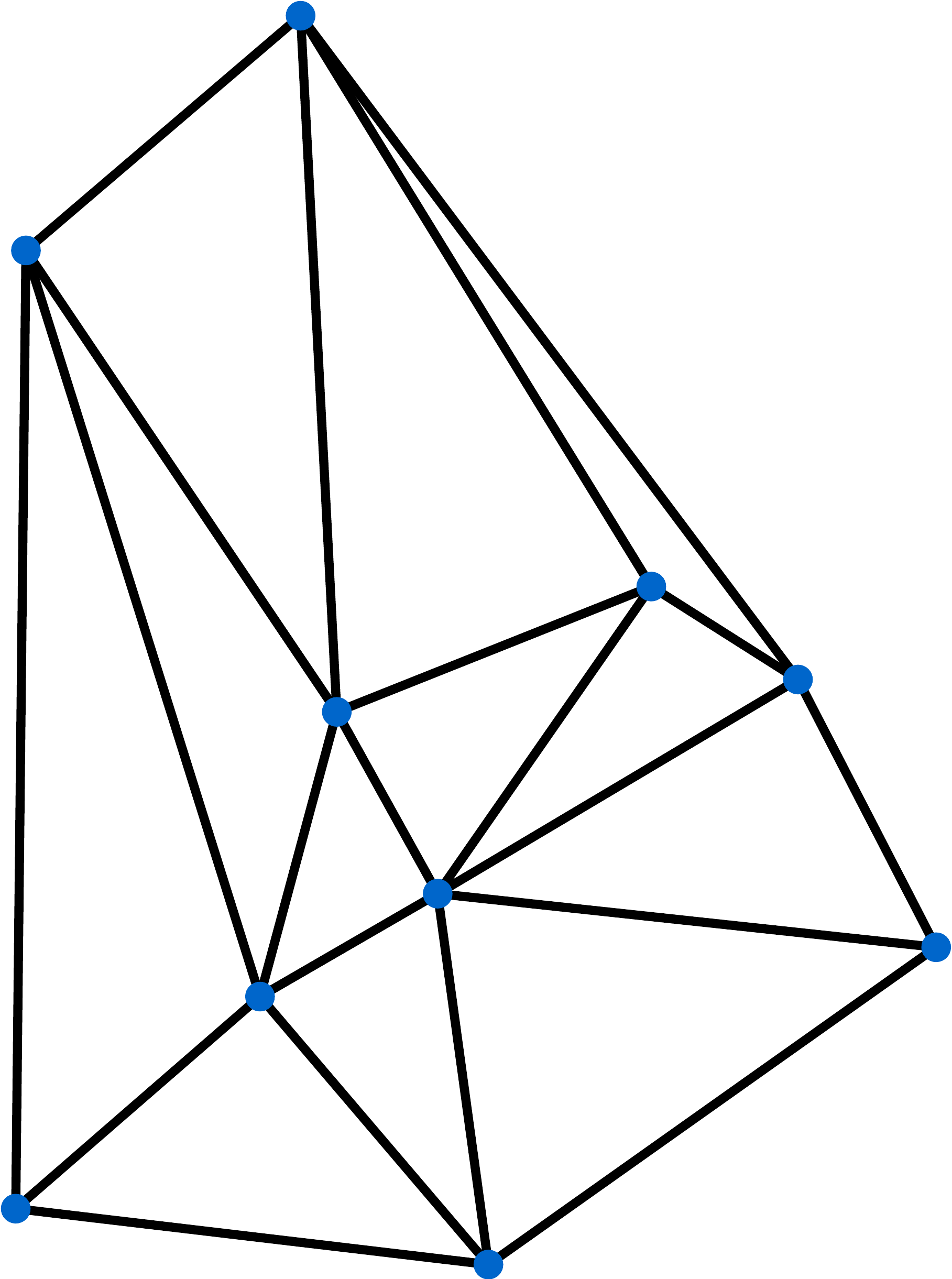} &
        \includegraphics[page=2, width=1\linewidth]{figures/maxDegree3sum.pdf} &
        \includegraphics[page=3, width=1\linewidth]{figures/maxDegree3sum.pdf} &
        \includegraphics[page=4, width=1\linewidth]{figures/maxDegree3sum.pdf} &
        \includegraphics[page=5, width=1\linewidth]{figures/maxDegree3sum.pdf} &
        \includegraphics[page=6, width=1\linewidth]{figures/maxDegree3sum.pdf} &
        \includegraphics[page=7, width=1\linewidth]{figures/maxDegree3sum.pdf} &
        \includegraphics[page=8, width=1\linewidth]{figures/maxDegree3sum.pdf}\\\midrule
        \includegraphics[page=1, width=1\linewidth]{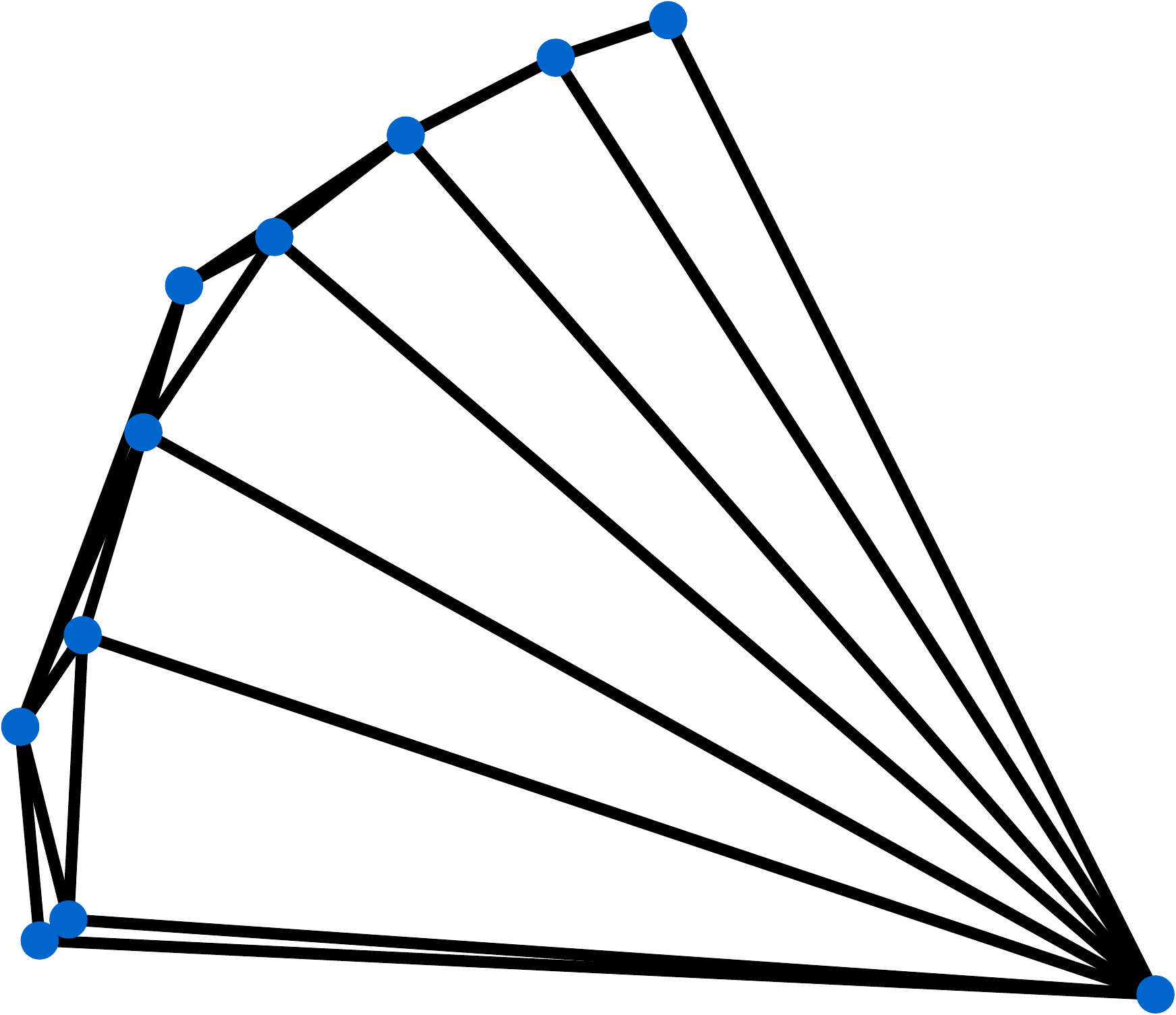} &
        \includegraphics[page=2, width=1\linewidth]{figures/maxDegree8sum.pdf} &
        \includegraphics[page=3, width=1\linewidth]{figures/maxDegree8sum.pdf} &
        \includegraphics[page=4, width=1\linewidth]{figures/maxDegree8sum.pdf} &
        \includegraphics[page=5, width=1\linewidth]{figures/maxDegree8sum.pdf} &
        \includegraphics[page=6, width=1\linewidth]{figures/maxDegree8sum.pdf} &
        \includegraphics[page=7, width=1\linewidth]{figures/maxDegree8sum.pdf} &
        \includegraphics[page=8, width=1\linewidth]{figures/maxDegree8sum.pdf}\\\midrule
        \includegraphics[page=1, width=1\linewidth]{figures/maxDegree9sum.pdf} &
        \includegraphics[page=2, width=1\linewidth]{figures/maxDegree9sum.pdf} &
        \includegraphics[page=3, width=1\linewidth]{figures/maxDegree9sum.pdf} &
        \includegraphics[page=4, width=1\linewidth]{figures/maxDegree9sum.pdf} &
        \includegraphics[page=5, width=1\linewidth]{figures/maxDegree9sum.pdf} &
        \includegraphics[page=6, width=1\linewidth]{figures/maxDegree9sum.pdf} &
        \includegraphics[page=7, width=1\linewidth]{figures/maxDegree9sum.pdf} &
        \includegraphics[page=8, width=1\linewidth]{figures/maxDegree9sum.pdf}\\
        \bottomrule
\end{tabu}
\vspace{\baselineskip}
\end{table*}

\begin{table*}[t]
\caption{Optimizing with maximum degree 5 (Bottleneck aggregation)}
\label{tab:maxdegree_bottleneck}
\small
\begin{tabu} to \textwidth {X[1,c,m]|X[1,c,m]|X[1,c,m]|X[1,c,m]|X[1,c,m]|X[1,c,m]|X[1,c,m]|X[1,c,m]}
        \toprule
        DT & Opposing Angles & Dual Edge Ratio & Dual Area overlap & Lens & Shrunk Circle & Triangular Lens & Shrunk Circumcircle\\
        \midrule
        \includegraphics[page=1, width=1\linewidth]{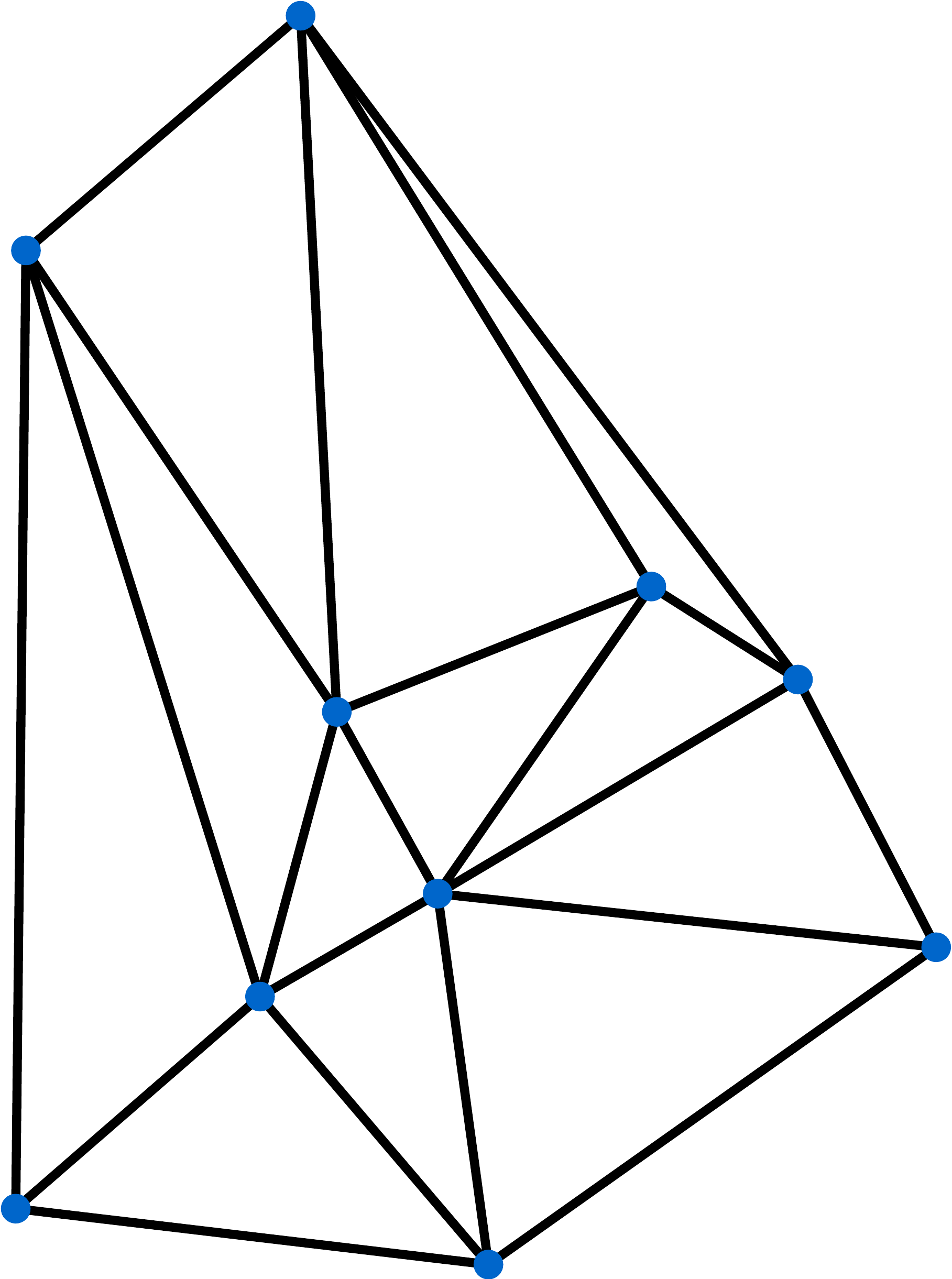} &
        \includegraphics[page=2, width=1\linewidth]{figures/maxDegree3max.pdf} &
        \includegraphics[page=3, width=1\linewidth]{figures/maxDegree3max.pdf} &
        \includegraphics[page=4, width=1\linewidth]{figures/maxDegree3max.pdf} &
        \includegraphics[page=5, width=1\linewidth]{figures/maxDegree3max.pdf} &
        \includegraphics[page=6, width=1\linewidth]{figures/maxDegree3max.pdf} &
        \includegraphics[page=7, width=1\linewidth]{figures/maxDegree3max.pdf} &
        \includegraphics[page=8, width=1\linewidth]{figures/maxDegree3max.pdf}\\\midrule
        \includegraphics[page=1, width=1\linewidth]{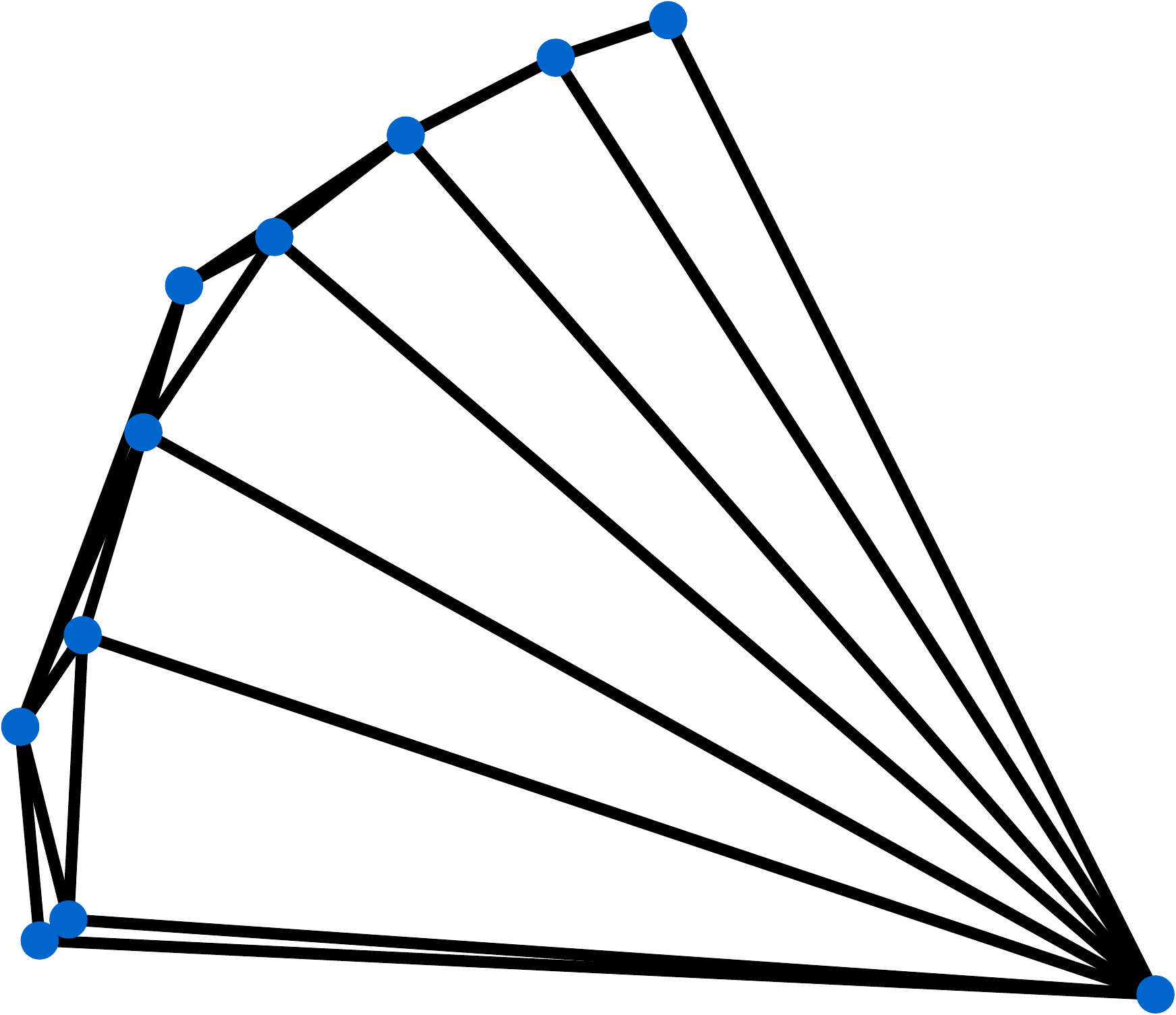} &
        \includegraphics[page=2, width=1\linewidth]{figures/maxDegree8max.pdf} &
        \includegraphics[page=3, width=1\linewidth]{figures/maxDegree8max.pdf} &
        \includegraphics[page=4, width=1\linewidth]{figures/maxDegree8max.pdf} &
        \includegraphics[page=5, width=1\linewidth]{figures/maxDegree8max.pdf} &
        \includegraphics[page=6, width=1\linewidth]{figures/maxDegree8max.pdf} &
        \includegraphics[page=7, width=1\linewidth]{figures/maxDegree8max.pdf} &
        \includegraphics[page=8, width=1\linewidth]{figures/maxDegree8max.pdf}\\\midrule
        \includegraphics[page=1, width=1\linewidth]{figures/maxDegree9max.pdf} &
        \includegraphics[page=2, width=1\linewidth]{figures/maxDegree9max.pdf} &
        \includegraphics[page=3, width=1\linewidth]{figures/maxDegree9max.pdf} &
        \includegraphics[page=4, width=1\linewidth]{figures/maxDegree9max.pdf} &
        \includegraphics[page=5, width=1\linewidth]{figures/maxDegree9max.pdf} &
        \includegraphics[page=6, width=1\linewidth]{figures/maxDegree9max.pdf} &
        \includegraphics[page=7, width=1\linewidth]{figures/maxDegree9max.pdf} &
        \includegraphics[page=8, width=1\linewidth]{figures/maxDegree9max.pdf}\\
        \bottomrule
\end{tabu}
\vspace{6cm}~
\end{table*}


\clearpage

\small
\begin{thebibliography}{10}

\bibitem{aichholzer2004lower}
O.~Aichholzer, F.~Hurtado, and M.~Noy.
\newblock A lower bound on the number of triangulations of planar point sets.
\newblock {\em Computational Geometry}, 29(2):135--145, 2004.

\bibitem{bern1993edge}
M.~Bern, H.~Edelsbrunner, D.~Eppstein, S.~Mitchell, and T.~S. Tan.
\newblock Edge insertion for optimal triangulations.
\newblock {\em Discrete \& Computational Geometry}, 10(1):47--65, 1993.

\bibitem{PaulChew1989}
L.~P. Chew.
\newblock Constrained {D}elaunay triangulations.
\newblock {\em Algorithmica}, 4(1-4):97--108, 1989.

\bibitem{dobkin1990delaunay}
D.~P. Dobkin, S.~J. Friedman, and K.~J. Supowit.
\newblock Delaunay graphs are almost as good as complete graphs.
\newblock {\em Discrete \& Computational Geometry}, 5(4):399--407, 1990.

\bibitem{d1989optimal}
E.~F. D’Azevedo and R.~B. Simpson.
\newblock On optimal interpolation triangle incidences.
\newblock {\em SIAM Journal on Scientific and Statistical Computing},
  10(6):1063--1075, 1989.

\bibitem{giannopoulos2010computing}
P.~Giannopoulos, R.~Klein, C.~Knauer, M.~Kutz, and D.~Marx.
\newblock Computing geometric minimum-dilation graphs is NP-hard.
\newblock {\em International Journal of Computational Geometry \&
  Applications}, 20(02):147--173, 2010.

\bibitem{gudmundsson2002higher}
J.~Gudmundsson, M.~Hammar, and M.~van Kreveld.
\newblock Higher order {D}elaunay triangulations.
\newblock {\em Computational Geometry}, 23(1):85--98, 2002.

\bibitem{dilation-survey}
J.~Gudmundsson and C.~Knauer.
\newblock Dilation and detours in geometric networks.
\newblock In T.~F. Gonzalez, editor, {\em Handbook of Approximation Algorithms
  and Metaheuristics}. Chapman and Hall/CRC, 2007.

\bibitem{jansen1993one}
K.~Jansen.
\newblock One strike against the min-max degree triangulation problem.
\newblock {\em Computational Geometry}, 3(2):107--120, 1993.

\bibitem{kant1997triangulating}
G.~Kant and H.~L. Bodlaender.
\newblock Triangulating planar graphs while minimizing the maximum degree.
\newblock {\em Information and Computation}, 135(1):1--14, 1997.

\bibitem{keil1992classes}
J.~M. Keil and C.~A. Gutwin.
\newblock Classes of graphs which approximate the complete {E}uclidean graph.
\newblock {\em Discrete \& Computational Geometry}, 7(1):13--28, 1992.

\bibitem{lawson1977software}
C.~L. Lawson.
\newblock Software for {C1} surface interpolation.
\newblock In {\em Mathematical software}, pages 161--194. Elsevier, 1977.

\bibitem{lee1986generalized}
D.-T. Lee and A.~K. Lin.
\newblock Generalized delaunay triangulation for planar graphs.
\newblock {\em Discrete \& Computational Geometry}, 1(3):201--217, 1986.

\bibitem{mullen2011hot}
P.~Mullen, P.~Memari, F.~de~Goes, and M.~Desbrun.
\newblock {HOT}: Hodge-optimized triangulations.
\newblock In {\em Proc. ACM SIGGRAPH 2011}, Article 103, pages 1--12. 2011.

\bibitem{mulzer2008minimum}
W.~Mulzer and G.~Rote.
\newblock Minimum-weight triangulation is {NP}-hard.
\newblock {\em Journal of the ACM (JACM)}, 55(2):1--29, 2008.

\bibitem{openproblemscccg17}
J.~O'Rourke.
\newblock Open problems from {CCCG} 2017, 2018.

\bibitem{pilz2014flip}
A.~Pilz.
\newblock Flip distance between triangulations of a planar point set is
  {APX}-hard.
\newblock {\em Computational Geometry}, 47(5):589--604, 2014.

\bibitem{rajan1994optimality}
V.~T. Rajan.
\newblock Optimality of the {D}elaunay triangulation in $\mathbb{R}^d$.
\newblock {\em Discrete \& Computational Geometry}, 12(2):189--202, 1994.

\bibitem{sharir2011counting}
M.~Sharir and A.~Sheffer.
\newblock Counting triangulations of planar point sets.
\newblock {\em The Electronic Journal of Combinatorics}, 18(P70):1, 2011.

\bibitem{sibson1978locally}
R.~Sibson.
\newblock Locally equiangular triangulations.
\newblock {\em The Computer Journal}, 21(3):243--245, 1978.

\bibitem{van2010optimization}
M.~van Kreveld, M.~L{\"o}ffler, and R.~I. Silveira.
\newblock Optimization for first order {D}elaunay triangulations.
\newblock {\em Computational Geometry}, 43(4):377--394, 2010.

\end{thebibliography}
\end{document}